\documentclass[12pt,a4paper]{article}

\usepackage{amssymb,amsmath,bm,mathrsfs,makeidx,amsfonts,graphicx,amsthm,multirow}
\usepackage[authoryear,nonamebreak,bibstyle]{natbib}
\usepackage{epsfig}
\usepackage{rotating}
\usepackage{extarrows}
\usepackage{xr}
\numberwithin{equation}{section}
\makeatletter
\def\ExtendSymbol#1#2#3#4#5{\ext@arrow 0099{\arrowfill@#1#2#3}{#4}{#5}}
\def\RightExtendSymbol#1#2#3#4#5{\ext@arrow 0359{\arrowfill@#1#2#3}{#4}{#5}}
\def\LeftExtendSymbol#1#2#3#4#5{\ext@arrow 6095{\arrowfill@#1#2#3}{#4}{#5}}
\makeatother

\newtheorem{thm}{Theorem}
\newtheorem{prop}{Proposition}
\newtheorem{lem}{Lemma}

\newtheorem{rmk}{Remark}

\newtheorem{assu}{Assumption}

\def\be{\begin{equation}}
\def\ee{\end{equation}}
\def\bea{\begin{eqnarray}}
\def\eea{\end{eqnarray}}

\newcommand{\non}{\nonumber \\}

\def\beginn{\begin{eqnarray*}}
\def\endn{\end{eqnarray*}}
\def\beginy{\begin{eqnarray}}
\def\endy{\end{eqnarray}}
\def\begine{\begin{enumerate}}
\def\ende{\end{enumerate}}

\newcommand{\bbA}{{\bf A}}

\newcommand{\bba}{{\bf a}}
\newcommand{\bbB}{{\bf B}}

\newcommand{\bbe}{{\bf e}}

\newcommand{\bbf}{{\bf f}}
\newcommand{\bbF}{{\bf F}}

\newcommand{\bbK}{{\bf K}}

\newcommand{\bbS}{{\bf S}}

\newcommand{\bbu}{{\bf u}}

\newcommand{\bbv}{{\bf v}}

\newcommand{\bbX}{{\bf X}}

\newcommand{\bbx}{{\bf x}}

\usepackage[justification=centering]{caption}

{
\begin{document}{

\title{Estimation of Cross--Sectional Dependence \\ in Large Panels}
\author
{ \ {Jiti Gao}$^{1}$, \ {Guangming Pan}$^{2}$, \ {Yanrong Yang}$^{3}$ and {Bo Zhang}$^{1}$\\
\normalsize{$^{1}$Monash University, Australia}\\
\normalsize{$^{2}$Nanyang Technological University, Singapore}\\
\normalsize{$^{3}$The Australian National University, Australia}\\
}
\maketitle
\begin{abstract}
Accurate estimation for extent of cross-sectional dependence in large panel data analysis is paramount to further statistical analysis on the data under study. Grouping more data with weak relations (cross--sectional dependence) together often results in less efficient dimension reduction and worse forecasting. This paper describes cross-sectional dependence among a large number of objects (time series) via a factor model and parameterizes its extent in terms of strength of factor loadings. A new joint estimation method is proposed for the parameter representing such extent and some other parameters involved in the estimation procedure. In particular, a marginal estimator is also proposed as the related other parameters are observed. Asymptotic distributions for the joint estimator and marginal estimator are both established. Various simulation designs illustrate the effectiveness of the proposed estimation method in the finite sample performance. Applications in cross-country macro-variables and stock returns from S$\&$P 500 are studied.
\smallskip

\smallskip


\end{abstract}

\section{Introduction}
Large panel data analysis attracts ever-growing interest in the modern literature of economics and finance. Cross--sectional dependence is popular in large panel data and the relevant literature focuses on testing existence of cross-sectional dependence. A survey on description and testing of cross-sectional dependence is given in \cite{SW2012}. \cite{P2004} utilizes sample correlations to test cross-sectional dependence while \cite{BFK2012} extend the classical Lagrangian multiplier (LM) test to the large dimensional case. \cite{CGL2012} and \cite{HPP2012} consider cross-sectional dependence tests for nonlinear econometric models. As more and more cross-sections are grouped together in panel data, it is quite natural and common for cross-sectional dependence to appear. Cross-sectional independence is an extreme hypothesis. Rejecting such a hypothesis does not provide much information about the relationship between different cross-sections under study. In view of this, measuring the degree of cross-sectional dependence is more important than just testing its presence. As we know, in comparison with cross-sectional dependence tests, few studies contribute to accessing the extent of cross-sectional dependence. \cite{Ng2006} uses spacings of cross-sectional correlations to exploit the ratio of correlated subsets over all sections. \cite{BKP2016} use a factor model to describe cross-sectional dependence and develop estimators that are based on a method of moments.

In this paper, we will contribute to this area: description and measure of the extent of cross-sectional dependence for large dimensional panel data with $N$ cross-section units and $T$ time series. The first natural question is how to describe cross-sectional dependence in panel data efficiently? To address this issue, the panel data literature mainly discusses two different ways of modelling cross-sectional dependence: the spatial correlation and the factor structure approach (see, for example, \cite{SW2012}). This paper utilizes the factor model to describe cross-sectional dependence as well as capturing time serial dependence, which can benefit further statistical inference such like forecasting. Actually, the factor model is not only a powerful tool to characterize cross-sectional dependence for economic and financial data, but also efficient in dealing with statistical inference for high dimensional data from a dimension-reduction point of view. Some related studies include \cite{FFL2008}, \cite{FLM2013} and \cite{PY2008}.

While it is a rare phenomenon to have cross-sectional independence for all $N$ sections, it is also unrealistic to assume that all $N$ sections are dependent. Hence the degree or extent of cross-sectional dependence is more significant in statistical inference for panel data. \cite{BN2006} illustrate that more data usually result in worse forecasting due to heterogeneity in the presence of cross-sectional dependence. Grouping strong-correlated cross-sections together is very significant in further study. In factor model, relation among cross-sections is described by common factors and the strength of this relation is reflected via factor loading for each cross-section. Larger factor loading for one cross-section means stronger relation of this cross-section with common factors. In this paper,  we suppose that some factor loadings are bounded away from zero while others are around zero. In detail, it is assumed that only $[N^{\alpha_0}] (0\leq\alpha_0\leq 1)$ of all $N$ factor loadings are individually important. Instead of measuring the extent by $\alpha_0N$, we adopt the parametrization $[N^{\alpha_0}]$. The proportion of $[N^{\alpha_0}]$ over the total $N$ is quite small which tends to $0$ as $0<\alpha_0<1$, while $\alpha_0N$ is comparable to $N$ because of the same order. In this sense, our model covers some ``sparse" cases that only a small part of the sections are cross-sectionally dependent.

With such parametrization of extent for cross-sectional dependence, one goal is to propose an estimation method for the parameter $\alpha_0$. This paper proposes a unified estimation method which incorporates two classical types of cross-sectional dependence: static and dynamic principal components. In fact, factor model is equivalent to principal component analysis (PCA) in some sense (see \cite{FLM2013}). Static PCA provides the common factor with most variation while dynamic PCA finds the common factor with largest ``aggregated" time-serial covariances. The existing literature, including \cite{Baing2002}, \cite{FFL2008}, \cite{FLM2013}, focuses on common factors from static PCA. In high dimensional time series, researchers prefer using dynamic PCA to derive common factors that can keep time-serial dependence. This is very important in high dimensional time series forecasting, e.g. \cite{LY2012}. 

In this paper, for our panel data $x_{it}, i=1, 2, \ldots, N; t=1, 2, \ldots, T$, an estimator for $\alpha_0$ is proposed based on the criterion of covariance between $\bar x_t$ and $\bar x_{t+\tau}$ for $0\leq\alpha_0\leq 1$, where $\bar x_t=\frac{1}{N}\sum^{N}_{i=1}x_{it}$ and $\tau\geq 0$. When $\tau=0$, it reduces to the approach proposed in \cite{BKP2016}. However, the criterion with $\tau=0$ can derive consistent estimation for $\alpha_0$ under the restriction $\alpha_0>0.5$. This is due to the interruption of variance of error components. We overcome this disadvantage by benefiting from possibility of disappearance of time-serial covariance in error components in dynamic PCA. The criterion $cov(\bar{x}_t, \bar{x}_{t+\tau})$ with $\tau>0$, under the scenario of common factors from dynamic PCA, is proposed to obtain consistent estimation for all ranges of $\alpha_0$ in $[0, 1]$. Furthermore, joint estimation approach for $\alpha_0$ and another population parameter (necessary in estimation) is established. If this population parameter is observed, marginal estimation is also provided. From the aspect of theoretical contribution, asymptotic distributions of the joint and marginal estimators are both developed.

The main contribution of this paper is summarized as follows. 
\begin{enumerate}
\item We construct two estimators for $\alpha_0$ by utilizing both joint estimation and marginal estimation, respectively. If the parameter $\kappa_0$ involved in the proposed criterion is unknown, the joint estimation of $\alpha_0$ and $\kappa_0$ will be adopted. Otherwise, we use the marginal estimation for $\alpha_0$. This estimation method is unified in the sense of covering two types of common factors derived from static PCA and dynamic PCA, respectively. Moreover, it includes the approach by \cite{BKP2016} as a special case.

\item We establish new asymptotic distributions for both the joint and the marginal estimators. The asymptotic marginal distribution coincides with that for the joint estimator for the case where $\kappa_0$ is assumed to be known. An estimator for the asymptotic variance involved in asymptotic distribution of the joint estimation method is established.

\item In practice, finite sample performances of the proposed estimators for several different types of common factors are provided.

\end{enumerate}

The rest of the paper is organized as follows. The model and the main assumptions are introduced in Section 2. Section 3 proposes both joint and marginal estimators that are based on the second moment criterion. Asymptotic properties for these estimators are established in Section 4. Section 5 reports the simulation results. Section 6 provides empirical applications to both cross-country macro-variables and stock returns in S$\&$P 500 market. Conclusions are included in Section 7. Main proofs are provided in Appendix A while some lemmas are listed in Appendix B. The proofs of lemmas are given in a supplementary material. 

\section{The model}

Let $x_{it}$ be a double array of random variables indexed by $i=1,\ldots,N$ and $t=1,\ldots,T$, over space and time, respectively.
The aim of this paper is to measure the extent of the cross-sectional dependence of the data $\{x_{it}: i=1,\ldots,N\}$.
In panel data analysis, there are two common models to describe cross-sectional dependence: spatial models and factor models.
In \cite{BKP2016}, a static approximate factor model is used. We consider a factor model as follows:
\begin{eqnarray}\label{factor1}
x_{it}&=&\mu_i+\boldsymbol{\beta}_{i0}^{'}\bbf_t+\boldsymbol{\beta}_{i1}^{'}\bbf_{t-1}+\cdots+\boldsymbol{\beta}_{is}^{'}\bbf_{t-s}+u_{it}\non
&=&\mu_i+\boldsymbol{\beta}_{i}^{'}(L)\bbf_t+u_{it}, \ \ i=1,2,\ldots,N; \ t=1,2,\ldots,T,
\end{eqnarray}
where $\bbf_t$ is the $m\times 1$ vector of unobserved factors (with m being fixed),
\begin{eqnarray*}
\boldsymbol{\beta}_i(L)=\boldsymbol{\beta}_{i0}+\boldsymbol{\beta}_{i1}L+\boldsymbol{\beta}_{i2}L^2+\cdots+\boldsymbol{\beta}_{is}L^{s},
\end{eqnarray*}
in which $\boldsymbol{\beta}_{i\ell}=(\beta_{i\ell 1}, \beta_{i\ell 2}, \ldots, \beta_{i\ell m})^{'}$, $\ell=0,1,\ldots,s$ are the associated vectors
of unobserved factor loadings and $L$ is the lag operator, here $s$ is assumed to be fixed, and $\mu_i, i=1,2,\ldots,N$ are constants that represent the mean values for all sections, and $\{u_{it}: i=1,\ldots,N; t=1,\ldots,T\}$ are idiosyncratic components.

Clearly, we can write (\ref{factor1}) in the static form:
\begin{eqnarray}\label{model1}
x_{it}=\mu_i+\boldsymbol{\beta}_i^{'}\bbF_t+u_{it},\ i=1,2,\ldots,N; t=1,2,\ldots,T,
\end{eqnarray}
where
\begin{eqnarray*}
\boldsymbol{\beta}_i=\left(
                       \begin{array}{c}
                         \boldsymbol{\beta}_{i0} \\
                         \boldsymbol{\beta}_{i1} \\
                         \vdots \\
                         \boldsymbol{\beta}_{is} \\
                       \end{array}
                     \right)_{m(s+1)} \ \ \ and \ \ \ \bbF_t=\left(
                                                      \begin{array}{c}
                                                        \bbf_t \\
                                                        \bbf_{t-1} \\
                                                        \vdots \\
                                                        \bbf_{t-s} \\
                                                      \end{array}
                                                    \right)_{m(s+1)}.
\end{eqnarray*}
This model has been studied in \cite{SW2002} and \cite{F2009}.

The dimension of $\bbf_t$ is called the number of dynamic factors and is denoted by $m$. Then the dimension of $\bbF_t$ is equal to $r=m(s+1)$. In factor analysis, $\boldsymbol{\beta}_i^{'}\bbF_t$ is called the common components of $x_{it}$.

We first introduce the following assumptions.

\begin{assu}\label{A1}
The idiosyncratic component $\{\bbu_t=(u_{1t}, u_{2t}, \ldots, u_{Nt})^{'}: t=1,2,\ldots,T\}$ follows a linear stationary process of the form:
\begin{eqnarray}
u_{it}=\sum^{+\infty}_{j=0}\phi_{ij}\Big(\sum^{+\infty}_{s=-\infty}\xi_{js}\nu_{j,t-s}\Big),
\end{eqnarray}
where $\{\nu_{is}: i=\ldots,-1,0,1,\ldots; s=0,1,\ldots\}$ is a double sequence of i.i.d. random variables with mean zero and unit variance, and
\begin{eqnarray}\label{u15}
\sup_{0<j<+\infty}\sum^{N}_{i=1}|\phi_{ij}|<+\infty \ \ \mbox{and} \ \ \sup_{0<j<+\infty}\sum^{+\infty}_{s=-\infty}|\xi_{js}|\leq+\infty.
\end{eqnarray}
Moreover,
\begin{eqnarray}\label{0503}
E(u_{it}u_{j,t+\tau})=\gamma_1(\tau)\gamma_2(|i-j|),
\end{eqnarray}
where $\gamma(\tau)$ is defined such that $\gamma_1(\tau) \gamma_2(0) = E\left[u_{it} u_{i, t+\tau}\right]$ and $\gamma_2(|i-j|)$ satisfies
\begin{eqnarray}\label{u14}
\sum^{N}_{i,j=1}\gamma_2(|i-j|)=O(N).
\end{eqnarray}
\end{assu}
\begin{rmk}\label{kefen}
Condition (\ref{u14}) borrows the idea of Assumption C of \cite{Baing2002} to describe weak cross-sectional dependence in error components.

\end{rmk}

\begin{assu}\label{A3}
For $\ell=0,1,2,\ldots,s$ and $k=1,2,\ldots,m$,
\begin{eqnarray}
&&\beta_{i\ell k}=v_{i\ell k}, \ i=1,2,\ldots,[N^{\alpha_{\ell k}}] \ \ \mbox{and} \ \ \sum^{N}_{i=[N^{\alpha_{\ell k}}]+1}\beta_{i\ell k}=O(1),
\end{eqnarray}
where $[N^{\alpha_{\ell k}}]\leq N^{\alpha_{\ell k}}$ is the largest integer part of $N^{\alpha_{\ell k}}$, $0<\alpha_{\ell k}\leq 1$ and $v_{i\ell k}\sim i.i.d.(\mu_{v},\sigma_{v}^2)$ has finite sixth moments, with $\mu_{v}\neq 0$ and $\sigma_{v}^2>0$. Moreover, $\{v_{i\ell k}: i=1,2,\ldots, N; \ell=0,1,\ldots,s; k=1,2,\ldots,m\}$ are assumed to be independent of the factors $\{\bbf_t: t=1,2,\ldots,T\}$ and the idiosyncratic components $\{u_{it}: i=1,2,\ldots,N; t=1,2,\ldots,T\}$.
\end{assu}

\begin{assu}\label{A2}
The factors $\{\bbf_t, t=1,2,\ldots,T\}$ are covariance stationary with the following representation:
\begin{eqnarray}
\bbf_t=\sum^{\infty}_{j=0}b_j\boldsymbol{\zeta}_{t-j},
\end{eqnarray}
where $\{\boldsymbol{\zeta}_t, t=\ldots,-1,0,1,\ldots\}$ is an i.i.d sequence of $m$-dimensional random vectors whose components are i.i.d with zero mean and unit variance, the fourth moments of $\{\boldsymbol{\zeta}_t, -\infty<t<\infty\}$ are finite, and the coefficients $\{b_j: j=0,1,2,\ldots\}$ satisfy $\sum^{\infty}_{j=0}|b_j|<\infty$.
Furthermore, the unobserved factors $\{\bbf_t: t=1,2,\ldots, T\}$ are independent of the idiosyncratic components $\{\bbu_{t}: t=1,2,\ldots,T\}$.

\end{assu}

Now we provide some justification for these two assumptions.
\begin{enumerate}
\item Justification of Assumption 1:
The weak stationarity assumption on the idiosyncratic components $\{\bbu_{t}: t=1,2,\ldots,T\}$ is commonly used in time series analysis. Rather than imposing an independence assumption, weak cross--sectional correlation and serial correlation are imposed via $\gamma_2(|i-j|)$ and $\gamma_1(\tau)$, respectively. The levels of weakness are described by (2.6). Note that when $\{u_{it}\}$ is independent across $(i,t)$, we have $\gamma_1(\tau)=0$ and $\gamma_2(|i-j|)=0$ which satisfy Conditions (2.6).

\item Justification of Assumption 2: The degree of cross-sectional dependence in $\{\bbx_{t}: t=1,2,\ldots,N\}$ crucially depends on the nature of the factor loadings. This assumption groups the factor loadings into two categories: a strong category with effects that are bounded away from zero, and a weak category with transitory effects that tend to zero. From this point, the first $[N^{\alpha_0}]$ sections are dependent while the rest are independent. Here $\alpha_0=\max(\alpha_{\ell k}: \ell=0,1,2,\ldots,s; k=1,2,\ldots,m)$.

To simplify the proof of Theorem 2 to be established below, we require the factor loadings to have the finite sixth moments. However, we believe that the finite second moment condition might just be sufficient by performing the truncation technique in the proof of Lemma 3 below.

\end{enumerate}

\section{The estimation method}

The aim of this paper is to estimate the exponent $\alpha_0=\max_{\ell,k}(\alpha_{\ell k})$, which describes the extent of cross-sectional dependence. As in \cite{BKP2016} , we consider the cross-sectional average $\bar x_t=1/N\sum^{N}_{i=1}x_{it}$ and then derive an estimator for $\alpha_0$ from the information of $\{\bar x_t: t=1,2,\ldots,T\}$. \cite{BKP2016} use the variance of the cross-sectional average $\bar x_t$ to estimate $\alpha_0$ and carry out statistical inference for an estimator of $\alpha_0$. Specifically they show that
\begin{eqnarray}\label{premethod}
Var(\bar x_t)=\widetilde{\kappa}_0[N^{2\alpha_0-2}]+N^{-1}c_N+O(N^{\alpha_0-2}),
\end{eqnarray}
where $\widetilde{\kappa}_0$ is a constant associated with the common components and $c_N$ is a bias constant incurred by the idiosyncratic errors.
From (\ref{premethod}), we can see that, in order to estimate $\alpha_0$, \cite{BKP2016} assume that $2\alpha_0-2>-1$, i.e. $\alpha_0>1/2$. Otherwise, the second term will have a higher order than the first term. So the approach by \cite{BKP2016} fails in the case of $0<\alpha_0<1/2$.

This paper is to propose a new estimator that is applicable to the full range of $\alpha_0$, i.e., $0\leq\alpha_0\leq1$. Based on the assumption that the common factors possess serial dependence that is stronger than that of the idiosyncratic components, we construct a so--called covariance criterion $Cov(\bar x_t, \bar x_{t+\tau})$, whose leading term does not include the idiosyncratic components for $0\leq\alpha_0\leq1$. In other words, the advantage of this covariance criterion over the variance criterion $Var(\bar x_t)$ lies on the fact that there is no interruption brought by the idiosyncratic components
$\{u_{it}: i=1,2,\ldots,N; t=1,2,\ldots,T\}$ in $Cov(\bar x_t, \bar x_{t+\tau})$.

We define \begin{eqnarray}\label{160530}
\kappa_{\tau}=\mu_v^2\sum^{s}_{\ell_1,\ell_2=0}\sum^{m}_{k=1}E(f_{k,t-\ell_1}f_{k,t+\tau-\ell_2}),
\end{eqnarray}
in which $\mu_v = E[v_{i \ell k}]$, and $s$ and $m$ are the same as in (\ref{factor1}). Here $\kappa_{\tau}$ comes from the leading term of $Cov(\bar x_t, \bar x_{t+\tau})$.

Next, we illustrate how the covariance $Cov(\bar x_t, \bar x_{t+\tau})$ implies the extent parameter $\alpha_0$ in detail. Let $[N^{a}]$          ($a\geq 0$) denote the largest integer part not greater than $N^a$. For simplicity, let $[N^{b}]$ ($b\leq0$) denote $\frac{1}{[N^{-b}]}$. Moreover, to simplify the notation, throughout the paper we also use the following notation:
\begin{eqnarray}
[N^{ka}]:=[N^{a}]^k, \ \ [N^{a-k}]:=\frac{[N^{a}]}{N^{k}}, \ \ \forall a, k\in\mathbb{R}.
\end{eqnarray}
But we would like to remind the reader that $[N^{ka}]$ is actually not equal to $[N^{a}]^k$. Next, we will propose an estimator for $\alpha_0$ under two different scenarios: the joint estimator $(\widetilde{\alpha}_{\tau}, \widetilde{\kappa}_{\tau})$ under the case of some other parameters being unknown while the marginal estimator $\widehat\alpha_{\tau}$ for the case of some other parameters being known.

\subsection{The marginal estimator $\widehat\alpha_{\tau}$ when $\kappa_{\tau}$ is known}

At first we consider the marginal estimator $\widehat\alpha_{\tau}$ to deal with the case where $\kappa_{\tau}$ is known. The parameter $\kappa_{\tau}$ describes the temporal dependence in the common factors. If we know this information in advance, the estimation of the extent of cross-sectional dependence becomes easy. We propose the following marginal estimation method.

Without loss of generality, we assume that $\alpha_{\ell k}=\alpha_0, \forall \ell=0,1,2,\ldots,s; k=1,2,\ldots,m$. Let Assumption 2 hold. Let $\bar x_{nt}$ be the cross--sectional average of $x_{it}$ over $i=1,2,\ldots,n$ with $n\leq N$. Similarly, $\bar\beta_{n\ell k}:=\frac{1}{n}\sum^{n}_{i=1}\beta_{i\ell k}$. Then
\begin{eqnarray*}
E(\bar\beta_{n\ell k})=\Big\{\begin{array}{cc}
                           \mu_v, & n\leq [N^{\alpha_0}] \\
                           \frac{[N^{\alpha_0}]}{n}\mu_v+\frac{K_{n\ell k}}{n}, & n>[N^{\alpha_0}],
                         \end{array}
\end{eqnarray*}
and
\begin{eqnarray*}
Var(\bar\beta_{n\ell k})=\Big\{\begin{array}{cc}
                             \frac{\sigma^2_v}{n}, & n\leq [N^{\alpha_0}] \\
                             \frac{[N^{\alpha_0}]}{n^2}\sigma^2_v, & n>[N^{\alpha_0}],
                           \end{array}
\end{eqnarray*}
where $K_{n\ell k}=\sum^{n}_{i=[N^{\alpha_0}]+1}\beta_{i\ell k}$.

It follows that
\begin{eqnarray}\label{f1}
Cov(\bar x_{nt}, \bar x_{n,t+\tau})&=&\sum^{s}_{\ell=0}\sum^{m}_{k=1}\Big((E[\bar\beta_{n\ell k}])^2+Var(\bar\beta_{n\ell k})\Big)E(f_{k,t-\ell}f_{k,t+\tau-\ell})\non
&&+\sum^{s}_{\ell_1\neq\ell_2}\sum^{m}_{k=1}E(\bar\beta_{n\ell_1k})E(\bar\beta_{n\ell_2k})E(f_{k,t-\ell_1}f_{k,t+\tau-\ell_2})+E(\bar u_{nt}\bar u_{n,t+\tau})\non
&=&\Big\{\begin{array}{cc}
           \kappa_{\tau}+O(n^{-1}), & n\leq [N^{\alpha_0}] \\
           \kappa_{\tau}\frac{[N^{2\alpha_0}]}{n^2}+O(\frac{[N^{\alpha_0}]}{n^2}), & n>[N^{\alpha_0}].
         \end{array}
\end{eqnarray}

A simple calculation indicates that
\begin{eqnarray*}
\ln\Big({\rm Cov}(\bar x_t, \bar x_{t+\tau})\Big)^2=\ln(\kappa_{\tau}^2)+(4\alpha_0-4)\ln(N)+O\left(\frac{1}{[N^{\alpha_0}]\kappa_{\tau}}\right),
\end{eqnarray*}
which implies
\begin{eqnarray}\label{a5}
\alpha_0=\frac{\ln\Big({\rm Cov}(\bar x_t, \bar x_{t+\tau})\Big)^2-\ln(\kappa_0^2)}{4\ln(N)}+1+O\left(\frac{1}{[N^{\alpha_0}]\kappa_{\tau}\ln(N)}\right),
\end{eqnarray}
where $\kappa_{\tau}$ is defined in (\ref{160530}).

Hence, for $0\leq\alpha_0\leq 1$, $\alpha_0$ can be estimated from (\ref{a5}) using a consistent estimator for ${\rm Cov}(\bar x_t, \bar x_{t+\tau})$ given by
\begin{eqnarray}\label{a7}
\widehat\sigma_N(\tau)=\frac{1}{T-\tau}\sum^{T-\tau}_{t=1}(\bar x_t-\bar x^{(1)})(\bar x_{t+\tau}-\bar x^{(2)}),
\end{eqnarray}
where $\bar x^{(1)}=\frac{1}{T-\tau}\sum^{T-\tau}_{t=1}\bar x_t$ and $\bar x^{(2)}=\frac{1}{T-\tau}\sum^{T-\tau}_{t=1}\bar x_{t+\tau}$ and $\bar{x}_t = \frac{1}{N} \sum_{i=1}^N x_{it}$. Thus, a consistent estimator for $\alpha_0$ is given by
\begin{eqnarray}\label{a6}
\widehat\alpha_{\tau}=\frac{\log\Big(\widehat\sigma_N(\tau)\Big)^2-\ln(\kappa_{\tau})^2}{4\ln(N)}+1.
\end{eqnarray}

\subsection{The joint estimator $(\widetilde{\alpha}_{\tau}, \widetilde{\kappa}_{\tau})$}
Now we consider the case where $\kappa_{\tau}$ is unknown. Recalling (\ref{f1}), we minimize the following quadratic form in terms of $\alpha$ and $\kappa$:
\be
Q_{NT}^{(1)}(\alpha,\kappa,\tau)=\sum^{[N^{\alpha}]}_{n=1}n^3\Big(\widehat\sigma_n(\tau)-\kappa\Big)^2+\sum^{N}_{n=[N^{\alpha}]+1}n^3\Big(\widehat\sigma_n(\tau)-\frac{[N^{2\alpha}]}{n^2}\kappa\Big)^2,
\label{jiti1}
\ee
where $\widehat\sigma_n(\tau)$ is a consistent estimator for $Cov(\bar x_{nt}, \bar x_{n,t+\tau})$ of the form:
\begin{eqnarray*}
\widehat\sigma_n(\tau)=\frac{1}{T-\tau}\sum^{T-\tau}_{t=1}\big(\bar x_{nt}-\bar x_{n}^{(1)}\big)\big(\bar x_{n,t+\tau}-\bar x_{n}^{(2)}\big),
\end{eqnarray*}
with $\bar x^{(1)}_n=\frac{1}{T-\tau}\sum^{T-\tau}_{t=1}\bar x_{nt}$ and $\bar x^{(2)}_n=\frac{1}{T-\tau}\sum^{T-\tau}_{t=1}\bar x_{n,t+\tau}$.

The joint estimator $(\widetilde{\alpha}_{\tau}, \widetilde\kappa_{\tau})$ can then be obtained by
\begin{eqnarray}\label{yry01}
\widetilde{\alpha}_{\tau}=\arg\max_{\alpha}\widehat{Q}^{(1)}_{NT}(\alpha,\tau) \ \ \mbox{and} \ \ \widetilde\kappa_{\tau}=\frac{\widehat q_1^{(1)}(\widetilde{\alpha}_{\tau},\tau)+[N^{4\widetilde{\alpha}_{\tau}}]\widehat q_2^{(1)}(\widetilde{\alpha}_{\tau},\tau)}{N^{(1)}(\widetilde{\alpha}_{\tau})},
\end{eqnarray}
where
\begin{eqnarray*}
&&\widehat q_1^{(1)}(\alpha,\tau)=\sum^{[N^{\alpha}]}_{n=1}n^3\widehat\sigma_n(\tau), \ \ \widehat q_2^{(1)}(\alpha,\tau)=\sum^{N}_{n=[N^{\alpha}]+1}n\widehat\sigma_n(\tau),\non
&&N^{(1)}(\alpha)=\sum^{[N^{\alpha}]}_{n=1}n^3+\sum^{N}_{n=[N^{\alpha}]+1}\frac{[N^{4\alpha}]}{n},
\end{eqnarray*}
and
$$\widehat{Q}_{NT}^{(1)}(\alpha,\tau)=\frac{(\widehat q_1^{(1)}(\alpha,\tau)+[N^{2\alpha}]\widehat q_2^{(1)}(\alpha,\tau))^2}{N^{(1)}(\alpha)}.$$
We give the full derivation of (\ref{yry01}) in Appendix A.

This joint method estimates $\alpha_0$ and $\kappa_{\tau}$ simultaneously. The above derivations show that it is easy to derive $\widetilde{\alpha}_{\tau}$ and then $\widetilde{\kappa}_{\tau}$. Of course, we can also use some other estimation methods to estimate $\kappa_{\tau}$ and then $\alpha_0$.
Notice that we use the weight function $w(n)=n^3$ in each summation part of the objective function $Q_{NT}^{(1)}(\alpha,\kappa,\tau)$ of (\ref{jiti1}). The involvement of a weight function is due to technical necessity in deriving an asymptotic distribution for $(\widetilde{\alpha}_{\tau},\widetilde{\kappa}_{\tau})$.
\subsection{Asymptotic Properties}

In this section, we will establish asymptotic distributions for the proposed joint estimator $(\widetilde{\alpha}_{\tau}, \widetilde\kappa_{\tau})$ and the marginal estimator $\widehat\alpha_{\tau}$, respectively. We assume that $\alpha_{\ell k}=\alpha_0$, $\forall \ell=0,1,\ldots,s$
and $k=1,2,\ldots,m$ for simplicity. The notation $a\asymp b$ denotes that $a=O(b)$ and $b=O(a)$.

For any $1\leq i,j\leq m$ and $0\leq h\leq T-1$, we define
\begin{eqnarray}\label{w15}
C_{ij}(h)=\frac{1}{T-h}\sum^{T-h}_{t=1}f_{i,t}f_{j,t+h}, \ \ c_{ij}(h) \equiv \sigma_{ij}(h)=E(f_{i,t}f_{j,t+h}).
\end{eqnarray}


\medskip

The following theorem establishes an asymptotic distribution for the marginal estimator $\widehat\alpha$.
At first we define some notation. $\boldsymbol{\Sigma}_{\tau}=E(\bbF_t\bbF_{t+\tau}^{'})$ and $\boldsymbol{\mu}_v=\mu_v\bbe_{m(s+1)}$, in which $\bbe_{m(s+1)}$ is an $m(s+1)\times 1$ vector with each element being $1$, $\boldsymbol{\Sigma}_v$ is an $m(s+1)$-dimensional diagonal matrix with each of the diagonal elements being $\sigma_v^2$ and

\begin{eqnarray}\label{h1}
\boldsymbol{\Omega}=\lim_{N,T\rightarrow\infty}var(\sqrt{T-\tau}vec\big(\bbS_{\tau}-\boldsymbol{\Sigma}_{\tau})),
\end{eqnarray}
where
\begin{eqnarray*}
\bbS_{\tau}=\frac{1}{T-\tau}\sum^{T-\tau}_{t=1}(\bbF_t-\bar\bbF_T)(\bbF_{t+\tau}-\bar\bbF_{T+\tau})^{'}
\end{eqnarray*}
and `vec' means that for a matrix $\bbX=(\bbx_1,\cdots,\bbx_n): q\times n$, $vec(\bbX)$ is the $qn\times 1$ vector defined as
\begin{eqnarray}
vec(\bbX)=\left(
            \begin{array}{c}
              \bbx_1 \\
              \vdots \\
              \bbx_n \\
            \end{array}
          \right).
\end{eqnarray}
Define
\bea
&& \sigma_{\tau}^2=\lim_{N,T\rightarrow\infty}\frac{v_{NT}}{[N^{\alpha_0}]}4\boldsymbol{\mu}_v^{'}\boldsymbol{\Sigma}_{\tau}\boldsymbol{\Sigma}_v\boldsymbol{\Sigma}_{\tau}\boldsymbol{\mu}_v
\label{163031}\\
&& + \lim_{N,T\rightarrow\infty}\frac{v_{NT}}{T-\tau}(\boldsymbol{\mu}_v^{'}\otimes\boldsymbol{\mu}_v^{'})\boldsymbol{\Omega}(\boldsymbol{\mu}_v\otimes\boldsymbol{\mu}_v),
\label{y111}
\eea
where $v_{NT}=\min([N^{\alpha_0}],T-\tau)$.

We are now ready to establish the main results of this paper in the following theorems and propositions.

\begin{thm}\label{thm1}
In addition to Assumptions \ref{A1}-\ref{A2}, we assume that

(i) for some constant $\delta>0$,
\begin{eqnarray}\label{05041}
E|\zeta_{it}|^{2+2\delta}<+\infty,
\end{eqnarray}
where $\zeta_{it}$ is the $i$-th component of $\boldsymbol{\zeta}_t$ and $\{\boldsymbol{\zeta}_t: \ldots,-1,0,1,\ldots\}$ is the sequence appeared in Assumption \ref{A2}.

(ii) The lag $\tau$ satisfies
\begin{eqnarray}\label{u12}
\frac{\tau}{(T-\tau)^{\delta/(2\delta+2)}}\rightarrow 0, \ \ as \ \ T\rightarrow\infty,
\end{eqnarray}
where $\delta$ is defined in (\ref{05041}).

(iii) The covariance matrix $\Gamma$ of the random vector
\begin{eqnarray}\label{cij}
\Big(C_{ij}(h^{'}): i=1,\ldots,m; j=1,\ldots,m; h^{'}=\tau-s,\ldots,\tau+s\Big)
\end{eqnarray}
is positive definite.

(iv) As $(N,T)\rightarrow(\infty, \infty)$,
\begin{eqnarray}\label{g3}
 v_{NT}^{1/2}\kappa_{\tau} \rightarrow \infty.
\end{eqnarray}
 Then there are $E_1$ and $E_2$ such that
\be
v_{NT}^{1/2} \, \left(\frac{(N^{4(\widehat\alpha_{\tau} -\alpha_0)}-1)\kappa_{\tau}}
{2}-E_1\right)E_2\rightarrow\mathcal{N}(0,\sigma_{\tau}^2),
\label{jiti2}
\ee
where
\begin{eqnarray}\label{jiti2a}
E_1=O_P((T-\tau)^{-1/2}N^{1/2-\alpha_0})+O_P(\gamma_1(\tau)N^{1-2\alpha_0})+O_P( (T-\tau)^{-1/2}N^{1-2\alpha_0})
\end{eqnarray}

and

\begin{eqnarray}\label{jiti2b}
E_2&=&1+O_P\left(\frac{(T-\tau)^{-1/2}N^{1/2-\alpha_0}}{\kappa_{\tau}}\right)\non
&&+O_P\left(\frac{\gamma_1(\tau)N^{1-2\alpha_0}}{\kappa_{\tau}}\right)+O_P\left( \frac{(T-\tau)^{-1/2}N^{1-2\alpha_0}}{\kappa_{\tau}}\right),
\end{eqnarray}
in which  $\kappa_{\tau}$ is defined in (\ref{160530}) and  $\sigma_{\tau}^2$ is defined in (\ref{163031})

\end{thm}

Under some extra conditions, the conclusion of Theorem 1 can be simplified as given in Proposition 1 below.

\begin{prop}\label{thm1fujia}
Let the conditions of Theorem \ref{thm1} hold.

(i) If, for either $\frac{1}{2}<\alpha_0\leq 1$ or $0<\alpha_0\leq\frac{1}{2}$, as $(N,T)\rightarrow(\infty, \infty)$,
\begin{eqnarray}\label{b4wupian}
\max\Big(\gamma_1(\tau)N^{1-2\alpha_0}, (T-\tau)^{-1/2}N^{1-2\alpha_0}\Big)\rightarrow 0,
\end{eqnarray}
then $\widehat\alpha_{\tau}$ is an asymptotically unbiased estimator of $\alpha_0$.

(ii) If, moreover, as $(N,T)\rightarrow(\infty, \infty)$,
\begin{eqnarray}\label{b4}
&&v_{NT}^{1/2}\max\Big(\gamma_1(\tau)N^{1-2\alpha_0}, (T-\tau)^{-1/2}N^{1-2\alpha_0}\Big)\rightarrow 0,\ \ \ as \ \ 0<\alpha_0\leq\frac{1}{2};\non
&&v_{NT}^{1/2}\max\Big(\gamma_1(\tau)N^{1-2\alpha_0},(T-\tau)^{-1/2}N^{1/2-\alpha_0}\Big)\rightarrow 0, \ \ \ as \ \ \frac{1}{2}<\alpha_0\leq 1,
\end{eqnarray}
then
\be
v_{NT}^{1/2} \, \frac{(N^{4(\widehat\alpha_{\tau} -\alpha_0)}-1)\kappa_{\tau}}
{2}\rightarrow\mathcal{N}(0,\sigma_{\tau}^2).
\label{jiti2c}
\ee


\end{prop}
From Proposition \ref{thm1fujia}, one can see that $\widehat\alpha$ is a consistent estimator of $\alpha_0$. Moreover, by a careful inspection on (\ref{jiti2a})--(\ref{jiti2b}) in Theorem \ref{thm1} one can see that Condition (\ref{b4}) can be replaced by some weak conditions to ensure the consistency of $\widehat\alpha$ under $(N, T)\rightarrow(\infty, \infty)$.

The following theorem establishes an asymptotic distribution for the joint estimator $(\widetilde{\alpha}, \widetilde{\kappa})$.

\begin{thm}\label{yyr001}
Under the conditions of Theorem \ref{thm1}, there are $E_3$ and $E_4$ such that
\begin{eqnarray}\label{yr40}
\left(
\begin{array}{c}
\widetilde{\kappa}_{\tau}v_{NT}^{1/2}(N^{2(\widetilde{\alpha}_{\tau}-\alpha_0)}-1)-E_3\\
v_{NT}^{1/2}(\widetilde{\kappa}_{\tau}-\kappa_{\tau})-E4\\
\end{array}
\right)
\stackrel{d}{\longrightarrow} \mathcal{N}\left(\left(
\begin{array}{c}
0\\
0\\
\end{array}
\right), \left(
\begin{array}{cc}
4\sigma_{\tau}^2 & -2\sigma_{\tau}^2\\
-2\sigma_{\tau}^2 & \sigma_{\tau}^2
\end{array}
\right)\right),
\end{eqnarray}

\begin{eqnarray}\label{yr40zhang3}
E_3&=&v_{NT}^{1/2} \cdot O_P\left(\frac{(T-\tau)^{-1/2}N^{1/2-\alpha_0}}{\log N}\right)\non
&&+ v_{NT}^{1/2} \cdot \left(O_P\left(\frac{\gamma_1(\tau)N^{1-2\alpha_0}}{\log N}\right)+O_P\left( \frac{(T-\tau)^{-1/2}N^{1-2\alpha_0}}{\log N}\right)\right)
\end{eqnarray}
and
\begin{eqnarray}\label{yr40zhang4}
E_4&=&v_{NT}^{1/2} \cdot O_P\left(\frac{(T-\tau)^{-1/2}N^{1/2-\alpha_0}}{\log N}\right)\non
&&+ v_{NT}^{1/2} \cdot \left(O_P\left(\frac{\gamma_1(\tau)N^{1-2\alpha_0}}{\log N}\right)+O_P\left( \frac{(T-\tau)^{-1/2}N^{1-2\alpha_0}}{\log N}\right)\right),
\end{eqnarray}
where $\kappa_{\tau}$ is defined in (\ref{160530}) and  $\sigma_{\tau}^2$ is defined in (\ref{163031}).

\end{thm}

\begin{prop}\label{yyr001fujia}
Let the conditions of Theorem \ref{yyr001} hold. 

(i) If, for either $\frac{1}{2}<\alpha_0\leq 1$ or when $0<\alpha_0\leq\frac{1}{2}$, as $(N,T)\rightarrow(\infty, \infty)$,
\begin{eqnarray}\label{yyr001b4wupian1}
\max\Big(\frac{\gamma_1(\tau)N^{1-2\alpha_0}}{\log N}, \frac{(T-\tau)^{-1/2}N^{1-2\alpha_0}}{\log N}\Big)\rightarrow 0,
\end{eqnarray}
then $\widetilde{\kappa}_{\tau}$ is an asymptotically unbiased estimator of $\kappa_{\tau}$.

(ii) Let
as $(N,T)\rightarrow(\infty, \infty)$,
\begin{eqnarray}\label{yyr001b4wupian1abc}
&&\max\Big(\frac{\gamma_1(\tau)N^{1-2\alpha_0}}{\kappa_{\tau}\log N}, \frac{(T-\tau)^{-1/2}N^{1-2\alpha_0}}{\kappa_{\tau}\log N}\Big)\rightarrow 0,\ \ \ as \ \ 0<\alpha_0\leq\frac{1}{2};\non
&&\max\Big(\frac{\gamma_1(\tau)N^{1-2\alpha_0}}{\kappa_{\tau}\log N},\frac{(T-\tau)^{-1/2}N^{1/2-\alpha_0})}{\kappa_{\tau}\log N}\Big)\rightarrow 0, \ \ \ as \ \ \frac{1}{2}<\alpha_0\leq 1.
\end{eqnarray}
Then $\widetilde{\alpha}_{\tau}$ is an asymptotically unbiased estimator of $\alpha_0$.

(iii) If, moreover,  as $(N,T)\rightarrow(\infty, \infty)$,
\begin{eqnarray}\label{yyr001b4}
&&v_{NT}^{1/2}\max\Big(\frac{\gamma_1(\tau)N^{1-2\alpha_0}}{\log N}, \frac{(T-\tau)^{-1/2}N^{1-2\alpha_0}}{\log N}\Big)\rightarrow 0,\ \ \ as \ \ 0<\alpha_0\leq\frac{1}{2};\non
&&v_{NT}^{1/2}\max\Big(\frac{\gamma_1(\tau)N^{1-2\alpha_0}}{\log N},\frac{(T-\tau)^{-1/2}N^{1/2-\alpha_0})}{\log N}\Big)\rightarrow 0, \ \ \ as \ \ \frac{1}{2}<\alpha_0\leq 1,
\end{eqnarray}
then
\begin{eqnarray}\label{yr40zhang}
\left(
\begin{array}{c}
\widetilde{\kappa}_{\tau}v_{NT}^{1/2}(N^{2(\widetilde{\alpha}_{\tau}-\alpha_0)}-1)\\
v_{NT}^{1/2}(\widetilde{\kappa}_{\tau}-\kappa_{\tau})\\
\end{array}
\right)
\stackrel{d}{\longrightarrow} \mathcal{N}\left(\left(
\begin{array}{c}
0\\
0\\
\end{array}
\right), \left(
\begin{array}{cc}
4\sigma_{\tau}^2 & -2\sigma_{\tau}^2\\
-2\sigma_{\tau}^2 & \sigma_{\tau}^2
\end{array}
\right)\right),
\end{eqnarray}
where $\kappa_{\tau}$ and $\sigma_{\tau}^2$ are defined in (\ref{160530}) and (\ref{163031}), respectively.

\end{prop}
\begin{rmk}\label{beizhu}
In fact, when $\frac{\gamma_1(\tau)}{\kappa_{\tau}}$ is bounded and $\frac{1}{2}<\alpha_0\leq 1$, $\widetilde{\alpha}_{\tau}$ is an asymptotically unbiased estimator of $\alpha_0$.

\end{rmk}

When the idiosyncratic components are independent, we can just use a finite lag $\tau$ (for example $\tau=1$). In this case, an asymptotic distribution for the estimator $\widehat{\alpha}$ is established in the following theorem.

\begin{thm}\label{coro1}
Let Assumptions \ref{A3} and \ref{A2} hold. In addition, suppose that $\tau$ is fixed and the following conditions (i)--(iii) hold:

(i) $\{\bbu_t: t=1,\ldots,T\}$ are independent with the mean of $\bbu_t$ being $\textbf{0}_{N\times 1}$ and its covariance matrix being $\boldsymbol{\Sigma}_{\bbu}$, where $\textbf{0}_{N\times 1}$ is an $N\times 1$ vector with zero components and the spectral norm $||\boldsymbol{\Sigma}_{\bbu}||$ is bounded.

(ii)
\begin{eqnarray}
&&\frac{v_{NT}^{1/2} (T-\tau)^{-1/2}N^{1-2\alpha_0}}{\log N} \rightarrow 0,\ \ \ as \ \ 0<\alpha_0<\frac{1}{2};
\label{new6}\\
&&\frac{v_{NT}^{1/2}(T-\tau)^{-1/2}N^{1/2-\alpha_0}}{\log N}\rightarrow 0, \ \ \ as \ \ \frac{1}{2}<\alpha_0\leq 1.
\nonumber
\end{eqnarray}

(iii)
$ v_{NT}^{1/2}\kappa_{\tau} \rightarrow \infty.$
\medskip

Then, as $(N,T)\rightarrow(\infty, \infty)$, we have \begin{eqnarray}\label{yr40zhangduli}
\left(
\begin{array}{c}
\widetilde{\kappa}_{\tau}v_{NT}^{1/2}(N^{2(\widetilde{\alpha}_{\tau}-\alpha_0)}-1)\\
v_{NT}^{1/2}(\widetilde{\kappa}_{\tau}-\kappa_{\tau})\\
\end{array}
\right)
\stackrel{d}{\longrightarrow} \mathcal{N}\left(\left(
\begin{array}{c}
0\\
0\\
\end{array}
\right), \left(
\begin{array}{cc}
4\sigma_{\tau}^2 & -2\sigma_{\tau}^2\\
-2\sigma_{\tau}^2 & \sigma_{\tau}^2
\end{array}
\right)\right),
\end{eqnarray}
 where $\kappa_{\tau}$ and $\sigma_{\tau}^2$ are defined in (\ref{160530}) and (\ref{163031}), respectively.
\end{thm}

Theorems \ref{yyr001}--3 and Proposition \ref{yyr001fujia} establish some asymptotic properties for the joint estimator $(\widetilde{\alpha}, \widetilde{\kappa})$. 
Before we will give the proofs of Theorems 1--3 in Appendices B and C below, we have some brief discussion about Condition (\ref{new6}), which is actually equivalent to the following three cases:

(a) $0<\alpha_0\leq\frac{1}{2}, \ [N^{\alpha_0}]<T-\tau, \ \frac{N^{1-3\alpha_0/2}}{(T-\tau)^{1/2}\log N}=o(1)$;

(b) $ \frac{1}{2}<\alpha_0\leq 1, \ [N^{\alpha_0}]<T-\tau; \ \frac{N^{1/2-\alpha_0/2}}{(T-\tau)^{1/2}\log N}=o(1)$;

(c) $\frac{1}{2}<\alpha_0\leq 1, \ [N^{\alpha_0}]\geq T-\tau,  \ \frac{N^{1/2-\alpha_0}}{\log N}=o(1)$.
\medskip

Under these three cases, we can provide some choices for $(N, T)$ as follows:
\medskip

(d) $0<\alpha_0<\frac{1}{2}, \ [N^{\alpha_0}]<T-\tau; \  T=\tau+[N^{2-3\alpha_0}]$;

(f) $\frac{1}{2}<\alpha_0\leq 1, \ [N^{\alpha_0}]<T-\tau, \ T=\tau+[N^{\alpha_0}]$;

(g) $\frac{1}{2}<\alpha_0\leq 1, \ [N^{\alpha_0}]\geq T-\tau, \ T=\tau+[N^{\alpha_0}/\log(N)]$.
\medskip

When $\tau\rightarrow\infty$, the term $\kappa_0$ will tend to $0$, because of $\boldsymbol{\Sigma}_{\tau}\rightarrow\textbf{0}$. So, as $\tau$ is very large, the value of $\ln (\kappa_0)$ may be negative in practice. Hence Theorem \ref{thm1} provides an alternative form for the asymptotic distribution of $N^{\widehat\alpha-\alpha_0}$ instead of $\widehat\alpha-\alpha_0$, and the case of $\tau$ being fixed is discussed in Theorem \ref{coro1}.

\subsection{Estimation for $\sigma_\tau^2$}
In this section, we propose an estimator for the parameter $\sigma_\tau^2$ in the asymptotic variance of established theorems above.

Let $n=[N^{\widetilde{\alpha}_{\tau}}]$ and
{\small
\begin{eqnarray}\label{a7gujisigma03}
\widehat\sigma_{i,T}(\tau)=\frac{1}{(T-\tau)}\sum_{t=1}^{T-\tau}\left( x_{it}-\frac{1}{(T-\tau)}\sum_{t=1}^{T-\tau}x_{it}\right) \left(\frac{1}{n}\sum_{i=1}^n x_{i,t+\tau}-\frac{1}{n(T-\tau)}\sum_{t=1}^{T-\tau} \sum_{i=1}^n x_{i,t+\tau}\right).
\end{eqnarray}
}The estimator for the first part of $\sigma_{\tau}^2$ in (\ref{163031}) is
\begin{eqnarray}
\widehat{\sigma}_{\tau,(1)}^2=\frac{4v_{nT}}{n}\widehat{\sigma}_{\tau,n}^2,
\end{eqnarray}
where $\widehat{\sigma}_{\tau,n}^2=\frac{1}{n-1}\sum^{n}_{i=1}\left(\widehat{\sigma}_{i,T}(\tau)-\widehat{\sigma}_{T}(\tau)\right)^2$ with $\widehat{\sigma}_T(\tau)=\frac{1}{n}\sum^{n}_{i=1}\widehat{\sigma}_{i,T}(\tau)$.

For the second term (\ref{y111}), the proposed estimator is
\begin{eqnarray}
\widehat{\sigma}_{\tau,(2)}^2=\frac{v_{nT}}{T-\tau}\widehat{\sigma}_{0,T}^2,
\end{eqnarray}
where
{\small
\begin{eqnarray*}
\widehat{\sigma}_{0,T}^2=\frac{1}{T-\tau-1}\sum^{T-\tau}_{t=1}\left(\widehat{\sigma}_{n,t}(\tau)-\widehat{\sigma}_n(\tau)\right)^2+\sum^{\ell}_{j=1}\frac{2}{T-\tau-j}\sum^{T-\tau-j}_{t=1}\left(\widehat{\sigma}_{n,t}(\tau)-\widehat{\sigma}_{n}(\tau)\right)\left(\widehat{\sigma}_{n,t+j}(\tau)-\widehat{\sigma}_n(\tau)\right),
\end{eqnarray*}
}with
\begin{eqnarray*}
\widehat{\sigma}_{n,t}(\tau)=\left(\frac{1}{n}\sum^{n}_{i=1}x_{it}-\frac{1}{n(T-\tau)}\sum^{T-\tau}_{t=1}\sum^{n}_{i=1}x_{it}\right)\left(\frac{1}{n}\sum^{n}_{i=1}x_{i,t+\tau}-\frac{1}{n(T-\tau)}\sum^{T-\tau}_{t=1}\sum^{n}_{i=1}x_{i,t+\tau}\right)
\end{eqnarray*}
and $\widehat{\sigma}_n(\tau)=\frac{1}{T-\tau}\sum^{T-\tau}_{t=1}\widehat{\sigma}_{n,t}(\tau)$.

Then the estimator for $\sigma^2_{\tau}$ is
\begin{eqnarray}
\widehat{\sigma}^2_{\tau}=\widehat{\sigma}^2_{\tau,(1)}+\widehat{\sigma}^2_{\tau,(2)}.
\end{eqnarray}
\begin{prop}\label{yyr001fujiaguji}
Under the conditions of Theorem \ref{thm1} and $N^{\alpha-1/2}(T-\tau)^{1/2}\rightarrow \infty$ when $\widetilde{\alpha}_{\tau}$ is a consistent estimator of $\alpha_0$,
\begin{eqnarray}\label{a7gujisigma05}
\widehat{\sigma}_{\tau,(1)}^2=\frac{4v_{nT}}{n}\boldsymbol{\mu}_v^{'}\boldsymbol{\Sigma}_{\tau}^{'}\boldsymbol{\Sigma}_v\boldsymbol{\Sigma}_{\tau}\boldsymbol{\mu}_v(1+o_p(1))
\end{eqnarray}
and
\begin{eqnarray}\label{a7gujisigma02a2s}
\widehat{\sigma}^2_{\tau,(2)}=\frac{v_{nT}}{T-\tau}(\boldsymbol{\mu}_v^{'}\otimes\boldsymbol{\mu}_v^{'})\boldsymbol{\Omega}(\boldsymbol{\mu}_v\otimes\boldsymbol{\mu}_v)(1+o_p(1))+o_p(1).
\end{eqnarray}
Then  $\widehat{\sigma}^2_{\tau}$ is a consistent estimator of $\sigma^2_{\tau}$.
\end{prop}

The proof of Proposition 3 is provided in Appendix C. Next, we evaluate the finite--sample performance of the proposed estimation methods and the resulting theory in Sections 4 and 5 below.


\section{Simulation}
In this section, we use three data generating processes (DGPs) to illustrate the finite sample performance of the joint and marginal estimators in different scenarios.

Example 1 studies the case of i.i.d error component and AR(1) modelled common factors. Example 2 extends the i.i.d error component in Example 1 to an AR(1) model. In Example 3, we investigate an MA(q) model for the common factors and an AR(1) type error component.

Before our analysis of each example, we provide a method of choosing an optimal value of $\tau$ in the following way.
\begin{eqnarray}
\widetilde{\tau}=\max_{\tau}\frac{\widetilde{\kappa}_{\tau}}{Q_{NT}^{(1)}\left(\widetilde{\alpha}_{\tau}, \widetilde{\kappa}_{\tau}, \tau\right)}
\end{eqnarray}
The idea of this proposed criterion is that we choose a value of $\tau$ to make larger $\widetilde{\kappa}_{\tau}$ and smaller $Q_{NT}^{(1)}\left(\widetilde{\alpha}_{\tau}, \widetilde{\kappa}_{\tau}, \tau\right)$. As $\widetilde{\kappa}_{\tau}$ contains temporal dependence in the common factors, it is reasonable to consider its large value to take into account the information included in it. For $Q_{NT}^{(1)}\left(\widetilde{\alpha}_{\tau}, \widetilde{\kappa}_{\tau}\right)$, it is the objective function for the joint estimator and hence we expect its small value corresponding to an optimal $\tau$.

\subsection{Example 1: no serial dependence in errors}
First,  we consider the following two-factor model
\begin{eqnarray}\label{yadd1}
x_{it}=\mu+\beta_{i1}f_{1t}+\beta_{i2}f_{2t}+u_{it}, \ \ i=1,2,\ldots,N; t=1,2,\ldots,T.
\end{eqnarray}
The factors are generated by
\begin{eqnarray}\label{yadd2}
f_{jt}=\rho_jf_{j,t-1}+\zeta_{jt}, \ \ j=1,2; \ t=-49,-48,\ldots,0,1,\ldots,T,
\end{eqnarray}
with $f_{j,-50}=0$ for $j=1,2$ and $\zeta_{jt}\stackrel{i.i.d}{\sim}\mathcal{N}(0,1)$. The idiosyncratic components $u_{it}$ are i.i.d from $\mathcal{N}(0, 1)$ and independent of $\{\zeta_{jt}: t=1, 2, \ldots, T; j=1, 2\}$.

The factor loadings are generated as
\begin{eqnarray}\label{yadd3}
&&\beta_{ij}=v_{ij}, \ \ for\ i=1,2,\ldots,M; j=1, 2;\non
&&\beta_{ij}=\rho^{i-M}, \ \ for \ i=M+1,M+2,\ldots,N; j=1, 2,
\end{eqnarray}
where $v_{ir}\stackrel{i.i.d}{\sim}U(0.5, 1.5)$, $M=[N^{\alpha_0}]$ and $\rho=0.8$. Moreover, we set $\mu=1$ and $\rho_j=0.9$ for $j=1,2$.

Under this data generating process, the numerical values of $\widehat{\alpha}$ and $\widetilde{\alpha}$ are reported in Table \ref{tb1}. The confidence interval for $\alpha_0$ is also calculated, i.e.
\begin{eqnarray}\label{yadd6}
\left[\widetilde{\alpha}-\frac{\log\left(1+\frac{z_{0.025}\cdot 4\sigma_{\tau}^2}{\widehat{\kappa}_{\tau}v_{NT}^{1/2}}\right)}{2\log(N)}, \widetilde{\alpha}+\frac{\log\left(1+\frac{z_{0.025}\cdot 4\sigma_{\tau}^2}{\widehat{\kappa}_{\tau}v_{NT}^{1/2}}\right)}{2\log(N)}\right]
\end{eqnarray}
which can be derived for the asymptotic distribution of $\widetilde{\alpha}$ in Theorem \ref{coro1}.
\begin{center}
\fbox{Table \ref{tb1} near here}
\end{center}
In this table, the number of cross-sections was $N=200$ and the time-length was $T=200$. As $\tau=0$, the estimator $\widehat{\alpha}$ is equivalent to the estimator provided in \cite{BKP2016}. Moreover, we calculate the marginal and joint estimates when $\tau=1, 2, 3, 4$. From Table \ref{tb1}, it can be seen that the estimator in \cite{BKP2016} behaves well in the case of $\alpha_0=0.8$ while becomes inconsistent as $\alpha_0=0.5$ and $\alpha_0=0.2$. When $\tau>0$, our estimator performs well for all cases including $\alpha_0=0.8, 0.5, 0.2$. However, as $\tau$ increases, the confidence interval will become larger. This phenomenon is consistent with our theoretical result since the confidence interval in (\ref{yadd6}) depends on $\widetilde{\kappa}_{tau}$ and this parameter decreases for the AR(1) model as $\tau$ increases. Furthermore, compared with the marginal estimation, the joint estimation is a bit worse than marginal estimator as expected due to more information is known in marginal estimation.

From this example, we can see that the estimator with $\tau>0$ is consistent when $\alpha_0<0.5$ at the cost of larger variance. Moreover, the estimator $\widetilde{\alpha}_{\widetilde{\tau}}$, with the choice of $\widetilde{\tau}$, is better than others, although there is a bit deviation. Intuitively, since $\kappa_{\tau}$ will decrease as $\tau$ increases under this example, and the error component has no temporal dependence, the optimal $\tau$ should be $1$ intuitively.

\subsection{Example 2: AR(1) for factors and errors}
In this part, we also consider the factor model (\ref{yadd1}). In this model, the common factors and factor loadings follow (\ref{yadd2}) and (\ref{yadd3}) respectively. The idiosyncratic error component $u_{it}$ follows an AR(1) model as follows:
\begin{eqnarray}
u_{it}=\eta_i\cdot g_t, \ \ g_t=h\cdot g_{t-1}+\epsilon_t, \ \ i=1, 2, \ldots, N; \ t=-49, -48, \ldots, 0, 1, \ldots, T,
\end{eqnarray}
where $g_{-50}=0$ and $\eta_i\sim\mathcal{N}(0, 1)$. Moreover, $u_{it}$ are independent of $\zeta_{jt}$ in (\ref{yadd2}). Here $h=0.2$ and $h=\frac{1}{\sqrt{N}}$ which are smaller than the strength of time serial dependence in the common factors with $\rho=0.9$.

The results when $h=0.2$ or $h=\frac{1}{\sqrt{N}}$ are listed in Tables \ref{tb2} and \ref{tb3}, respectively.
\begin{center}
\fbox{Table \ref{tb2} and Table \ref{tb3} near here}
\end{center}
These two tables are derived as $N=400, T=200$. In Table \ref{tb2}, the values of $\tau$ are comparably large in order to ensure that the common factors and the idiosyncratic error components have different strengths of time-serial dependence. In fact, the autocorrelation of the common factors and the error component with time--lag $\tau$ are of the orders $\rho^{\tau}_j$ and $h^{\tau}$, respectively. When $h$ is constant, the error component has a weak strength order only as $\tau$ tends to infinity. When $h=\frac{1}{\sqrt{N}}$, $h^{\tau}$ tends to zero for any value of $\tau$. This is why $\tau$ in Table \ref{tb3} takes relatively small values.

This example is more complicated than Example 1 due to time--serial dependence in the error component. Similar to Example 1, the marginal or the joint estimator with $\tau=0$ performs inconsistently when $\alpha_0=0.5$ and $\alpha_0=0.2$, while the marginal or the joint estimator with $\tau>0$ is consistent for all cases. Compared with Example 1, all the results have relatively larger variances due to the complex structures.

Furthermore, the choice for $\widetilde{\tau}$ in Table \ref{tb2} is around 5 while that in Table \ref{tb3} is close to $1$. These choices are reasonable due to strong time-serial dependence in the error component for the date generating process in Table \ref{tb2}.

\subsection{Example 3: MA(2) for factors and errors}
Now we consider the factor model (\ref{yadd1}). In this model, factor loadings follow (\ref{yadd3}), respectively.
The common factors $f_{jt}$ and the idiosyncratic error component $u_{it}$ follow MA(2) models as follows:
\begin{eqnarray}
f_{jt}=Z_{jt}+\theta_1Z_{j,t-1}+\theta_2Z_{j,t-2}, \ \ j=1, 2; \ t=1, 2, \ldots, T,
\end{eqnarray}
and
\begin{eqnarray}
u_{it}=\eta_i\cdot g_t, \ \ g_t=K_t+h_1\cdot K_{t-1}+h_2\cdot K_{t-2}, \ \ i=1, 2, \ldots, N; \ t=1, 2,  \ldots, T,
\end{eqnarray}
where $Z_{j,s}\sim\mathcal{N}(0, 1)$, $K_s\sim\mathcal{N}(0, 1)$ with $s=-1, 0, 1, 2, \ldots, T$ and $\eta_i\sim\mathcal{N}(0, 1)$.
Moreover, $u_{it}$ are independent of $\zeta_{jt}$ in (\ref{yadd2}). Here $\theta_1=0.8$, $\theta_2=0.6$, $h_1=\frac{1}{\log(N)}$ and $h_2=\frac{1}{\sqrt{N}}$.

Table \ref{tb4} presents the results for the case of $N=400$ and $T=200$.
\begin{center}
\fbox{Table \ref{tb4} near here}
\end{center}

As well known, MA(2) model has zero autocorrelation when $\tau>2$. Hence Table \ref{tb4} reports the results with $\tau=0, 1, 2$. We impose a ``weaker" MA model for error components in the sense of the coefficients tend to zero. This is to guarantee weaker strength of the time-serial dependence in the error component than that for the common factors.

Except similar observations to the first two examples, the proposed estimator behaves a bit better under MA structure than AR structure in examples 1 and 2. The main reason relies on that MA structure has larger time-serial dependence than that for the AR structure.

Similarly, the value of $\widetilde{\tau}$ is also around $1$. This is because the temporal dependence in the error component is quite weak.


\section{Empirical applications}

In this section, we show how to obtain an estimate for the exponent of cross--sectional dependence, $\alpha_0$, for each of the following panel data sets: quarterly cross-country data used in global modelling and daily stock returns on the constitutes of Standard and Poor $500$ index.

\subsection{Cross-country dependence of macro-variables}

We provide an estimate for $\alpha_0$ for each of the datasets: Real GDP growth (RGDP), Consumer price index (CPI), Nominal equity price index (NOMEQ), Exchange rate of country $i$ at time $t$ expressed in US dollars (FXdol), Nominal price of oil in US dollars (POILdolL), and Nominal short-term and long-term interest rate per annum (Rshort and Rlong) computed over $33$ countries.\footnote{The datasets are downloaded from http://www-cfap.jbs.cam.ac.uk/research/gvartoolbox/download.html.} The observed cross-country time series, $y_{it}$, over the full sample period, are standardized as $x_{it}=(y_{it}-\bar y_i)/s_i$, where $\bar y_{i}$ is the sample mean and $s_i$ is the corresponding standard deviation for each of the time series. Table \ref{tb7} reports the corresponding results.

For the standardized data $x_{it}$, we regress it on the cross-section mean $\bar x_t=\frac{1}{N}\sum^{N}_{i=1}x_{it}$, i.e., $x_{it}=\delta_i\bar x_t+u_{it}$ for $i=1,2,\ldots,N$, where $\delta_i$, $i=1,2,\ldots,N$, are the regression coefficients. With the availability of the OLS estimate $\widehat\delta_i$ for $\delta_i$, we have the estimated versions, $\widehat u_{it}$, of the form: $\widehat u_{it}=x_{it}-\hat\delta_i\bar x_t$.

Since our proposed estimation methods rely on the different extent of serial dependence of the factors and idiosyncratic components, we provide some autocorrelation graphs of $\{\bar x_t=\frac{1}{N}\sum^{N}_{i=1}x_{it}: t=1,2,\ldots,T\}$ and $\{\bar u_t=\frac{1}{N}\sum^{N}_{i=1}u_{it}: t=1,2,\ldots,T\}$ for each group of the real dataset under investigation (see Figures \ref{fig1}--\ref{fig4}). From these graphs, it is easy to see that CPI, NOMEQ, FXdol and POILdolL have distinctive serial dependences in the factor part $\bar x_{t}$ and idiosyncratic part $\bar u_{t}$. All the observed real data $x_{it}$ are serially dependent.

\begin{center}
\fbox{Figures \ref{fig1}-\ref{fig4} near here}
\end{center}

Due to the existence of serial dependence in the idiosyncratic component, we use the proposed second moment criterion. The marginal estimator $\widehat\alpha$ and the joint estimator $\widetilde{\alpha}$ for these real data are provided in Table \ref{tb7}. We use $\tau=10$ for two estimators. We can see from Table \ref{tb7} that the values of $\widehat\alpha$ and $\widetilde{\alpha}$ are different from the those provided by \cite{BKP2016}.  Some estimated values are not $1$. This phenomenon implies that a factor structure might be a good approximation for modelling global dependency, and the value of $\alpha_0=1$ typically assumed in the empirical factor literature might be exaggerating the importance of the common factors for modelling cross-sectional dependence at the expense of other forms of dependency that originate from trade or financial inter-linkage that are more local or regional rather than global in nature. Furthermore, it is noted that our model is different from that given by \cite{BKP2016} and the difference mainly lies on that our model only imposes serial dependence on factor processes and assumes that the idiosyncratic errors are independent. Different models may bring in different exponents.

\begin{center}
\fbox{Table \ref{tb7} near here}
\end{center}

\subsection{Cross--sectional exponent of stock-returns}

One of the important considerations in the analysis of financial markets is the extent to which asset returns are interconnected. The classical model is the capital asset pricing model of \cite{S1964} and the arbitrage pricing theory of \cite{R1976}. Both theories have factor representations with at least one strong common factor and an idiosyncratic component that could be weakly cross-sectionally correlated (see \cite{C1983}). The strength of the factors in these asset pricing models is measured by the exponent of the cross-sectional dependence, $\alpha_0$. When $\alpha_0=1$, as it is typically assumed in the literature, all individual stock returns are significantly affected by the factors, but there is no reason to believe that this will be the case for all assets and at all times. The disconnection between some asset returns and the market factors could occur particularly at times of stock market booms and busts where some asset returns could be driven by some non-fundamentals. Therefore, it would be of interest to investigate possible time variations in the exponent $\alpha_0$ for stock returns.

We base our empirical analysis on daily returns of $96$ stocks in the Standard $\&$ Poor $500$ (S$\&$P500) market during the period of January, 2011-December, 2012. The observations $r_{it}$ are standardized as $x_{it}=(r_{it}-\bar r_i)/s_i$, where $\bar r_i$ is the sample mean of the returns over all the sample and $s_i$ is the corresponding standard deviation. For the standardized data $x_{it}$, we regress it on the cross-section mean $\bar x_t=\frac{1}{N}\sum^{N}_{i=1}x_{it}$, i.e., $x_{it}=\delta_i\bar x_t+u_{it}$ for $i=1,2,\ldots,N$, where $\delta_i$, $i=1,2,\ldots,N$, are the regression coefficients. Based on the OLS estimates: $\widehat\delta_i$ for $\delta_i$, we define $\widehat u_{it}=x_{it}-\hat\delta_i\bar x_t$. The autocorrelation functions (ACFs) of the cross-sectional averages $\bar x_{t}=\frac{1}{N}\sum^{N}_{i=1}x_{it}$ and $\bar u_{t}=\frac{1}{N}\sum^{N}_{i=1}u_{it}$ are presented in Figure \ref{fig5}.

\begin{center}
\fbox{Figure \ref{fig5} near here}
\end{center}

From Figure \ref{fig5}, we can see that the serial dependency of the common factor component is stronger than that of the idiosyncratic component. We use the estimates $\widehat\alpha$ and $\widetilde{\alpha}$ to characterize the serial dependences of the common factors and the idiosyncratic component. The estimates $\widehat\alpha$ and $\widetilde{\alpha}$ are calculated with the choice of $\tau=10$. Table \ref{tb8} reports the estimates with several different sample sizes. As comparison, the estimates from \cite{BKP2016} are also reported. From the table, we can see that their estimation method does not work when $\alpha$ is smaller than $1/2$. The results also show that the cross-sectional exponent of stock returns in S$\&$P500 are smaller than $1$. This indicates the support of using different levels of loadings for the common factor model as assumed in Assumption 2, rather than using the same level of loadings in such scenarios.
\medskip

\begin{center}
\fbox{Table \ref{tb8} near here}
\end{center}

Furthermore, Figure \ref{fig6} provides the marginal estimate $\widehat\alpha$ and the joint estimate $\widetilde{\alpha}$ for the first $130$ days of all the period. It shows that the estimated values for $\alpha_0$ with the two methods are quite similar. On the other hand, since a $130$-day period is short, meanwhile, it is reasonable that the estimates didn't change very much.

\begin{center}
\fbox{Figure \ref{fig6} near here}
\end{center}

\section{Conclusions and discussion}

In this paper, we have examined the issue of how to estimate the extent of cross--sectional dependence for large dimensional panel data. The extent of cross--sectional dependence
is parameterized as $\alpha_0$, by assuming that only $[N^{\alpha_0}]$ sections are relatively strongly dependent. Compared with the estimation method proposed by \cite{BKP2016}, we have developed a unified `moment' method to estimate $\alpha_0$. Especially,  when stronger serial dependence exists in the factor process than that for the idiosyncratic errors (dynamic principal component analysis), we have recommended the use of the covariance function between the cross-sectional average values of the observed data at different lags to estimate $\alpha_0$. One advantage of this new approach is that it can deal with the case of $0\leq\alpha_0\leq1/2$.

Due to some unknown parameters involved in the panel data model, in addition to the proposed marginal estimation method, we have also constructed a joint estimation method for $\alpha_0$ and the related unknown parameters. The asymptotic properties of the estimators have all been established. The simulation results and an empirical application to two datasets have shown that the proposed estimation methods work well numerically.

Future research includes discussion about how to estimate factors and factor loadings in factor models, and determine the number of factors for the case of $0<\alpha_0<1$. Existing methods available for factor models, such as \cite{Baing2002}, \cite{AH2013}, \cite{ABC2010}, \cite{O2010}, for the case of $\alpha_0=1$, may not be directly applicable, and should be extended to deal with the case of $0<\alpha_0<1$. Such issues are all left for future work.

\section{Acknowledgements}

The first, the third and fourth authors acknowledge the Australian Research Council Discovery Grants Program for its support under Grant numbers: DP150101012 \& DP170104421. Thanks from the second author also go to the Ministry of Education in Singapore for its financial support under Grant $\#$ ARC 14/11.

{\small

\begin{figure}[h]
\caption{ACF of RGDP and CPI}\label{fig1}
\centering
\includegraphics[scale=0.4]{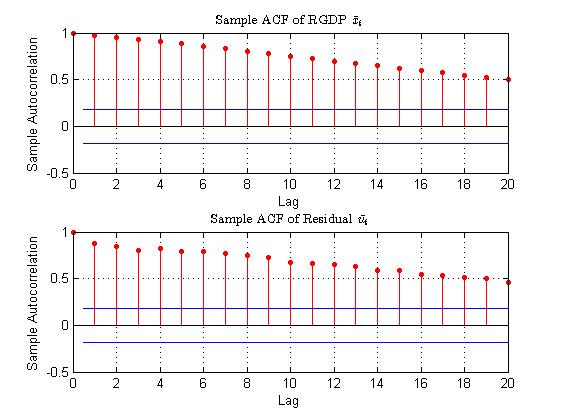}
\includegraphics[scale=0.4]{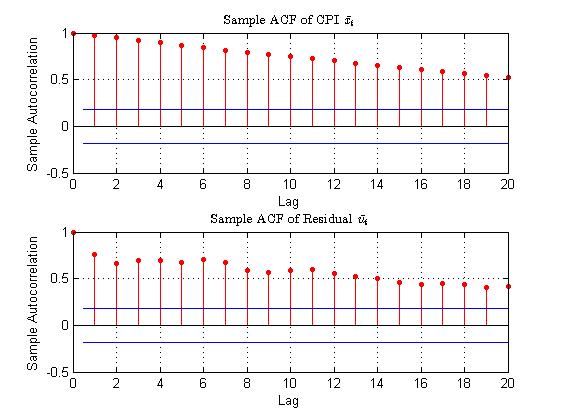}
\\
\end{figure}

\begin{figure}[h]
\caption{ACF of NOMEQ and FXdol}\label{fig2}
\centering
\includegraphics[scale=0.4]{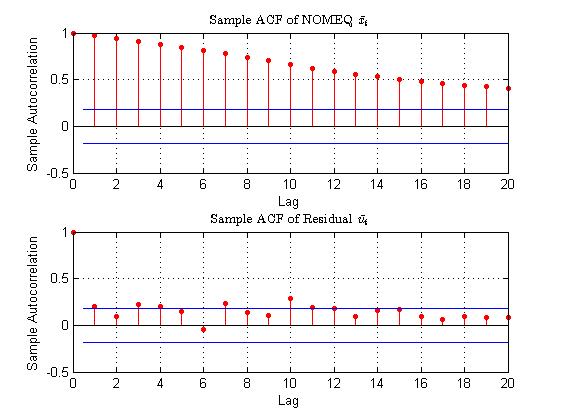}
\includegraphics[scale=0.4]{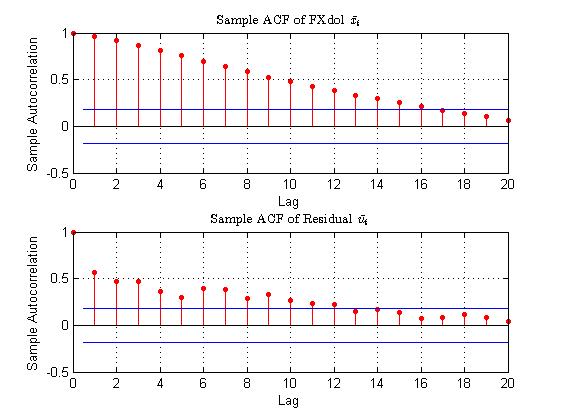}
\\
\end{figure}

\begin{figure}[h]
\caption{ACF of Rshort and Rlong}\label{fig3}
\centering
\includegraphics[scale=0.4]{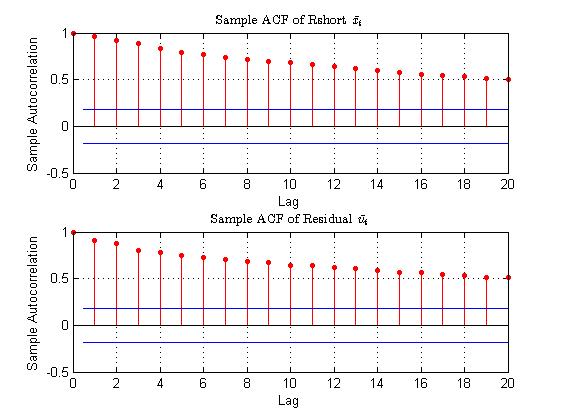}
\includegraphics[scale=0.4]{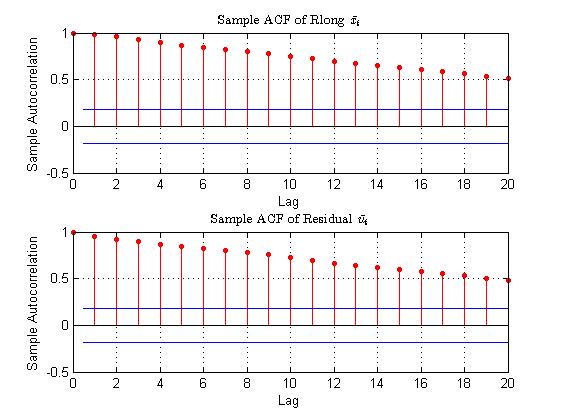}
\\
\end{figure}

\begin{figure}[h]
\caption{ACF of POILdolL}\label{fig4}
\centering
\includegraphics[scale=0.4]{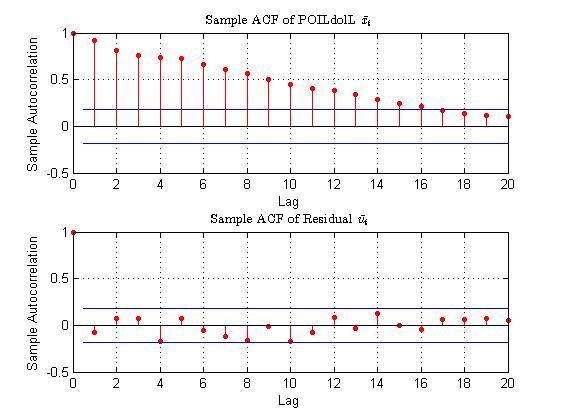}
\\
\end{figure}


\begin{figure}[h]
\caption{ACF of averages of 96 stock returns}\label{fig5}
\centering
\includegraphics[scale=0.4]{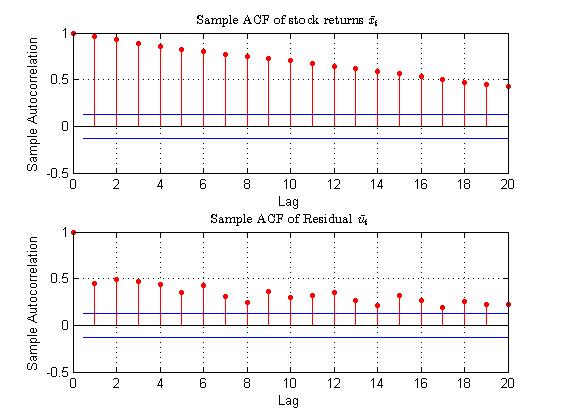}
\\
\end{figure}

\begin{figure}[h]
\caption{130-day joint and marginal estimators for $96$ stocks of S$\&$P 500}\label{fig6}
\centering
\includegraphics[scale=0.4]{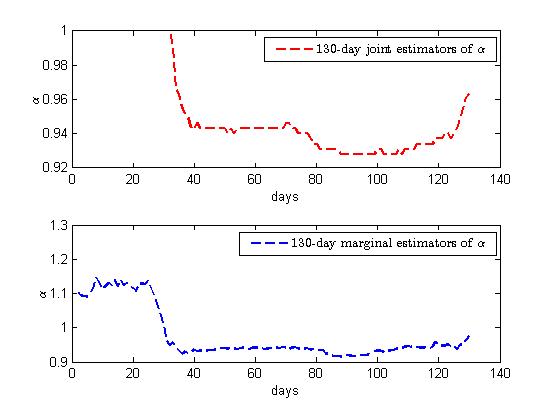}
\\
\end{figure}

\newpage

\begin{table}[h]
{\small \caption{\label{tb1}: Example 1} }
\begin{center}
{\scriptsize
\begin{tabular}{ccccccccc}
\hline
 $\tau$  &  0  & 1   & 2   & 3   &  4  \\
\hline
$\alpha_0=0.8$           &       $\widetilde{\tau}=1$        &              &    &    &    \\
$\widehat{\alpha}$     &  0.7912 & 0.7899 & 0.8101& 0.8112 & 0.8127\\
$\widetilde{\alpha}$   &  0.7901 & 0.7892 & 0.8123  & 0.8120   & 0.8134\\
$90\%$ CI Upper       &  0.7933 & 0.7929 & 0.8407 & 0.8420 & 0.8467 \\
$90\%$ CI Lower      &  0.7869 & 0.7855 &  0.7839 & 0.7820 & 0.7801 \\
\hline
$\alpha_0=0.5$ & $\widetilde{\tau}=1$ & & & & &\\
$\widehat{\alpha}$ & 0.5972    & 0.4901   & 0.5112   & 0.5127   &  0.5130 \\
$\widetilde{\alpha}$   & 0.6011   & 0.4900   & 0.5119   & 0.5131   &   0.5135 \\
$95\%$ CI Upper        & 0.6153 & 0.4945 & 0.5399 & 0.5441 & 0.5469 \\
$95\%$ CI Lower        & 0.5869 & 0.4855 & 0.4839 & 0.4821 & 0.4801\\
\hline
$\alpha_0=0.2$ & $\widetilde{\tau}=2$ & & & & &\\
$\widehat{\alpha}$ &  0.3140  & 0.2004   & 0.2042   & 0. 2086  & 0.2101 \\
$\widetilde{\alpha}$   &  0.3213  &  0.2010  &  0.2055  &  0.2091  &  0.2114 \\
$90\%$ CI Upper & 0.3457 & 0.2166 & 0.2272 & 0.2361 & 0.2427 \\
$90\%$ CI Lower & 0.2969 & 0.1854 & 0.1838 & 0.1821 & 0.1801\\
\hline
\end{tabular}}
\\ \end{center}\medskip
\end{table}

\begin{table}[h]
{\small \caption{\label{tb2}: Example 2 with $h=0.2$} }
\begin{center}
{\scriptsize
\begin{tabular}{cccccccc}
\hline
 $\tau$  &  0  & 5   & 10   & 15   &  20  \\
\hline
$\alpha_0=0.8$           &     $\widetilde{\tau}=4$          &              &    &    &    \\
$\widehat{\alpha}$     &  0.8101 & 0.8091 & 0.8199& 0.8231 & 0.8304\\
$\widetilde{\alpha}$   &  0.8144 & 0.8098 & 0.8205  & 0.8279   & 0.8311\\
$90\%$ CI Upper       &  0.8505 & 0.8334 & 0.8645 & 0.8798 & 0.9295 \\
$90\%$ CI Lower      &  0.7783 & 0.7862 &  0.7765 & 0.77603 & 0.7327 \\
\hline
$\alpha_0=0.5$ &$\widetilde{\tau}=7$ & & & & &\\
$\widehat{\alpha}$ & 0.6018    & 0.5078   & 0.5183   & 0.5217   &  0.5368 \\
$\widetilde{\alpha}$   & 0.6087   & 0.5101   & 0.5214   & 0.5268   &   0.5401 \\
$90\%$ CI Upper& 0.6383 & 0.5340 & 0.5663 & 0.5933 & 0.6475 \\
$90\%$ CI Lower& 0.5791 & 0.4862 & 0.4765 & 0.4603 & 0.4327\\
\hline
$\alpha_0=0.2$ & $\widetilde{\tau}=3$ & & & & &\\
$\widehat{\alpha}$ &  0.3127  & 0.2009   & 0.2147   & 0. 2201  & 0.2366 \\
$\widetilde{\alpha}$   &  0.3203 &  0.2048  &  0.2179  &  0.2293  &  0.2517 \\
$90\%$ CI Upper & 0.3510 & 0.2234 & 0.2593 & 0.2983 & 0.3707 \\
$90\%$ CI Lower & 0.2896 & 0.1862 & 0.1765 & 0.1603 & 0.1327\\
\hline
\end{tabular}}
\\ \end{center}\medskip
\end{table}

\begin{table}[h]
{\small \caption{\label{tb3}: Example 2 with $h=\frac{1}{\sqrt{N}}$} }
\begin{center}
{\scriptsize
\begin{tabular}{cccccccc}
\hline
 $\tau$  &  0  & 1   & 2   & 3   &  4  \\
\hline
$\alpha_0=0.8$           &      $\widetilde{\tau}=2$         &              &    &    &    \\
$\widehat{\alpha}$     &  0.8002 & 0.8011 & 0.8148& 0.8160 & 0.8169\\
$\widetilde{\alpha}$   &  0.8098 & 0.8018 & 0.8191  & 0.8199   & 0.8205\\
$90\%$ CI Upper       &  0.8326 & 0.8237 & 0.8669 & 0.8749 & 0.8807 \\
$90\%$ CI Lower      &  0.7870 & 0.7799 &  0.7713 & 0.7649 & 0.7603 \\
\hline
$\alpha_0=0.5$ & $\widetilde{\tau}=1$& & & & &\\
$\widehat{\alpha}$ & 0.6102    & 0.5099   & 0.5184   & 0.5210   &  0.5221 \\
$\widetilde{\alpha}$   & 0.6183   & 0.5110   & 0.5192   & 0.5222   &   0.5238 \\
$90\%$ CI Upper& 0.6665 & 0.5447 & 0.5681 & 0.5752 & 0.5809 \\
$90\%$ CI Lower& 0.5701 & 0.4773 & 0.4703 & 0.4692 & 0.4667\\
\hline
$\alpha_0=0.2$ &$\widetilde{\tau}=1$ & & & & &\\
$\widehat{\alpha}$ &  0.3201  & 0.2048   & 0.2102   & 0. 2112  & 0.2199 \\
$\widetilde{\alpha}$   &  0.3211  &  0.2100  &  0.2108  &  0.2141  &  0.2205 \\
$90\%$ CI Upper & 0.3531 & 0.2399 & 0.2447 & 0.2554 & 0.2773 \\
$90\%$ CI Lower & 0.2891 & 0.1801 & 0.1769 & 0.1728 & 0.1637\\
\hline
\end{tabular}}
\\ \end{center}\medskip
\end{table}

\begin{table}[htp]
{\small \caption{\label{tb4}: Example 3} }
\begin{center}
{\scriptsize
\begin{tabular}{cccccccc}
\hline
 $\tau$  &  0  & 1   & 2  \\
\hline
$\alpha_0=0.8$           &   $\widetilde{\tau}=2$            &              &     \\
$\widehat{\alpha}$     &  0.8006 & 0.8022 & 0.8104\\
$\widetilde{\alpha}$   &  0.8010 & 0.8038 & 0. 8118 \\
$90\%$ CI Upper       &  0.8102 & 0.8204 & 0.8509   \\
$90\%$ CI Lower      &  0.7918 & 0.7872 &  0.7727 \\
\hline
$\alpha_0=0.5$ & $\widetilde{\tau}=2$& & & & &\\
$\widehat{\alpha}$ & 0.5999    & 0.5034   & 0.5127  \\
$\widetilde{\alpha}$   & 0.6103   & 0.5089   & 0.5189  \\
$90\%$ CI Upper& 0.6288 & 0.5306 & 0.5651\\
$90\%$ CI Lower& 0.5918 & 0.4872 & 0.4727 \\
\hline
$\alpha_0=0.2$ & $\widetilde{\tau}=1$& & & & &\\
$\widehat{\alpha}$ &  0.3048  & 0.2089   & 0.2103  \\
$\widetilde{\alpha}$   &  0.3104  &  0.2103  &  0.2193  \\
$90\%$ CI Upper & 0.3290 & 0.2334 & 0.2659 \\
$90\%$ CI Lower & 0.2918 & 0.1872 & 0.1727 \\
\hline
\end{tabular}}
\\ \end{center}\medskip
\end{table}

\begin{table}[h]
{\small \caption{\label{tb7}Exponent of cross-country dependence of macro-variables}}
\begin{center}
{\footnotesize
\begin{tabular}{cccccccccccccccccc}
&\multicolumn{5}{c}{}\\
\hline
 & N  &  T  &  $\widehat\alpha$   &  $\widetilde\alpha$   &  BKP16  \\
\hline
Real GDP growth(RGDP)                & 33 & 122 &	0.925&	0.916&	0.937   \\
Upper Bound  & & & & 1.028 & 1.038 \\
Lower Bound  & & & & 0.804 & 0.836\\
Consumer Price Index(CPI)             & 33 & 122 &	0.947&	0.933&	0.949    \\
Upper Bound  & & & &  1.054 & 1.073 \\
Lower Bound  & & & & 0.812 & 0.825\\
Nominal equity price index(NOMEQ) & 33 & 122 &	0.952&	0.921&	0.959   \\
Upper Bound  & & & & 1.041 & 1.079 \\
Lower Bound  & & & & 0.801 & 0.839\\
Short-term interest rates(Rshort)      & 33 & 122 &	0.943&	0.938&	0.961   \\
Upper Bound  & & & & & \\
Lower Bound  & & & & 1.010 & 1.032\\
Long-term interest rates(Rlong)       & 33 & 122 &	0.984&	0.971&	0.982    \\
Upper Bound  & & & & 1.061 & 1.070\\
Lower Bound  & & & & 0.881 & 0.894\\
\hline
\end{tabular}}
\\ \end{center}\medskip
\ \ \ \ \ \ \ \ \ \ \ \ \ \ \ \ \ {\footnotesize * BKP16 is the estimator in \cite{BKP2016}.  }
\end{table}

\begin{table}[h]
{\small \caption{\label{tb8}Exponent of cross-sectional exponent of stock returns}}
\begin{center}
{\footnotesize
\begin{tabular}{cccccccccccccccccc}
&\multicolumn{8}{c}{}\\
\hline
(N,T) & (20,60)  &  (50,80)  &  (70,100)   &  (90,110)   &  (96,125)  &  (96,100)  &  (96,80) & (96,60)\\
\hline
$\widehat\alpha$        & 0.469 & 0.799 & 0.809 & 0.813 & 0.822 & 0.843  & 0.812 & 0.833\\
$\widetilde\alpha$      & 0.502 & 0.708 & 0.901 & 0.869 & 0.842 & 0.863  & 0.823 & 0.846\\
Upper Bound & 0.609 & 0.814 & 0.968 & 0.923 & 0.902 & 0.934 & 0.942 & 0.895\\
Lower Bound & 0.395& 0.602 & 0.834 & 0.815 & 0.782& 0.792& 0.704& 0.797\\
BKP16 & 1.002 & 0.639 & 0.793 & 0.842 & 0.882 & 0.901 &   0.898 & 0.859\\
Upper Bound & 1.113 & 0.696 & 0.772 & 0.903 & 0.970 & 1.014 & 1.049 & 0.998\\
Lower Bound & 0.891& 0.582& 0.814 & 0.781 & 0.794& 0.788& 0.747 & 0.720\\
\hline
\end{tabular}}
\\ \end{center}\medskip
\ \ \ \ \ \ \ \ \ \ \ \ \ \ \ \ \ {\footnotesize *BKP16 is the estimator in \cite{BKP2016}.  }
\end{table}

}
\newpage

{\small

This material includes two appendices, i.e. Appendices A and B. Appendix A provides the proofs of Theorems 1 and 2 in the main paper. Some lemmas used in the proofs of Theorems 1 and 2 are given in Appendix B. The proof of Theorem 3 in the main paper is omitted since it is similar to that of Theorem 2.

Throughout this material, we use $C$ to denote a constant which may be different from line to line and $||\cdot||$ to denote the spectral norm or the Euclidean norm of a vector. In addition, the notation $a_n\asymp b_n$ means that $a_n=O_P(b_n)$ and $b_n=O_P(a_n)$.

\renewcommand{\theequation}{A.\arabic{equation}}
\setcounter{equation}{0}

\section{Appendix A: Proofs of the main results}

This section provides the whole derivation of (\ref{yry01}) and the proofs of Theorems 1 and 2. The proofs will use Lemmas 1 and 2, which are given Appendix B below.
For easy of presentation, we first prove Theorem 1 for the marginal estimator.

\subsection{Full derivation of (\ref{yry01})}
Consider the condition that
\begin{eqnarray*}
\frac{\partial Q^{(1)}_{NT}(\alpha,\kappa,\tau)}{\partial\kappa}=0,
\end{eqnarray*}
which is equivalent to $\sum^{[N^{\alpha}]}_{n=1}n^3\Big(\hat\sigma_n(\tau)-\kappa\Big)+[N^{2\alpha}]\sum^{N}_{n=[N^{\alpha}]+1}n\Big(\hat\sigma_n(\tau)-\frac{[N^{2\alpha}]}{n^2}\kappa\Big)=0$.

This derives
\begin{eqnarray}\label{a8}
\kappa=\kappa(\alpha):=\frac{\sum^{[N^{\alpha}]}_{n=1}n^3\hat\sigma_n(\tau)+\sum^{N}_{n=[N^{\alpha}]+1}n[N^{2\alpha}]\hat\sigma_n(\tau)}{\sum^{[N^{\alpha}]}_{n=1}n^3+\sum^{N}_{n=[N^{\alpha}]+1}\frac{[N^{4\alpha}]}{n}}.
\end{eqnarray}
Then we can obtain
\begin{eqnarray}
\kappa=\frac{\hat q_1^{(1)}(\alpha,\tau)+[N^{2\alpha}]\hat q_2^{(1)}(\alpha,\tau)}{N^{(1)}(\alpha)}.
\end{eqnarray}

We now introduce an additional expression:
\begin{eqnarray*}
Q^{(1)}(\tau)=\sum^{N}_{n=1}n^3\hat\sigma_n^2(\tau).
\end{eqnarray*}

Then
\begin{eqnarray}
Q^{(1)}_{NT}(\alpha,\kappa,\tau)&=&Q^{(1)}(\tau)+\kappa^2\sum^{[N^{\alpha}]}_{n=1}n^3+\kappa^2[N^{4\alpha}]\sum^{N}_{n=[N^{\alpha}]+1}n^{-1}-2\kappa \hat q_1^{(1)}(\alpha,\tau)-2\kappa [N^{2\alpha}]\hat q_2^{(1)}(\alpha,\tau)\non
&=&Q^{(1)}(\tau)+\kappa^2N^{(1)}(\alpha)-2\kappa\big(\hat q_1^{(1)}(\alpha,\tau)+[N^{2\alpha}]\hat q_2^{(1)}(\alpha,\tau)\big)\non
&=&Q^{(1)}(\tau)-\frac{\big(\hat q_1^{(1)}(\alpha,\tau)+[N^{2\alpha}]\hat q_2^{(1)}(\alpha,\tau)\big)^2}{N^{(1)}(\alpha)}.
\end{eqnarray}

Since $Q^{(1)}(\tau)$ does not depend on $\alpha$, minimizing $Q^{(1)}_{NT}(\alpha,\kappa)$ is equivalent to maximizing the term:
$$\hat{Q}_{NT}^{(1)}(\alpha,\tau)=\frac{(\hat q_1^{(1)}(\alpha,\tau)+[N^{2\alpha}]\hat q_2^{(1)}(\alpha,\tau))^2}{N^{(1)}(\alpha)}.$$

\subsection{Proof of Theorem \ref{thm1}}
\begin{proof}
Based on model (2.2) in the main paper, we have
\begin{eqnarray*}
\bar x_t=\frac{1}{N}\sum^{N}_{i=1}x_{it}=\mu+\boldsymbol{\bar\beta}_N^{'}\bbF_t+\bar u_t,
\end{eqnarray*}
where $\boldsymbol{\bar\beta}_N^{'}=1/N\sum^{N}_{i=1}\boldsymbol{\beta}_i$, $\mu=1/N\sum^{N}_{i=1}\mu_i$ and $\bar u_t=\frac{1}{N}\sum^{N}_{i=1}u_{it}$.
Then we have
\begin{eqnarray*}
\bar x^{(1)}=\frac{1}{T-\tau}\sum^{T-\tau}_{t=1}\bar x_t=\mu+\boldsymbol{\bar\beta}_N^{'}\bar\bbF_T+\bar u^{(1)}, \ \
\bar x^{(2)}=\frac{1}{T-\tau}\sum^{T-\tau}_{t=1}\bar x_{t+\tau}=\mu+\boldsymbol{\bar\beta}_N^{'}\bar\bbF_{T+\tau}+\bar u^{(2)},
\end{eqnarray*}
where $\bar\bbF_T=\frac{1}{T-\tau}\sum^{T-\tau}_{t=1}\bbF_t$, $\bar\bbF_{T+\tau}=\frac{1}{T-\tau}\sum^{T-\tau}_{t=1}\bbF_{t+\tau}$,
$\bar u^{(1)}=\frac{1}{T-\tau}\sum^{T-\tau}_{t=1}\bar u_t$ and $\bar u^{(2)}=\frac{1}{T-\tau}\sum^{T-\tau}_{t=1}\bar u_{t+\tau}$.

Then the auto-covariance estimator $\widehat\sigma_N(\tau)$ can be written as
\begin{eqnarray}\label{b2}
\widehat\sigma_N(\tau)&=&\frac{1}{T-\tau}\sum^{T-\tau}_{t=1}\Big(\big(\boldsymbol{\bar\beta}_N^{'}(\bbF_t-\bar\bbF_T)+\bar u_t-\bar u^{(1)}\big)\big(\boldsymbol{\bar\beta}_N^{'}(\bbF_{t+\tau}-\bar\bbF_{T+\tau})+\bar u_{t+\tau}-\bar u^{(2)}\big)\Big)\non
&=&\frac{1}{T-\tau}\sum^{T-\tau}_{t=1}\Big(\boldsymbol{\bar\beta}_N^{'}(\bbF_t-\bar\bbF_T)(\bbF_{t+\tau}-\bar\bbF_{T+\tau})^{'}\boldsymbol{\bar\beta}_N\Big)+C_N,
\end{eqnarray}
where $C_N=c_{N1}+c_{N2}+c_{N3}$ with
\begin{eqnarray*}
&&c_{N1}=\frac{1}{T-\tau}\sum^{T-\tau}_{t=1}\Big((\bar u_t-\bar u^{(1)})(\bar u_{t+\tau}-\bar u^{(2)})\Big),\\
&&c_{N2}=\frac{1}{T-\tau}\sum^{T-\tau}_{t=1}\Big(\boldsymbol{\bar\beta}_N^{'}(\bbF_t-\bar\bbF_T)(\bar u_{t+\tau}-\bar u^{(2)})\Big),\\
&&c_{N3}=\frac{1}{T-\tau}\sum^{T-\tau}_{t=1}\Big(\boldsymbol{\bar\beta}_N^{'}(\bbF_{t+\tau}-\bar\bbF_{T+\tau})(\bar u_t-\bar u^{(1)})\Big).
\end{eqnarray*}

Denote
\begin{eqnarray*}
\bbS_{\tau}=\frac{1}{T-\tau}\sum^{T-\tau}_{t=1}(\bbF_t-\bar\bbF_T)(\bbF_{t+\tau}-\bar\bbF_{T+\tau})^{'}.
\end{eqnarray*}

From (3.10) in the main paper, we can obtain
\begin{eqnarray}\label{b3}
\boldsymbol{\bar\beta}_N^{'}\bbS_{\tau}\boldsymbol{\bar\beta}_N=[N^{2\alpha_0-2}]\bar\bbv_N^{'}\bbS_{\tau}\bar\bbv_N+R_N,
\end{eqnarray}
where
\begin{eqnarray}\label{h6}
R_N=[N^{\alpha_0-2}]\bar\bbv_N^{'}\bbS_{\tau}\bbK_{\rho}+[N^{\alpha_0-2}]\bbK_\rho^{'}\bbS_{\tau}\bar\bbv_N+N^{-2}\bbK_\rho^{'}\bbS_{\tau}\bbK_\rho.
\end{eqnarray}

Therefore, from (\ref{b2}) and (\ref{b3}), we have
\begin{eqnarray}\label{e4}
\ln(\widehat\sigma_N(\tau))^2&=&\ln(\boldsymbol{\bar\beta}_N^{'}\bbS_{\tau}\boldsymbol{\bar\beta}_N)^2+\ln(1+\frac{C_N}{\boldsymbol{\bar\beta}_N^{'}\bbS_{\tau}\boldsymbol{\bar\beta}_N})^2\non
&=&4(\alpha_0-1)\ln(N)+\ln(\bar\bbv_N^{'}\bbS_{\tau}\bar\bbv_N)^2\non
&&+\ln\Big(1+\frac{R_N}{[N^{2\alpha_0-2}]\bar\bbv_N^{'}\bbS_{\tau}\bar\bbv_N}\Big)^2+\ln\Big(1+\frac{C_N}{\boldsymbol{\bar\beta}_N^{'}\bbS_{\tau}\boldsymbol{\bar\beta}_N}\Big)^2.
\end{eqnarray}

It follows from (\ref{a6}) in the main paper and (\ref{e4}) that
\begin{eqnarray}\label{b14}
&&4(\widehat\alpha_{\tau}-\alpha_0)\ln(N)+\ln(\kappa_{\tau}^2)-\ln(\bar\bbv_N^{'}\bbS_{\tau}\bar\bbv_N)^2\non
&=&\ln\left(1+\frac{R_N}{[N^{2\alpha_0-2}]\bar\bbv_N^{'}\bbS_{\tau}\bar\bbv_N}\right)^2+\ln\left(1+\frac{C_N}{\boldsymbol{\bar\beta}_N^{'}\bbS_{\tau}\boldsymbol{\bar\beta}_N}\right)^2.
\end{eqnarray}

From Lemma 1 in Appendix C, which provides the central limit theorem for $\bar\bbv_N^{'}\bbS_{\tau}\bar\bbv_N$, and condition (3.24) in the main paper, we conclude that, as $N,T\rightarrow\infty$,
\begin{eqnarray}\label{d7}
\frac{\bar\bbv_N^{'}\bbS_{\tau}\bar\bbv_N-\kappa_{\tau}}{\kappa_{\tau}}\stackrel{i.p.}{\longrightarrow} 0.
\end{eqnarray}

Evidently, $\|\bbK_{\rho}\|\leq C$. Moreover, by Assumption 2,
\begin{eqnarray}\label{r1}
E(||\bar\bbv_N||^2)=E\left(\sum^{s}_{\ell=0}\sum^{m}_{k=1}\bar v_{N\ell k}^2\right) =E\left(\sum^{s}_{\ell=0}\sum^{m}_{k=1}\frac{1}{[N^{\alpha_0}]}\sum^{[N^{\alpha_0}]}_{i,j=1}v_{i\ell k}v_{j\ell k}\right) \leq C
\end{eqnarray}
and by Assumption 3, we have
\begin{eqnarray*}
E||\bbS_{\tau}||\leq\frac{1}{T-\tau}\sum^{T-\tau}_{t=1}E||\bbF_t\bbF_{t+\tau}^{'}||
=\frac{1}{T-\tau}\sum^{T-\tau}_{t=1}E(\sum^{s}_{j=0}\bbf_{t+\tau-j}^{'}\bbf_{t-j})\leq C.
\end{eqnarray*}

So $||\bar\bbv_N||=O_P(1)$ and $||\bbS_{\tau}||=O_P(1)$. These derivations, together with (\ref{h6}), ensure
\begin{eqnarray}\label{d6}
R_N=O_P([N^{\alpha_0-2}]).
\end{eqnarray}

We conclude from (\ref{d6}) and (\ref{d7}) that
\begin{eqnarray}
\frac{R_N}{[N^{2\alpha_0-2}]\bar\bbv_N^{'}\bbS_{\tau}\bar\bbv_N}=O_P\Big(\frac{1}{[N^{\alpha_0}]\kappa_0}).
\end{eqnarray}
Therefore
\begin{eqnarray}\label{tt5}
\ln\Big(1+\frac{R_N}{[N^{2\alpha_0-2}]\bar\bbv_N^{'}\bbS_{\tau}\bar\bbv_N}\Big)^2=r_{NT}+o_P(r_{NT})=O_P\Big(\frac{1}{[N^{\alpha_0}]\kappa_0}\Big),
\end{eqnarray}
where $r_{NT}=\frac{2R_N}{[N^{2\alpha_0-2}]\bar\bbv_N^{'}\bbS_{\tau}\bar\bbv_N}
+\Big(\frac{R_N}{[N^{2\alpha_0-2}]\bar\bbv_N^{'}\bbS_{\tau}\bar\bbv_N}\Big)^2$, and we have used the simple fact that
\be
\lim\limits_{x\rightarrow 0}\frac{\ln (1+x)-x}{x}=0.
\label{g4}
\ee

It follows that
\begin{eqnarray}\label{t4}
\sqrt{\min\big([N^{\alpha_0}], T-\tau\big)}\ln\Big(1+\frac{R_N}{[N^{2\alpha_0-2}]\bar\bbv_N^{'}\bbS_{\tau}\bar\bbv_N}\Big)^2
=O_P\Big(\frac{1}{[N^{\alpha_0/2}]\kappa_0}\Big)=o_P(1).
\end{eqnarray}

Meanwhile, based on the decomposition of $C_N=\sum^{3}_{i=1}c_{Ni}$, we evaluate the orders of the following terms: $\frac{c_{Ni}}{\boldsymbol{\bar\beta}_N^{'}\bbS_{\tau}\boldsymbol{\bar\beta}_N}$ for $i=1,2,3$.

For $c_{N1}$, we need to evaluate the orders of $\bar u^{(i)}$, $i=1,2$ and $\frac{1}{T-\tau}\sum^{T-\tau}_{t=1}\bar u_t\bar u_{t+\tau}$. The order of $\frac{1}{T-\tau}\sum^{T-\tau}_{t=1}\bar u_t\bar u_{t+\tau}$ will be provided in Lemma 2 in Appendix C.

By Assumption 1, we have

\begin{eqnarray}\label{b6}
&&E\Big(\sum^{N}_{i_1,i_2=1}\sum^{T-\tau}_{t_1,t_2=1}u_{i_1t_1}u_{i_2t_2}\Big)\non
&=&E\sum^{N}_{i_1,i_2=1}\sum^{T-\tau}_{t_1,t_2=1}\Big(\sum^{+\infty}_{j_1=0}\phi_{i_1j_1}\sum^{+\infty}_{s_1=-\infty}\xi_{j_1s_1}\nu_{j_1,t_1-s_1}\Big)
\Big(\sum^{+\infty}_{j_2=0}\phi_{i_2j_2}\sum^{+\infty}_{s_2=-\infty}\xi_{j_2s_2}\nu_{j_2,t_2-s_2}\Big)\non
&=&E\sum^{N}_{i_1,i_2=1}\sum^{T-\tau}_{t_1,t_2=1}\sum^{+\infty}_{j_1=0}\phi_{i_1j_1}\sum^{+\infty}_{s_1=-\infty}\xi_{j_1s_1}\nu_{j_1,t_1-s_1}^2\phi_{i_2j_1}\xi_{j_1,t_1-s_1-t_2}\non
&=&\sum^{N}_{i_1,i_2=1}\sum^{T-\tau}_{t_1,t_2=1}\sum^{+\infty}_{j_1=0}\sum^{+\infty}_{s_1=-\infty}\phi_{i_1j_1}\phi_{i_2j_1}\xi_{j_1s_1}\xi_{j_1,t_1-s_1-t_2}\non
&\leq&\sum^{T-\tau}_{t_1=1}\sum^{N}_{i_1=1}\sum^{+\infty}_{j_1=0}|\phi_{i_1j_1}|\sum^{N}_{i_2=1}|\phi_{i_2j_1}|\sum^{+\infty}_{s_1=-\infty}|\xi_{j_1s_1}|\sum^{T}_{t_2=1}|\xi_{j_1,t_1-s_1-t_2}|
=O\big(N(T-\tau)\big).
\end{eqnarray}

From (\ref{b6}) and the fact that $E(\bar u^{(1)})=0$, we have
\begin{eqnarray}
Var(\bar u^{(1)})=\frac{1}{N^2(T-\tau)^2}E\Big(\sum^{N}_{i_1,i_2=1}\sum^{T-\tau}_{t_1,t_2=1}u_{i_1t_1}u_{i_2t_2}\Big)=O\Big(\frac{1}{N(T-\tau)}\Big)
\end{eqnarray}
and then it follows that
\begin{eqnarray}\label{u9}
\bar u^{(1)}=O_P(\frac{1}{\sqrt{N(T-\tau)}}).
\end{eqnarray}

Similarly, we have $\bar u^{(2)}=O_P\left(\frac{1}{\sqrt{N(T-\tau)}}\right)$. Combining (\ref{u9}) and Lemma 2 in Appendix D, we get
\begin{eqnarray}\label{zhang2}
c_{N1}=O_P\Big(\max\left(\frac{\gamma_1(\tau)}{N}, \frac{1}{N\sqrt{T-\tau}}\right)\Big).
\end{eqnarray}

This, together with (\ref{d7}), (\ref{b3}) and (\ref{d6}), implies that
\begin{eqnarray}\label{b12}
\frac{c_{N1}}{\boldsymbol{\bar\beta}_N^{'}\bbS_{\tau}\boldsymbol{\bar\beta}_N}=O_P\Big(\max\left(\frac{\gamma_1(\tau)N^{1-2\alpha_0}}{\kappa_{\tau}}, \frac{(T-\tau)^{-1/2}N^{1-2\alpha_0}}{\kappa_{\tau}}\right)\Big).
\end{eqnarray}

We then prove
\begin{eqnarray}\label{b5}
\frac{c_{N2}}{\boldsymbol{\bar\beta}_N^{'}\bbS_{\tau}\boldsymbol{\bar\beta}_N}=O_P\Big(\frac{(T-\tau)^{-1/2}N^{1/2-\alpha_0}}{\kappa_{\tau}}\Big).
\end{eqnarray}

By Assumption 3, we have $E[c_{N2}]=0$ and then its variance
\begin{eqnarray*}
&&Var\Big[\frac{1}{T-\tau}\sum^{T-\tau}_{t=1}\Big(\boldsymbol{\bar\beta}_N^{'}(\bbF_t-\bar\bbF_T)(\bar u_{t+\tau}-\bar u^{(2)})\Big)\Big]\non
&=&\frac{1}{(T-\tau)^2}\sum^{T-\tau}_{t_1,t_2=1}E\Big(\boldsymbol{\bar\beta}_N^{'}(\bbF_{t_1}-\bar\bbF_T)\boldsymbol{\bar\beta}_N^{'}(\bbF_{t_2}-\bar\bbF_T)\Big)
E\Big((\bar u_{t_1+\tau}-\bar u^{(2)})(\bar u_{t_2+\tau}-\bar u^{(2)})\Big)\non
&=&O\Big(\frac{[N^{2\alpha_0-2}]}{N(T-\tau)}\Big),
\end{eqnarray*}
where the last equality uses (\ref{b6}) and the fact that via (3.10) in the main paper and (\ref{r1}):
\begin{eqnarray*}
E\Big(\boldsymbol{\bar\beta}_N^{'}(\bbF_{t}-\bar\bbF_T)\Big)^2
\leq \Big[[N^{2\alpha_0-2}]E\big(||\bar \bbv_N||^2\big)+n^{-2}\|\bbK_\rho\|\Big]E\big(||\bbF_t-\bar\bbF_T||^2\big)
=O([N^{2\alpha_0-2}]).
\end{eqnarray*}

Hence
\begin{eqnarray}
\frac{1}{T-\tau}\sum^{T-\tau}_{t=1}\Big(\boldsymbol{\bar\beta}_N^{'}(\bbF_t-\bar\bbF_T)(\bar u_{t+\tau}-\bar u)\Big)
=O_P\Big(\frac{[N^{\alpha_0-1}]}{(T-\tau)^{1/2}N^{1/2}}\Big).
\end{eqnarray}

In view of this, (\ref{d7}), (\ref{b3}) and (\ref{d6}), we can obtain (\ref{b5}). Similarly, one may obtain
\begin{eqnarray}\label{b13}
\frac{c_{N3}}{\boldsymbol{\bar\beta}_N^{'}\bbS_{\tau}\boldsymbol{\bar\beta}_N}=O_P\Big(\frac{(T-\tau)^{-1/2}N^{1/2-\alpha_0}}{\kappa_{\tau}}\Big).
\end{eqnarray}



By (\ref{b14}), we have
\begin{eqnarray}\label{u10}
\kappa_{\tau}^2N^{4(\widehat\alpha_{\tau}-\alpha_0)}=(\bar\bbv_N^{'}\bbS_{\tau}\bar\bbv_N)^2\Big(1+\frac{R_N}{[N^{2\alpha_0-2}]\bar\bbv_N^{'}\bbS_{\tau}\bar\bbv_N}\Big)^2
\Big(1+\frac{C_N}{\boldsymbol{\bar\beta}^{'}_N\bbS_{\tau}\boldsymbol{\bar\beta}_N}\Big)^2.
\end{eqnarray}

From (\ref{d7}), (\ref{tt5}) and (\ref{b5})-(\ref{u10}), it follows that
\begin{eqnarray}\label{u11}
&&\frac{\kappa_{\tau}N^{4(\widehat\alpha_{\tau}-\alpha_0)}-\kappa_{\tau}}
{2(1+o_p(1))}\non
&=&\frac{\kappa_{\tau}^2N^{4(\widehat\alpha_{\tau}-\alpha_0)}-\kappa_{\tau}^2}
{\bar\bbv_N^{'}\bbS_{\tau}\bar\bbv_N+\kappa_{\tau}}\non
&=&\Big(\bar\bbv_N^{'}\bbS_{\tau}\bar\bbv_N-\kappa_{\tau}\Big)
\Big(1+\frac{R_N}{[N^{2\alpha_0-2}]\bar\bbv_N^{'}\bbS_{\tau}\bar\bbv_N}\Big)^2
\Big(1+\frac{C_N}{\boldsymbol{\bar\beta}^{'}_N\bbS_{\tau}\boldsymbol{\bar\beta}_N}\Big)^2\non
&&+\frac{\kappa_{\tau}^2}{\bar\bbv_N^{'}\bbS_{\tau}\bar\bbv_N+\kappa_{\tau}}
\Big[\Big(1+\frac{R_N}{[N^{2\alpha_0-2}]\bar\bbv_N^{'}\bbS_{\tau}\bar\bbv_N}\Big)^2
\Big(1+\frac{C_N}{\boldsymbol{\bar\beta}^{'}_N\bbS_{\tau}\boldsymbol{\bar\beta}_N}\Big)^2-1\Big]\non
&=&\Big(\bar\bbv_N^{'}\bbS_{\tau}\bar\bbv_N-\kappa_{\tau}\Big)\Big(1+O_P(\frac{1}{[N^{\alpha_0}]\kappa_{\tau}})+O_P(\frac{(T-\tau)^{-1/2}N^{1/2-\alpha_0}}{\kappa_{\tau}})\non
&&+O_P(\frac{\gamma_1(\tau)N^{1-2\alpha_0}}{\kappa_{\tau}})+O_P( \frac{(T-\tau)^{-1/2}N^{1-2\alpha_0}}{\kappa_{\tau}})\Big)\non
&&+\kappa_{\tau}\Big(O_P(\frac{1}{[N^{\alpha_0}]\kappa_{\tau}})+O_P(\frac{(T-\tau)^{-1/2}N^{1/2-\alpha_0}}{\kappa_{\tau}})\non
&&+O_P(\frac{\gamma_1(\tau)N^{1-2\alpha_0}}{\kappa_{\tau}})+O_P( \frac{(T-\tau)^{-1/2}N^{1-2\alpha_0}}{\kappa_{\tau}})\Big).
\end{eqnarray}

With  Lemma 1 in Appendix C, we obtain (\ref{jiti2})-(\ref{jiti2b}).

\end{proof}

\subsection{Proof of Theorem \ref{yyr001}}
\begin{proof}
Recall that
\begin{eqnarray*}
\widetilde{\alpha}_{\tau}=\arg\max_{\alpha}\widehat{Q}_{NT}^{(1)}(\alpha,{\tau}), \ \ where \ \ \widehat{Q}_{NT}^{(1)}(\alpha,{\tau})=\frac{\big(\widehat q_1^{(1)}(\alpha,{\tau})+[N^{2\alpha}]\widehat q_2^{(1)}(\alpha,{\tau})\big)^2}{N^{(1)}(\alpha)},
\end{eqnarray*}
and
\begin{eqnarray*}
\widehat q^{(1)}_1(\alpha,{\tau})=\sum^{[N^{\alpha}]}_{n=1}n^3\widehat\sigma_n(\tau), \ \ \widehat q_2^{(1)}(\alpha,{\tau})=\sum^{N}_{n=[N^{2\alpha}]+1}n\widehat\sigma_n(\tau), \ \
N^{(1)}(\alpha)=\sum^{[N^{\alpha}]}_{n=1}n^3+\sum^{N}_{n=[N^{\alpha}]+1}\frac{[N^{4\alpha}]}{n}
\end{eqnarray*}
with
\begin{eqnarray}
\widehat\sigma_n(\tau)=\frac{1}{T-\tau}\sum^{T-\tau}_{t=1}(\bar x_{nt}-\bar x_n^{(1)})(\bar x_{n,t+\tau}-\bar x_n^{(2)}).
\end{eqnarray}
Define
\begin{eqnarray}\label{zhang1}
\sigma_n=\Big\{\begin{array}{cc}
           \kappa_{\tau}, & n\leq [N^{\alpha_0}] \\
           \kappa_{\tau}\frac{[N^{2\alpha_0}]}{n^2}, & n>[N^{\alpha_0}],
         \end{array}
\end{eqnarray}
and
\begin{eqnarray}\label{zhang1tidai}
\check{\sigma}_n=\Big\{\begin{array}{cc}
           \bar\bbv_n^{'}\bbS_{\tau}\bar\bbv_n, & n\leq [N^{\alpha_0}] \\
           \frac{[N^{2\alpha_0}]}{n^2}\bar\bbv_{N}^{'}\bbS_{\tau}\bar\bbv_N\frac{[N^{2\alpha_0}]}{n^2}, & n>[N^{\alpha_0}],
         \end{array}
\end{eqnarray}

It is easy to see that the true value $\alpha_0$ satisfies $\alpha_0=\arg\max_{\alpha}Q_N^{(1)}(\alpha,{\tau})$, where $Q_N^{(1)}(\alpha,{\tau})=\frac{\big(q_1^{(1)}(\alpha,{\tau})+[N^{\alpha}]q^{(1)}_2(\alpha,{\tau})\big)^2}{N^{(1)}(\alpha)}$, and $q^{(1)}_1(\alpha,{\tau})$ and $q^{(1)}_2(\alpha,{\tau})$ are respectively obtained from $\widehat q^{(1)}_1(\alpha,{\tau})$ and $\widehat q_2^{(1)}(\alpha,{\tau})$ with $\widehat\sigma_n(\tau)$ replaced by $\sigma_n$.

Similarly, there exists  $\check{\alpha}_{\tau}$ satisfies $\check{\alpha}_{\tau}=\arg\max_{\alpha}\check{Q}_N^{(1)}(\alpha,{\tau})$, where $\check{Q}_N^{(1)}(\alpha,{\tau})=\frac{\big(\check{q}_1^{(1)}(\alpha,{\tau})+[N^{\alpha}]\check{q}^{(1)}_2(\alpha,{\tau})\big)^2}{N^{(1)}(\alpha)}$, and $\check{q}^{(1)}_1(\alpha,{\tau})$ and $\check{q}^{(1)}_2(\alpha,{\tau})$ are respectively obtained from $\widehat q^{(1)}_1(\alpha,{\tau})$ and $\widehat q_2^{(1)}(\alpha,{\tau})$ with $\widehat\sigma_n(\tau)$ replaced by $\check{\sigma}_n$.

We first prove the CLT for $(\check{\alpha}_{\tau},\check{\kappa}_{\tau})$ where $\check{\kappa}_{\tau}=\frac{\check{q}_1^{(1)}(\check{\alpha}_{\tau},{\tau})+[N^{2\check{\alpha}_{\tau}}]\breve{q}^{(1)}_2(\check{\alpha}_{\tau},{\tau})}{N^{(1)}(\check{\alpha}_{\tau})}.$ And then we consider the difference between $(\check{\alpha}_{\tau},\check{\kappa}_{\tau})$ and $(\widetilde{\alpha}_{\tau},\widetilde{\kappa}_{\tau})$.

 It follows that
\begin{eqnarray}\label{yry6}
&&\left|\check{Q}_{NT}^{(1)}(\alpha,{\tau})-\check{Q}_{NT}^{(1)}(\alpha_0,{\tau})\right|
\non
&=& \left|\check{Q}_{NT}^{(1)}(\alpha,{\tau})-Q_{NT}^{(1)}(\alpha,{\tau})-[\check{Q}_{NT}^{(1)}(\alpha_0,{\tau})-Q_{NT}^{(1)}(\alpha_0,{\tau})]+Q_{NT}^{(1)}(\alpha,{\tau})-Q_{NT}^{(1)}(\alpha_0,{\tau})\right|
\non
&\leq&2\max_{\alpha}\big|\check{Q}_{NT}^{(1)}(\alpha,{\tau})-Q_{NT}^{(1)}(\alpha,{\tau})\big|+ \left|Q_{NT}^{(1)}(\alpha,{\tau})-Q_{NT}^{(1)}(\alpha_0,{\tau})\right|.
\end{eqnarray}

We next evaluate the two terms on the right hand of (\ref{yry6}). Consider the first term on the right hand of (\ref{yry6}). Rewrite it as
\begin{eqnarray*}
&N^{(1)}(\alpha)\Big(\check{Q}_{NT}^{(1)}(\alpha,{\tau})-Q_{N}^{(1)}(\alpha,{\tau})\Big)\non
&=\Big(\check q_1^{(1)}(\alpha,{\tau})-q^{(1)}_1(\alpha,{\tau})+[N^{2\alpha}]\big(\check q_2^{(1)}(\alpha,{\tau})-q_2^{(1)}(\alpha,{\tau})\big)\Big)\nonumber \\
&\cdot\Big(\check q_1^{(1)}(\alpha,{\tau})+q_1^{(1)}(\alpha,{\tau})+[N^{2\alpha}]\big(\check q_2^{(1)}(\alpha,{\tau})+q_2^{(1)}(\alpha,{\tau})\big)\Big).
\end{eqnarray*}

A direct calculation which is similar to (\ref{zhang2})-(\ref{u10}), together with Lemma 1 in Appendix C, yields
\begin{eqnarray}\label{05n1}
\check\sigma_n(\tau)-\sigma_n
=\big\{\begin{array}{ccc}
  \bar\bbv_n^{'}\bbS_{\tau}\bar\bbv_n-\kappa_{\tau}= O_P(v_{nT}^{-1/2}), & n\leq[N^{\alpha_0}];\\
  \frac{[N^{2\alpha_0}]}{n^2}(\bar\bbv_{N}^{'}\bbS_{\tau}\bar\bbv_N-\kappa_{\tau})
  =O_P\big(\frac{[N^{2\alpha_0}]}{n^2}v_{NT}^{-1/2}\big), & n>[N^{\alpha_0}],
\end{array}
\end{eqnarray}
where $v_{nT}=\min(n,T-\tau)$ for $n\leq [N^{\alpha_0}]$, $v_{NT}=\min([N^{\alpha_0}],T-\tau)$

It follows that 
\begin{eqnarray*}
&&\check q_1^{(1)}(\alpha,{\tau})-q_1^{(1)}(\alpha,{\tau})=\sum^{[N^{\alpha}]}_{n=1}n^3\big(\check\sigma_n(\tau)-\sigma_n\big)\non
&=&\Big\{\begin{array}{cc}
   O_P\big(\sum^{[N^{\alpha}]}_{n=1}n^3v_{nT}^{-1/2}\big), & \alpha\leq\alpha_0; \\
   O_P\big(\sum^{[N^{\alpha_0}]}_{n=1}n^3v_{nT}^{-1/2}+\sum^{[N^{\alpha}]}_{n=[N^{\alpha_0}]+1}n^3\frac{[N^{2\alpha_0}]}{n^2}v_{NT}^{-1/2}\big), & \alpha>\alpha_0,
 \end{array}\non
&=&\Big\{\begin{array}{cc}
   O_P\big([N^{4\alpha}](v_{NT}^{(1)})^{-1/2}\big), & \alpha\leq\alpha_0; \\
   O_P\big([N^{4\alpha_0}]v_{NT}^{-1/2}+[N^{2\alpha_0}]\cdot|[N^{2\alpha}]-[N^{2\alpha_0}]|v_{NT}^{-1/2}\big), & \alpha>\alpha_0,
 \end{array}
\end{eqnarray*}
where $v_{NT}^{(1)}=\min([N^{\alpha}], T-\tau)$. Similarly, we have
\begin{eqnarray*}
&&[N^{2\alpha}]\big(\check q_2^{(1)}(\alpha,{\tau})-q_2^{(1)}(\alpha,{\tau})\big)
=[N^{2\alpha}]\sum^{N}_{n=[N^{\alpha}]+1}n\big(\check\sigma_n(\tau)-\sigma_n\big)\non
&=&\Big\{\begin{array}{cc}
   \ O_P\Big([N^{2\alpha_0+2\alpha}]v_{NT}^{-1/2}-[N^{4\alpha}](v_{NT}^{(1)})^{-1/2}+[N^{2\alpha+2\alpha_0}](\log N^{1-\alpha_0})v_{NT}^{-1/2}\Big), & \alpha\leq\alpha_0;\\
   \ O_P\Big([N^{2\alpha_0+2\alpha}](\log N^{1-\alpha})v_{NT}^{-1/2}\Big), & \alpha>\alpha_0.
   \end{array}
\end{eqnarray*}

It also follows from  (\ref{05n1}) that
\begin{eqnarray*}
&&\check q_1^{(1)}(\alpha,{\tau})+q_1^{(1)}(\alpha,{\tau})
=\sum^{[N^{\alpha}]}_{n=1}n^3\big(\check\sigma_n(\tau)+\sigma_n\big)\non
&=&\Big\{\begin{array}{cc}
    O_P\Big(\sum^{[N^{\alpha}]}_{n=1}n^3\big(\kappa_{\tau}+v_{nT}^{-1/2}\big)\Big), & \alpha\leq\alpha_0;\\
  O_P\Big(\sum^{[N^{\alpha_0}]}_{n=1}n^3\big(v_{NT}^{-1/2}+\kappa_{\tau}\big)+\sum^{[N^{\alpha}]}_{n=[N^{\alpha_0}]+1}n^3\frac{[N^{2\alpha_0}]}{n^2}v_{NT}^{-1/2}\Big), & \alpha>\alpha_0,
\end{array}\non
&=&\Big\{\begin{array}{cc}
    O_P\Big([N^{4\alpha}]\kappa_{\tau}+[N^{4\alpha}](v_{NT}^{(1)})^{-1/2}\Big), & \alpha\leq\alpha_0;\\
    O_P\Big([N^{4\alpha_0}](\kappa_{\tau}+v_{NT}^{-1/2})+([N^{2\alpha+2\alpha_0}]-[N^{4\alpha_0}])v_{NT}^{-1/2}\Big), & \alpha>\alpha_0
\end{array}
\end{eqnarray*}
and that
\begin{eqnarray*}
&&[N^{2\alpha}]\big(\check q_2^{(1)}(\alpha,{\tau})+q_2^{(1)}(\alpha,{\tau})\big)
=[N^{2\alpha}]\sum^{N}_{n=[N^{\alpha}]+1}n\big(\check\sigma_n(\tau)+\sigma_n\big)\non
&=&\Big\{\begin{array}{cc}
   O_P\Big(-[N^{4\alpha}]\big((v_{NT}^{(1)})^{-1/2}+\kappa_{\tau}\big)+[N^{2\alpha+2\alpha_0}]\big(1+\log N^{1-\alpha_0}\big)\big(v_{NT}^{-1/2}+\kappa_{\tau}\big)\Big), & \alpha\leq\alpha_0;\\
   O_P\Big([N^{2\alpha_0+2\alpha}]\big(\log N^{1-\alpha}\big)\big(v_{NT}^{-1/2}+\kappa_{\tau}\big)\Big), & \alpha>\alpha_0.
\end{array}
\end{eqnarray*}
Moreover,
\begin{equation}\label{05f2}
N^{(1)}(\alpha)=\sum^{[N^{\alpha}]}_{n=1}n^3+\sum^{N}_{n=[N^{\alpha}]+1}\frac{[N^{4\alpha}]}{n}\asymp\Big([N^{4\alpha}]+[N^{4\alpha}]\log (\frac{N}{[N^\alpha]})\Big).
\end{equation}

Summarizing the above derivations implies
\begin{eqnarray}\label{yry4}
\check{Q}^{(1)}_{NT}(\alpha)-Q^{(1)}_N(\alpha)=O_P\Big([N^{4\alpha_0}]v_{NT}^{-1/2}\kappa_{\tau}\log N^{1-\alpha}\Big).
\end{eqnarray}

Consider the second term on the right hand of (\ref{yry6}). To this end, write
\begin{equation}\label{05f3}
Q^{(1)}_N(\alpha,{\tau})-Q^{(1)}_N(\alpha_0,{\tau})
=\frac{1}{N^{(1)}(\alpha)}((a_1+a_2)(a_3+a_4))+\frac{N^{(1)}(\alpha_0)-N^{(1)}(\alpha)}{N^{(1)}(\alpha)N^{(1)}(\alpha_0)}a_5^2,
\end{equation}
where
\begin{eqnarray*}
&a_1=q_1^{(1)}(\alpha,{\tau})-q_1^{(1)}(\alpha_0,{\tau}), \ \ a_2=[N^{2\alpha}]\big(q^{(1)}_2(\alpha,{\tau})-q^{(1)}_2(\alpha_0,{\tau})\big),\ \ a_3=q_1^{(1)}(\alpha,{\tau})+q_1^{(1)}(\alpha_0,{\tau}), \non
&a_4=[N^{2\alpha}]\big(q^{(1)}_2(\alpha,{\tau})+q^{(1)}_2(\alpha_0,{\tau})\big) \ \mbox{and} \ a_5=q^{(1)}_1(\alpha_0,{\tau})+[N^{2\alpha_0}]q^{(1)}_2(\alpha_0,{\tau}).
\end{eqnarray*}

Straightforward calculations indicate that
\begin{eqnarray*}
&a_1=O\big(\big|[N^{4\alpha_0}]-[N^{4\alpha}]\big|\kappa_{\tau}\big), \ \ a_2=O\big([N^{2\alpha}]\big|[N^{2\alpha_0}]-[N^{2\alpha}]\big|\kappa_{\tau}\big), \ \ a_3=O\big(([N^{4\alpha_0}]+[N^{4\alpha}])\kappa_{\tau}\big), \\
& a_4=O\big([N^{2\alpha+2\alpha_0}](\log N^{1-\alpha})\kappa_{\tau}\big) \ \mbox{and} \  a_5=O\Big(\big([N^{4\alpha_0}]+[N^{4\alpha_0}](\log N^{1-\alpha_0})\big)\kappa_{\tau}\Big).
\end{eqnarray*}

It follows from (\ref{05f2}) that
\bea
&& \Big|\frac{N^{(1)}(\alpha_0)-N^{(1)}(\alpha)}{N^{(1)}(\alpha_0)N^{(1)}(\alpha)}\Big|
\asymp\frac{\Big|([N^{4\alpha_0}]-[N^{4\alpha}])(1+\log\frac{N}{[N^{\alpha}]})-[N^{4\alpha_0}]\log[N^{\alpha_0-\alpha}]\Big|}
{([N^{4\alpha_0}]\log[N^{1-\alpha_0}])([N^{4\alpha}]\log[N^{1-\alpha}])}
\nonumber\\
&& \geq M \frac{\log N\Big|(1-\alpha)[N^{4\alpha}]-(1-\alpha_0)[N^{4\alpha_0}]\Big|}
{([N^{4\alpha_0}]\log[N^{1-\alpha_0}])([N^{4\alpha}]\log[N^{1-\alpha}])},
\nonumber
\eea
where and in what follows $M$ stands for some positive constant number which may be different values from line to line, to save notation.
From the above orders we conclude that the second term on the right hand of (\ref{05f3}) is the leading term, compared with its first term.
In view of this and the fact that $\alpha_0$ is the maximizer of $Q_N^{(1)}(\alpha,{\tau})$, we obtain from (\ref{05f3}) that
\begin{eqnarray}\label{yry5}
&& Q_N^{(1)}(\alpha_0,{\tau})-Q_N^{(1)}(\alpha,{\tau})
\geq  M\frac{\Big|N^{(1)}(\alpha_0)-N^{(1)}(\alpha)\Big|}{N^{(1)}(\alpha)N^{(1)}(\alpha_0)}a_5^2
\non
&&\geq M\frac{\log N\Big|(1-\alpha)[N^{4\alpha}]-(1-\alpha_0)[N^{4\alpha_0}]\Big|\Big([N^{4\alpha_0}]\kappa_{\tau}\Big)^2\log[N^{1-\alpha_0}]}
{[N^{4\alpha_0}](\log[N^{1-\alpha}])[N^{4\alpha}]}.
\end{eqnarray}

Note that (\ref{05n1}) holds uniformly in $\alpha$ so that (\ref{yry4}) is true when $\alpha$ is replaced with $\check{\alpha}_{\tau}$. Also (\ref{yry5}) holds when $\alpha$ is replaced with $\check{\alpha}_{\tau}$. We conclude from (\ref{yry6})
and the fact that $\check{\alpha}_{\tau}$ is the maximizer of $\check{Q}_{NT}^{(1)}(\alpha,{\tau})$ that
$$
\left|Q_N^{(1)}(\alpha_0,{\tau})-Q_N^{(1)}(\check{\alpha}_{\tau},{\tau})\right|\leq 2\max_{\alpha=\check{\alpha}_{\tau},\alpha_0}\big|\check{Q}_{NT}^{(1)}(\alpha,{\tau})-Q_{N}^{(1)}(\alpha,{\tau})\big|,
$$
which, together with (\ref{yry4}) and (\ref{yry5}), yields
\begin{eqnarray}\label{160531(1)}
\big|(1-\alpha_0)-(1-\check{\alpha}_{\tau})[N^{4(\check{\alpha}_{\tau}-\alpha_0)}]\big|=O_P\left(v_{NT}^{-1/2}\kappa_{\tau}^{-1}\right).
\end{eqnarray}
\medskip

We next consider the consistency of $\check{\kappa}_{\tau}$. It is easy to see that
\begin{eqnarray}\label{160531(30)}
\check{\kappa}_{\tau}=\frac{\check q_1^{(1)}(\check{\alpha}_{\tau},{\tau})+[N^{2\check{\alpha}_{\tau}}]\check q_2^{(1)}(\check{\alpha}_{\tau},{\tau})}{N^{(1)}(\check{\alpha}_{\tau})} \ \mbox{and} \
\kappa_{\tau}=\frac{q_1^{(1)}(\alpha_0,{\tau})+[N^{2\alpha_0}]q^{(1)}_2(\alpha_0,{\tau})}{N^{(1)}(\alpha_0)}.
\end{eqnarray}

It follows that
\begin{eqnarray}\label{yry8}
\check\kappa_{\tau}-\kappa_{\tau}=\frac{1}{N^{(1)}(\check{\alpha}_{\tau})}(b_1+b_2)+\frac{N^{(1)}(\alpha_0)-N^{(1)}(\check{\alpha}_{\tau})}{N^{(1)}(\check{\alpha}_{\tau})N^{(1)}(\alpha_0)}b_3,
\end{eqnarray}
where $b_1=\check{q}_1^{(1)}(\check{\alpha}_{\tau},{\tau})-q_1^{(1)}(\alpha_0,{\tau})$, $b_2=[N^{2\check{\alpha}_{\tau}}]\check{q}_2^{(1)}(\check{\alpha}_{\tau},{\tau})-[N^{2\alpha_0}]q_2^{(1)}(\alpha_0,{\tau})$ and $b_3=q_1^{(1)}(\alpha_0,{\tau})+[N^{2\alpha_0}]q_2^{(1)}(\alpha_0,{\tau})$.

The orders of $b_i$, $i=1,2,3$ are listed below.
\begin{eqnarray*}
&&b_1=O_P\Big(|[N^{4\alpha_0}]-[N^{4\check{\alpha}_{\tau}}]|\kappa_{\tau}+\min\{[N^{4\check{\alpha}_{\tau}}],[N^{4\alpha_0}]\}v_{NT}^{-1/2}\Big),
\nonumber\\ 
&&b_2=O_P\Big(|[N^{2\check{\alpha}_{\tau}}]-[N^{2\alpha_0}]|\cdot\big(v_{NT}^{-1/2}+\kappa_{\tau}\big)\cdot\big([N^{2\check{\alpha}_{\tau}}]+[N^{2\alpha_0}]\log N^{1-\alpha_0}\big)+[N^{4\alpha_0}](\log N^{1-\alpha_0})v_{NT}^{-1/2}\Big),
\nonumber\\
&& b_3=O_P\Big([N^{4\alpha_0}]\kappa_{\tau}+[N^{4\alpha_0}]\cdot(\log N^{1-\alpha_0})\kappa_{\tau}\Big). 
\nonumber
\end{eqnarray*}

We then conclude from these orders, (\ref{yry8}) and (\ref{05f2}) that
\begin{equation}\label{yry10b}
\check\kappa_{\tau}-\kappa_{\tau} =O_P(v_{NT}^{-1/2}).
\end{equation}

The convergence rate of $(\check{\alpha}_{\tau}, \check{\kappa_{\tau}})$ immediately follows. The next aim is to derive an asymptotic distribution for the joint estimator $(\check{\alpha}_{\tau}, \check{\kappa_{\tau}})$. In view of (\ref{160531(1)}) and (\ref{yry10b}), it is enough to consider those $\alpha$ and $\kappa$ within a compact interval $D(C)$:
\begin{eqnarray}\label{05nn2}
D(C)=\left\{(\alpha,\kappa): \quad \label{c1} \alpha=\alpha_0+\frac{1}{2}\frac{\ln(1+s_1\kappa_{\tau}^{-1}v_{NT}^{-1/2})}{\ln N},\quad
\kappa=\kappa_{\tau}+s_2v_{NT}^{-1/2}\right\},
\end{eqnarray} where $|s_j|\leq C, j=1,2$ with $C$ being some positive constant independent of $n$. Consider that
$$
\check{Q}_{NT}^{(1)}(\alpha,\kappa)=\sum^{[N^{\alpha}]}_{n=1}n^3\Big(\check{\sigma}_n(\tau)-\kappa\Big)^2+\sum^{N}_{n=[N^{\alpha}]+1}n^3\Big(\check{\sigma}_n(\tau)-\frac{[N^{2\alpha}]}{n^2}\kappa\Big)^2
$$
and $(\check{\alpha}_{\tau},\check{\kappa}_{\tau})=\arg\min_{\alpha,\kappa}\check{Q}_{NT}^{(1)}(\alpha,\kappa)$.

Without loss of generality, we assume that $\alpha\leq\alpha_0$ below. First, we simplify $\Big(\check{Q}_{NT}^{(1)}(\alpha,\kappa)-\check{Q}_{NT}^{(1)}(\alpha_0,\kappa_{\tau})\Big)$. To this end, write
\begin{eqnarray*}
&&\check{Q}_{NT}^{(1)}(\alpha,\kappa)-\check{Q}_{NT}^{(1)}(\alpha_0,\kappa_{\tau})
=\sum^{[N^{\alpha}]}_{n=1}n^3\Big(\big(\check{\sigma}_n(\tau)-\kappa\big)^2-\big(\check{\sigma}_n(\tau)-\kappa_{\tau}\big)^2\Big)\non
&&+\sum^{N}_{n=[N^{\alpha}]+1}n^3\Big(\big(\check{\sigma}_n(\tau)-\frac{[N^{2\alpha}]}{n^2}\kappa\big)^2-\big(\check{\sigma}_n(\tau)-\frac{[N^{2\alpha}]}{n^2}\kappa_{\tau}\big)^2\Big)
+\Big(\sum^{[N^{\alpha}]}_{n=1}-\sum^{[N^{\alpha_0}]}_{n=1}\Big)n^3\big(\check{\sigma}_n(\tau)-\kappa_{\tau}\big)^2\non
&&+\sum^{N}_{n=[N^{\alpha}]+1}n^3\big(\check{\sigma}_n(\tau)-\frac{[N^{2\alpha}]}{n^2}\kappa_{\tau}\big)^2-\sum^{N}_{n=[N^{\alpha_0}]+1}n^3\big(\check{\sigma}_n(\tau)-\frac{[N^{2\alpha_0}]}{n^2}\kappa_{\tau}\big)^2
=\sum^{8}_{j=1}A_j,
\end{eqnarray*}
where
$$
A_1=\sum^{[N^{\alpha}]}_{n=1}2n^3\check{\sigma}_n(\tau)(\kappa_{\tau}-\kappa);\ \
A_2=\sum^{[N^{\alpha}]}_{n=1}n^3(\kappa^2-\kappa_{\tau}^2);\
A_3=\sum^{N}_{n=[N^{\alpha}]+1}2[N^{2\alpha}]n\check{\sigma}_n(\tau)(\kappa_{\tau}-\kappa);
$$$$
A_4=\sum^{N}_{n=[N^{\alpha}]+1}\frac{[N^{4\alpha}]}{n}(\kappa^2-\kappa_{\tau}^2);
A_5=\sum^{[N^{\alpha_0}]}_{n=[N^{\alpha}]+1}-n^3\big(\check{\sigma}_n(\tau)-\kappa_{\tau}\big)^2;\ \
A_6=\sum^{[N^{\alpha_0}]}_{n=[N^{\alpha}]+1}n^3\big(\check{\sigma}_n(\tau)-\frac{[N^{2\alpha}]}{n^2}\kappa_{\tau}\big)^2;
$$$$
A_7=\sum^{N}_{n=[N^{\alpha_0}]+1}\frac{[N^{4\alpha}]-[N^{4\alpha_0}]}{n}\kappa_{\tau}^2;\ \
A_8=\sum^{N}_{n=[N^{\alpha_0}]+1}2n\kappa_{\tau}\check{\sigma}_n(\tau)([N^{2\alpha_0}]-[N^{2\alpha}]).
$$

The orders of $A_j$, $j=1,\cdots,8$, are evaluated below. 

(\ref{05n1}), together with the fact that $\alpha,\kappa \in D(C)$, implies
\begin{eqnarray*}
&&A_1=O_P([N^{4\alpha}]\kappa_{\tau}|\kappa_{\tau}-\kappa|), \ \ A_2=O([N^{4\alpha}]\kappa_{\tau}|\kappa-\kappa_{\tau}|), \non
&&A_3=O_P([N^{2\alpha+2\alpha_0}](\log N^{1-\alpha})\kappa_{\tau}|\kappa-\kappa_{\tau}|), \ \
A_4=O([N^{4\alpha}](\log N)\kappa_{\tau}|\kappa-\kappa_{\tau}|),\non
&&A_5=O_P\big(|[N^{4\alpha}]-[N^{4\alpha_0}]|v_{NT}^{-1}\big), \ \
A_6=O_P\big(|[N^{2\alpha_0}]\kappa_{\tau}-[N^{2\alpha}]\kappa|^2\log N^{\alpha_0-\alpha}\big),\non
&&A_7=O\big(|[N^{4\alpha}]-[N^{4\alpha_0}]|(\log N)\kappa_{\tau}^2\big),\ \
A_8=O_P\big(|[N^{2\alpha_0}]-[N^{2\alpha}]|\cdot[N^{2\alpha_0}](\log N^{1-\alpha_0})\kappa_{\tau}^2\big),
\end{eqnarray*}
with $v_{NT}=\min([N^{\alpha_0}], T-\tau)$.

From the above orders and (\ref{c1}), we see that $A_3, A_4, A_7$ and $A_8$ are the leading terms. We then conclude
\begin{eqnarray}\label{yyr2}
&&\check{Q}_{NT}^{(1)}(\alpha,\kappa)-\check{Q}_{NT}^{(1)}(\alpha_0,\kappa_0)=(A_3+A_{8})+(A_4+A_7)+O_P(\delta^{(1)}_{NT})\non
&=&\sum^{N}_{n=[N^{\alpha_0}]+1}2n\check{\sigma}_N(\tau)\big([N^{2\alpha_0}]\kappa_{\tau}-[N^{2\alpha}]\kappa\big)
+\sum^{N}_{n=[N^{\alpha_0}]+1}\frac{[N^{4\alpha}]\kappa^2-[N^{4\alpha_0}]\kappa_{\tau}^2}{n}+O_P(\delta^{(1)}_{NT})\non
&=&\Big([N^{2\alpha_0}]\kappa_{\tau}-[N^{2\alpha}]\kappa\Big)
\Big(\sum^{N}_{n=[N^{\alpha_0}]+1}\big(2n\check{\sigma}_n(\tau)-\frac{2[N^{2\alpha_0}]\kappa_{\tau}}{n}\big)\non
& + & \sum^{N}_{n=[N^{\alpha_0}]+1}\frac{[N^{2\alpha_0}]\kappa_{\tau}-[N^{2\alpha}]\kappa}{n}\Big)+O_P(\delta^{(1)}_{NT}),
\end{eqnarray}
where $\delta^{(1)}_{NT}=o_P(A_3+A_8+A_4+A_7)$, uniformly on the compact interval $D(C)$.

Recalling (\ref{zhang2})-(\ref{u10}), lemma \ref{lem2} and (\ref{05n1}), when $n \geq [N^{\alpha_0}]$,

Then
\begin{eqnarray}\label{yyr10}
\sum^{N}_{n=[N^{\alpha_0}]+1}2n\check{\sigma}_n(\tau)=\bar\bbv_N^{'}\bbS_{\tau}\bar\bbv_N\Big(\sum^{N}_{n=[N^{\alpha_0}]+1}\frac{2[N^{2\alpha_0}]}{n}\Big).
\end{eqnarray}

Let
\begin{eqnarray}\label{160531(5)}
g_N(s_1,s_2)=v_{NT}\frac{\check{Q}_{NT}^{(1)}(\alpha,\kappa)-\check{Q}_{NT}^{(1)}(\alpha_0,\kappa_{\tau})}{[N^{2\alpha_0}]\sum^{N}_{n=[N^{\alpha_0}]+1}\frac{2[N^{2\alpha_0}]}{n}},
\end{eqnarray}
where $s_1$ and $s_2$ are defined in (\ref{05nn2}).

By (\ref{yyr2}) and (\ref{yyr10}) we have
\begin{equation}\label{yyr11}
g_N(s_1,s_2)
=r_{NT}v_{NT}^{1/2}\Big(\bar\bbv_N^{'}\bbS_{\tau}\bar\bbv_N-\kappa_{\tau}\Big)
+\frac{1}{2}r_{NT}^2
+O_P(v_{NT}d^{(1)}_{NT}),
\end{equation}
where $r_{NT}=v_{NT}^{1/2}\frac{[N^{2\alpha}]\kappa-[N^{2\alpha_0}]\kappa_{\tau}}{[N^{2\alpha_0}]}$ and $d^{(1)}_{NT}=\frac{\delta^{(1)}_{NT}+\sum^{N}_{n=[N^{\alpha_0}]+1}2n(\frac{[N^{\alpha_0}]}{n^2})}{[N^{2\alpha_0}]\sum^{N}_{n=[N^{\alpha_0}]+1}\frac{2[N^{2\alpha_0}]}{n}}$.

With notation $\xi=\ln(1+s_1\kappa_{\tau}^{-1}v_{NT}^{-1/2})/\ln N$, we obtain
$$
\ln N^\xi=\ln(1+s_1\kappa_{\tau}^{-1}v_{NT}^{-1/2}),
$$
which implies $N^\xi=1+s_1\kappa_{\tau}^{-1}v_{NT}^{-1/2}$. This, together with (\ref{05nn2}), ensures
\begin{eqnarray}\label{20160623}
[N^{2\alpha-2\alpha_0}]=N^{\xi}=1+s_1\kappa_{\tau}^{-1}v_{NT}^{-1/2}, \ \ \ \ \kappa-\kappa_{\tau}=s_2v_{NT}^{-1/2}.
\end{eqnarray}
By (\ref{20160623}) and the definition of $r_{NT}$, we have
\begin{eqnarray}\label{05n4}
r_{NT}
&=&v_{NT}^{1/2}\frac{[N^{2\alpha}]\kappa-[N^{2\alpha}]\kappa_{\tau}}{[N^{2\alpha_0}]}
+v_{NT}^{1/2}\frac{[N^{2\alpha}]\kappa_{\tau}-[N^{2\alpha_0}]\kappa_{\tau}}{[N^{2\alpha_0}]}\non
&=&v_{NT}^{1/2}[N^{2\alpha-2\alpha_0}](\kappa-\kappa_{\tau})
+v_{NT}^{1/2}\kappa_{\tau}([N^{2\alpha-2\alpha_0}]-1)\non
&=&s_1+s_2+s_1s_2\kappa_{\tau}^{-1}v_{NT}^{-1/2}.
\end{eqnarray}

Here we would like to point out that the last term of (\ref{yyr11}) converges to zero in probability uniformly in $s_1,s_2\in [-C,C]$, in view of (\ref{yyr2}) and the tightness in $s_1$ and $s_2$ is straightforward due to the structure of $r_{NT}$ in (\ref{05n4}).

Let $\check{s}_1$ and $\check{s}_2$ be $s_1$ and $s_2$ respectively with $(\alpha, \kappa)$ replaced by $(\check{\alpha}_{\tau}, \check{\kappa}_{\tau})$. By the definition of $(\check{\alpha}_{\tau}, \check{\kappa}_{\tau})$ in (3.9) of the main paper, we know that $g_N(s_1, s_2)$ takes the minimum value at $(\check{s}_1, \check{s}_2)$. 
Moreover, from (\ref{yyr11}) and (\ref{05n4}) a key observation is that
\begin{eqnarray}\label{160531(10)}
\check{s}_1+\check{s}_2=-v_{NT}^{1/2}\left(\bar\bbv_N^{'}\bbS_{\tau}\bar\bbv_N-\kappa_{\tau}\right)+o_p(1)
\end{eqnarray}
(one can verify this by taking derivative with respective to $s_1$ and $s_2$ in (\ref{yyr11})).


Next, we analyze $\check{s}_2$. Recall that $\check{s}_2=v_{NT}^{1/2}\left(\check{\kappa}_{\tau}-\kappa_{\tau}\right)$.
By the definition of $\check{\kappa}_{\tau}$ in (3.7) of the main paper, we first provide the leading term of $\check{\kappa}_{\tau}$.
It is easy to see that the leading terms of the numerator and the denominator of $\check{\kappa}_{\tau}$ are $[N^{2\check{\alpha}_{\tau}}]\check{q}_2^{(1)}(\check{\alpha}_{\tau})$ and $\sum^{N}_{n=[N^{\check{\alpha}_{\tau}}]+1}\frac{[N^{4\check{\alpha}_{\tau}}]}{n}$ respectively.

Recalling (\ref{yry8}), we have the following evaluations:
\begin{eqnarray}\label{160531(20)zhang}
b_1= \check{q}_1^{(1)}(\check{\alpha}_{\tau})-q_1^{(1)}(\alpha_0)=O_P\Big(|[N^{4\alpha_0}]-[N^{4\check{\alpha}_{\tau}}]|\kappa_{\tau}+\min\{[N^{4\check{\alpha}_{\tau}}],[N^{4\alpha_0}]\}v_{NT}^{-1/2}\Big),
\end{eqnarray}

\begin{eqnarray}\label{160531(20)zhang4}
b_3=O_P\Big([N^{4\alpha_0}]\kappa_{\tau}+[N^{4\alpha_0}]\cdot(\log N^{1-\alpha_0})\kappa_{\tau}\Big),
\end{eqnarray}

\begin{eqnarray}\label{160531(20)zhang2}
b_2=[N^{2\check{\alpha}}_{\tau}]\check{q}_2^{(1)}(\check{\alpha}_{\tau})-[N^{2\alpha_0}]\check{q}_2^{(1)}(\alpha_0)+[N^{2\alpha_0}]\check{q}_2^{(1)}(\alpha_0)-[N^{2\alpha_0}]q_2^{(1)}(\alpha_0),
\end{eqnarray}

\begin{eqnarray}\label{160531(20)zhang3}
&&[N^{2\alpha_0}]\check{q}_2^{(1)}(\alpha_0)-[N^{2\alpha_0}]q_2^{(1)}(\alpha_0)\non
&=&\sum^{N}_{n=[N^{\alpha_0}]+1}n\cdot[N^{2\alpha_0}]\cdot\left(\check{\sigma}_{n}(\tau)-\sigma_n \right)\non
&=&\left(\bar{\bbv}_N^{'}\bbS_{\tau}\bar{\bbv}_N-\kappa_{\tau}\right)\sum^{N}_{n=[N^{\alpha_0}]+1}\frac{[N^{4\alpha_0}]}{n}
\end{eqnarray}

\begin{eqnarray}\label{160531(20)}
&&[N^{2\check{\alpha}_{\tau}}]\check{q}_2^{(1)}(\check{\alpha}_{\tau})-[N^{2\alpha_0}]\check{q}_2^{(1)}(\alpha_0)\non
&=&\left([N^{2\check{\alpha}_{\tau}}]-[N^{2\alpha_0}]\right)\check{q}_2^{(1)}\left(\check{\alpha}_{\tau}\right)
+[N^{2\alpha_0}]\left(\check{q}_2^{(1)}\left(\check{\alpha}_{\tau}\right)-\check{q}_2^{(1)}(\alpha_0)\right)\non
&=&\left([N^{2\check{\alpha}_{\tau}}]-[N^{2\alpha_0}]\right)\cdot\sum^{N}_{n=[N^{\alpha_0}]+1}n\check{\sigma}_n(\tau)\non
&&+\left([N^{2\check{\alpha}_{\tau}}]-[N^{2\alpha_0}]\right)\cdot\sum^{N}_{n=[N^{\alpha_0}]+1}n\check{\sigma}_n(\tau)
+[N^{2\alpha_0}]\sum^{[N^{\alpha_0}]}_{n=[N^{\check{\alpha}_{\tau}}]+1}n\check{\sigma}_n(\tau)\non
&=&\left([N^{2\check{\alpha}_{\tau}}]-[N^{2\alpha_0}]\right)\cdot\left(\sum^{N}_{n=[N^{\alpha_0}]+1}n\check{\sigma}_n(\tau)\right)
\cdot\left(1+o_P(1)\right)\non
&=&O_P\Big(\kappa_{\tau}|[N^{2\check{\alpha}_{\tau}}]-[N^{2\alpha_0}]|N^{2\alpha_0}\log N \Big)
\end{eqnarray}
and
\begin{eqnarray}\label{160531(21)}
&&[N^{4\widetilde{\alpha}}]\cdot\sum^{N}_{n=[N^{\widetilde{\alpha}}]+1}\frac{1}{n}
-[N^{4\alpha_0}]\cdot\sum^{N}_{n=[N^{\alpha_0}]+1}\frac{1}{n}\non
&=&[N^{4\widetilde{\alpha}}]\cdot\sum^{[N^{\alpha_0}]}_{n=[N^{\widetilde{\alpha}}]+1}\frac{1}{n}
+\left([N^{4\widetilde{\alpha}}]-[N^{4\alpha_0}]\right)\cdot\sum^{N}_{n=[N^{\alpha_0}]+1}\frac{1}{n}\non
&=&\left([N^{4\widetilde{\alpha}}]-[N^{4\alpha_0}]\right)\cdot\left(\ln\frac{N}{[N^{\alpha_0}]}\right)\left(1+o(1)\right).
\end{eqnarray}

It follows from (\ref{160531(20)zhang})-(\ref{160531(21)}) and (\ref{160531(30)})-(\ref{yry8})that
\begin{eqnarray}\label{160526(1)}
v_{NT}^{1/2}(\check{\kappa}_{\tau}-\kappa_{\tau})
&=&v_{NT}^{1/2}\left(\bar{\bbv}_N^{'}\bbS_{\tau}\bar{\bbv}_N-\kappa_{\tau}\right)+o_p(1).
\end{eqnarray}

Then we can get the CLT of $(\check{\alpha}_{\tau},\check{\kappa}_{\tau})$ with (\ref{160531(10)}),(\ref{160526(1)}) and Lemma 1.

\begin{eqnarray}\label{yr40zhongtu}
\left(
\begin{array}{c}
\check{\kappa}_{\tau}v_{NT}^{1/2}(N^{2(\check{\alpha}_{\tau}-\alpha_0)}-1)\\
v_{NT}^{1/2}(\check{\kappa}_{\tau}-\kappa_{\tau})\\
\end{array}
\right)
\stackrel{d}{\longrightarrow} \mathcal{N}\left(\left(
\begin{array}{c}
0\\
0\\
\end{array}
\right), \left(
\begin{array}{cc}
4\sigma_{\tau}^2 & -2\sigma_{\tau}^2\\
-2\sigma_{\tau}^2 & \sigma_{\tau}^2
\end{array}
\right)\right),
\end{eqnarray}
where $\kappa_{\tau}$ is defined in (\ref{160530}) and  $\sigma_{\tau}^2$ is defined in (\ref{163031}).

The difference between $(\check{\alpha}_{\tau},\check{\kappa}_{\tau})$ and $(\widetilde{\alpha}_{\tau},\widetilde{\kappa}_{\tau})$ can be obtained by the similar idea and the fact
\begin{eqnarray}\label{05n1cn}
\hat{\sigma}_n-\check{\sigma}_n=C_n
=\big\{\begin{array}{ccc}
  O_P(\frac{\gamma_1(\tau)}{n})+O_P(\frac{1}{n^{1/2}(T-\tau)^{1/2}}), & n\leq[N^{\alpha_0}];\\
   O_P(\frac{\gamma_1(\tau)}{n})+O_P(\frac{1}{n(T-\tau)^{1/2}})+O_P(\frac{[N^{\alpha_0}]}{n^{3/2}(T-\tau)^{1/2}}), & n>[N^{\alpha_0}].
\end{array}
\end{eqnarray}
Then we complete the proof.

\end{proof}

\subsection{Proof of Proposition \ref{yyr001fujiaguji}}
\begin{proof}
Since $\widetilde{\alpha}_{\tau}$ is a consistent estimator of $\alpha_0$, we consider $n=[N^{\alpha_0}]$. Write
\begin{eqnarray}\label{a7gujisigma03a}
\widehat\sigma_{i,T}(\tau)&=&\frac{1}{(T-\tau)}\sum_{t=1}^{T-\tau}(\mathbf{\beta}_i'(\mathbf{F_t}-\mathbf{\bar{F}_T})+u_{it}-\bar{u}_{iT})\non
&&((\mathbf{F_{t+\tau}}-\mathbf{\bar{F}_{T+\tau}})'\mathbf{\bar{\beta}}_n+\bar{u}_{n,t+\tau}-\bar{u}_{n,T+\tau})\non
&=&\mathbf{\beta}_i'\mathbf{S}_{\tau}\mathbf{\bar{\beta}}_n+\tilde{c}_{i,T}.
\end{eqnarray}
The main part of $\widehat\sigma_{i,T}(\tau)$ is $\mathbf{\beta}_i'\mathbf{S}_{\tau}\mathbf{\bar{\beta}}_n$. It follows that

\begin{eqnarray}\label{a7gujisigma04a}
\widehat\sigma_{i,T}(\tau)-\widehat\sigma_{T}(\tau)=(\mathbf{\beta}_i-\mathbf{\bar{\beta}}_n)'\mathbf{S}_{\tau}\mathbf{\bar{\beta}}_n+\tilde{c}_{i,T}-\frac{1}{n}\sum_{i=1}^n\tilde{c}_{i,T}.
\end{eqnarray}

We can conclude that
\begin{eqnarray}\label{a7gujisigma04a4}
\sum_{i=1}^{n}(\widehat\sigma_{i,T}(\tau)-\widehat\sigma_{T}(\tau))^2&=&\sum_{i=1}^{n}((\mathbf{\beta}_i-\mathbf{\bar{\beta}}_n)'\mathbf{S}_{\tau}\mathbf{\bar{\beta}}_n)^2(1+o_p(1))\non
&=&(n-1)\boldsymbol{\mu}_v^{'}\boldsymbol{\Sigma}_{\tau}^{'}\boldsymbol{\Sigma}_v\boldsymbol{\Sigma}_{\tau}\boldsymbol{\mu}_v(1+o_p(1)).
\end{eqnarray}
This implies  (\ref{a7gujisigma05}).

Note that
\begin{eqnarray}\label{a7gujisigma02aquan}
(T-\tau)E(\widehat\sigma_{n}(\tau)-\mathbf{\bar{\beta}}_n'\mathbf{\Sigma}_{\tau}\mathbf{\bar{\beta}}_n)^2=(\boldsymbol{\mu}_v^{'}\otimes\boldsymbol{\mu}_v^{'})\boldsymbol{\Omega}(\boldsymbol{\mu}_v\otimes\boldsymbol{\mu}_v)(1+o_p(1)).
\end{eqnarray}
Then
\begin{eqnarray}\label{a7gujisigma02a1quan}
\frac{1}{T-\tau}E(\sum_{t=1}^{T-\tau}(\widehat\sigma_{n,t}(\tau)-\mathbf{\bar{\beta}}_n'\mathbf{\Sigma}_{\tau}\mathbf{\bar{\beta}}_n))^2=(\boldsymbol{\mu}_v^{'}\otimes\boldsymbol{\mu}_v^{'})\boldsymbol{\Omega}(\boldsymbol{\mu}_v\otimes\boldsymbol{\mu}_v)(1+o_p(1))
\end{eqnarray}

and
\begin{eqnarray}\label{a7gujisigma02a2quan}
\hat{\sigma}_{0,T}^2&=&\frac{1}{T-\tau-1}\sum_{t=1}^{T-\tau}(\widehat\sigma_{n,t}(\tau)-\widehat\sigma_{n}(\tau))^2+\non
&&\sum_{j=1}^{l}\frac{2}{T-\tau-j}\sum_{t=1}^{T-\tau-j}(\widehat\sigma_{n,t}(\tau)-\widehat\sigma_{n}(\tau))(\widehat\sigma_{n,t+j}(\tau)-\widehat\sigma_{n}(\tau))\non
&=&(\boldsymbol{\mu}_v^{'}\otimes\boldsymbol{\mu}_v^{'})\boldsymbol{\Omega}(\boldsymbol{\mu}_v\otimes\boldsymbol{\mu}_v)(1+o_p(1))+O_P(\frac{l}{N^{\alpha-1/2}(T-\tau)^{1/2}})\non
&&+O_P(\sum_{j=l+1}^{\infty}\{\gamma(1,j)^2+\gamma(1,\tau+j)\gamma(1,|\tau-j|)\}).
\end{eqnarray}
Note that
\begin{eqnarray}\label{a7gujisigma02a3}
\sum_{j=l+1}^{\infty}\gamma(1,j)^2=O(\sum_{j=l+1}^{\infty}\sum_{i=0}^{\infty}|b_i||b_{i+j}|)=O(\sum_{i=0}^{\infty}|b_i|\sum_{k=i+l+1}^{\infty}|b_{k}|)=O(\sum_{k=i+l+1}^{\infty}|b_{k}|)
\end{eqnarray}
and
\begin{eqnarray}\label{a7gujisigma02a4}
\sum_{j=l+1}^{\infty}\gamma(1,\tau+j)\gamma(0,|\tau-j|)=O(\sum_{j=l+1}^{\infty}\sum_{i=0}^{\infty}|b_i||b_{\tau+i+j}|)=O(\sum_{k=i+\tau+l+1}^{\infty}|b_{k}|).
\end{eqnarray}
Then when $l \rightarrow \infty$, $l=o(T)$ and $l=o(N^{\alpha-1/2}(T-\tau)^{1/2})$,
\begin{eqnarray*}
\hat{\sigma}_{0,T}^2&=&\frac{1}{T-\tau-1}\sum_{t=1}^{T-\tau}(\widehat\sigma_{n,t}(\tau)-\widehat\sigma_{n}(\tau))^2+\non
&&\sum_{j=1}^{l}\frac{2}{T-\tau-j}\sum_{t=1}^{T-\tau-j}(\widehat\sigma_{N,t}(\tau)-\widehat\sigma_{n}(\tau))(\widehat\sigma_{n,t+j}(\tau)-\widehat\sigma_{n}(\tau))\non
&=&(\boldsymbol{\mu}_v^{'}\otimes\boldsymbol{\mu}_v^{'})\boldsymbol{\Omega}(\boldsymbol{\mu}_v\otimes\boldsymbol{\mu}_v)(1+o_p(1))+o_p(1).
\end{eqnarray*}
We can conclude (\ref{a7gujisigma02a2s}).
\end{proof}
\renewcommand{\theequation}{C.\arabic{equation}}
\setcounter{equation}{0}

\section{Appendix B: Some lemmas}

In this appendix, we provide the necessary lemmas used in the proofs of the main theorems above.

\subsection{Lemmas \ref{lem1} and \ref{lem2}}
\begin{lem}\label{lem1}
In addition to Assumptions 1 and 3, we assume that
$\tau$ is fixed or $\tau$ tends to infinity satisfying
\begin{eqnarray}\label{u12a}
\frac{\tau}{(T-\tau)^{\delta/(2\delta+2)}}\rightarrow 0, \ \ as \ \ T\rightarrow\infty,
\end{eqnarray}
for some constant $\delta>0$. Moreover, under (\ref{u12}), we assume that
\begin{eqnarray}
E|\zeta_{it}|^{2+2\delta}<+\infty,
\end{eqnarray}
where $\zeta_{it}$ is the $i$-th component of $\boldsymbol{\zeta}_t$ and $\{\boldsymbol{\zeta}_t: \ldots,-1,0,1,\ldots\}$ is the sequence appeared in Assumption 3. And
the covariance matrix $\Gamma$ of the random vector
\begin{eqnarray}
\Big(C_{ij}(h^{'}): i=1,\ldots,m; j=1,\ldots,m; h^{'}=\tau-s,\ldots,\tau+s\Big)
\end{eqnarray}
is positive definite, where $C_{ij}(h)$ is defined in (\ref{w15}) just above Theorem 1 in the main paper.

Then as $N,T\rightarrow\infty$, we have
\begin{eqnarray*}
v_{NT}^{1/2}(\bar\bbv_N^{'}\bbS_{\tau}\bar\bbv_N-\boldsymbol{\mu}_v^{'}\boldsymbol{\Sigma}_{\tau}\boldsymbol{\mu}_v)\stackrel{d}\rightarrow \mathcal{N}\Big(0,\sigma_0^2\Big),
\end{eqnarray*}
where $$\sigma_0^2=\lim_{N,T\rightarrow\infty}\frac{v_{NT}}{[N^{\alpha_0}]}4\boldsymbol{\mu}_v^{'}\boldsymbol{\Sigma}_{\tau}\boldsymbol{\Sigma}_v\boldsymbol{\Sigma}_{\tau}\boldsymbol{\mu}_v
+\lim_{N,T\rightarrow\infty}\frac{v_{NT}}{T-\tau}\ (\boldsymbol{\mu}_v^{'}\otimes\boldsymbol{\mu}_v^{'})var(\sqrt{T-\tau}vec\big(\bbS_{\tau}-\boldsymbol{\Sigma}_{\tau}\big))(\boldsymbol{\mu}_v\otimes\boldsymbol{\mu}_v), $$
$\boldsymbol{\Sigma}_{\tau}=\mathbb{E}(\bbF_t\bbF_{t+\tau}^{'})$,
$\boldsymbol{\mu}_v=\mu_v\bbe_{m(s+1)}$, where $\bbe_{m(s+1)}$ is an $m(s+1)\times 1$ vector with each element being $1$, `vec' means that for a matrix $\bbX=(\bbx_1,\cdots,\bbx_n): q\times n$, $vec(\bbX)$ is the $qn\times 1$ vector defined as
\begin{eqnarray}
vec(\bbX)=\left(
            \begin{array}{c}
              \bbx_1 \\
              \vdots \\
              \bbx_n \\
            \end{array}
          \right).
\end{eqnarray}

\end{lem}

\begin{lem}\label{lem2}
Under Assumptions 1 and 3, we have
\begin{eqnarray}\label{u0}
\frac{1}{T-\tau}\sum^{T-\tau}_{t=1}\bar u_t\bar u_{t+\tau}=O_P\Big(\max\big(\frac{\gamma_1(\tau)}{N}, \frac{1}{N\sqrt{T-\tau}}\big)\Big).
\end{eqnarray}
\end{lem}

We prove the above lemmas in Appendix C below.

\section{Appendix C: \ Proofs of Lemmas \ref{lem1} and \ref{lem2}}

\begin{proof}[Proof of Lemma \ref{lem1}]
Write
\begin{eqnarray}\label{b8}
\bar\bbv_N^{'}\bbS_{\tau}\bar\bbv_N-\boldsymbol{\mu}_v^{'}\boldsymbol{\Sigma}_{\tau}\boldsymbol{\mu}_v
&=&(\bar\bbv_N^{'}-\boldsymbol{\mu}_v^{'})\bbS_{\tau}\bar\bbv_N
+\boldsymbol{\mu}_v^{'}(\bbS_{\tau}-\boldsymbol{\Sigma}_{\tau})\bar\bbv_N
+\boldsymbol{\mu}_v^{'}\boldsymbol{\Sigma}_{\tau}(\bar\bbv_N-\boldsymbol{\mu}_v)\non
&=&(\bar\bbv_N^{'}-\boldsymbol{\mu}_v^{'})(\bbS_\tau\bar\bbv_N+\boldsymbol{\Sigma}_{\tau}\boldsymbol{\mu}_v)
+\boldsymbol{\mu}_v^{'}(\bbS_{\tau}-\boldsymbol{\Sigma}_{\tau})\bar\bbv_N.
\end{eqnarray}

Since the elements of the vector $\bar\bbv_N$ are all i.i.d., we have
\begin{eqnarray}\label{b9}
\sqrt{[N^{\alpha_0}]}(\bar\bbv_N-\boldsymbol{\mu}_v)\stackrel{d}{\rightarrow}N(0,\boldsymbol{\Sigma}_v),\ \ as\ \ N\rightarrow\infty,
\end{eqnarray}
where $\boldsymbol{\Sigma}_v$ is an $m(s+1)$-dimensional diagonal matrix with each of the diagonal elements being $\sigma_v^2$.

Moreover, under Assumption 3, we have
\begin{eqnarray}\label{b15}
\bbS_{\tau}-\boldsymbol{\Sigma}_{\tau}\stackrel{i.p.}\rightarrow 0, \ \ as\ \ T\rightarrow\infty,
\end{eqnarray}
(one may see (\ref{h7}) below).
It follows from (\ref{b9}) and (\ref{b15}) that, if $\tau$ is fixed,
\begin{eqnarray}\label{b21}
&&\sqrt{[N^{\alpha_0}]}(\bar\bbv_N^{'}-\boldsymbol{\mu}_v^{'})(\bbS_{\tau}\bar v_N+\boldsymbol{\Sigma}_{\tau}\boldsymbol{\mu}_v)\non
&=&\sqrt{[N^{\alpha_0}]}\Big((\bar\bbv_N^{'}-\boldsymbol{\mu}_v^{'})\bbS_{\tau}(\bar\bbv_N-\boldsymbol{\mu}_v)
+(\bar\bbv_N^{'}-\boldsymbol{\mu}_v^{'})(\bbS_{\tau}-\boldsymbol{\Sigma}_{\tau})\boldsymbol{\mu}_v
+2(\bar\bbv_N^{'}-\boldsymbol{\mu}_v^{'})\boldsymbol{\Sigma}_{\tau}\boldsymbol{\mu}_v\Big)\non
&=&2\sqrt{[N^{\alpha_0}]}(\bar\bbv_N^{'}-\boldsymbol{\mu}_v^{'})\boldsymbol{\Sigma}_{\tau}\boldsymbol{\mu}_v+o_p(1)\stackrel{d}{\rightarrow}\mathcal{N}(0, 4\boldsymbol{\mu}_v^{'}\boldsymbol{\Sigma}_{\tau}^{'}\boldsymbol{\Sigma}_v\boldsymbol{\Sigma}_{\tau}\boldsymbol{\mu}_v).
\end{eqnarray}

When $\tau$ goes to infinity and satisfies (\ref{u12}), we have $\lim_{\tau\rightarrow\infty}\boldsymbol{\Sigma}_{\tau}=0$. In fact, we consider one element $\gamma(h)=Cov(f_{k,t}, f_{k,t+h})$ of $\boldsymbol{\Sigma}_{\tau}$:
\begin{eqnarray*}
\gamma(h)=E\Big(\sum^{+\infty}_{j_1=0}b_{j_1}\zeta_{k,t-j_1}\sum^{+\infty}_{j_2=0}b_{j_2}\zeta_{t+h-j_2}\Big)
=\sum^{+\infty}_{j_1=0}b_{j_1}b_{h+j_1}.
\end{eqnarray*}

Then
\begin{eqnarray*}
\sum^{+\infty}_{h=0}|\gamma(h)|=\sum^{+\infty}_{h=0}|\sum^{+\infty}_{j=0}b_jb_{h+j}|\leq\Big(\sum^{+\infty}_{j=0}|b_j|\Big)^2<+\infty.
\end{eqnarray*}
From this, we can see that $\gamma(h)\rightarrow 0$ as $h\rightarrow\infty$. So as $\tau\rightarrow\infty$, $\boldsymbol{\Sigma}_{\tau}\rightarrow 0$. Hence, under this case,
\begin{eqnarray}
\sqrt{[N^{\alpha_0}]}(\bar\bbv_N^{'}-\boldsymbol{\mu}_v^{'})(\bbS_{\tau}\bar v_N+\boldsymbol{\Sigma}_{\tau}\boldsymbol{\mu}_v)\stackrel{i.p.}\rightarrow 0.
\end{eqnarray}

Under Assumption 3, by Theorem 14 in Chapter 4 of \cite{H1970}, when $\tau$ is fixed, the sample covariance of the stationary time series $\{\bbf_t: t=1,2,\ldots,T\}$ has the following asymptotic property:
\begin{eqnarray}\label{h7}
\sqrt{T-\tau}\Big(vec\big(\boldsymbol{\widehat\gamma}(h)-\boldsymbol{\gamma}(h)\big), 0\leq h\leq\ell\Big)
\stackrel{d}{\rightarrow}N(0, \boldsymbol{\omega}),
\end{eqnarray}
\begin{eqnarray*}
\boldsymbol{\omega}=\lim_{N,T\rightarrow\infty}var(\sqrt{T-\tau}vec\big(\boldsymbol{\widehat\gamma}(h)-\boldsymbol{\gamma}(h)\big)).
\end{eqnarray*}

$\boldsymbol{\widehat\gamma}(h)=\frac{1}{T-h}\sum^{T-h}_{t=1}(\bbf_t-\bar\bbf^{(1)})(\bbf_{t+h}-\bar\bbf^{(2)})^{'}$, $\bar\bbf^{(1)}=\frac{1}{T-h}\sum^{T-h}_{t=1}\bbf_t$,
$\bar\bbf^{(2)}=\frac{1}{T-h}\sum^{T-h}_{t=1}\bbf_{t+h}$, and $\boldsymbol{\gamma}(h)=Cov(\bbf_t,\bbf_{t+h})$.
Note that the expression of $vec\big(\boldsymbol{\widehat\gamma}(h)\big)$ is
\begin{eqnarray*}
vec\Big(\boldsymbol{\widehat\gamma}(h)\Big)=\Big(\widetilde{cov}(1,1), \widetilde{cov}(2,1),\ldots,\widetilde{cov}(m,1),\ldots,\widetilde{cov}(1,m),\ldots,\widetilde{cov}(m,m)\Big)^{'},
\end{eqnarray*}
with $\widetilde{cov}(i,j)=\frac{1}{T-h}\sum^{T-h}_{t=1}f_{it}f_{j,t+h}-\frac{1}{T-h}\sum^{T-h}_{t=1}f_{it}\frac{1}{T-h}\sum^{T-h}_{t=1}f_{j,t+h}$.

Here we would like to point out that although Theorem 14 of \cite{H1970} gives the CLT for the sample covariance $\boldsymbol{\check\gamma}=\frac{1}{T-h}\sum^{T-h}_{t=1}f_{it}f_{j,t+h}$, the asymptotic distribution of $\boldsymbol{\widehat\gamma}$ is the same as that of $\boldsymbol{\check\gamma}$ (one can verify it along similar lines).

The CLT in Theorem 14 of \cite{H1970} is provided for finite lags $h$ and $r$ only. If both $h$ and $r$ tend to infinity as $T\rightarrow\infty$, we develop a corresponding CLT in Lemma \ref{lem9}. 

Moreover note that the expansion of $vec\Big(\bbS_{\tau}-\boldsymbol{\Sigma}_{\tau}\Big)$ has a form of
\begin{eqnarray*}
vec\Big(\bbS_{\tau}-\boldsymbol{\Sigma}_{\tau}\Big)=\left(
                                                      \begin{array}{c}
                                                        vec\Big(\boldsymbol{\widehat\gamma}(\tau)-\boldsymbol{\gamma}(\tau)\Big) \\
                                                        \vdots\\
                                                        vec\Big(\boldsymbol{\widehat\gamma}(\tau+s)-\boldsymbol{\gamma}(\tau+s)\Big) \\
                                                        vec\Big(\boldsymbol{\widehat\gamma}(\tau-1)-\boldsymbol{\gamma}(\tau-1)\Big) \\
                                                        \vdots \\
                                                        vec\Big(\boldsymbol{\widehat\gamma}(\tau+s-1)-\boldsymbol{\gamma}(\tau+s-1)\Big) \\
                                                        \vdots \\
                                                        vec\Big(\boldsymbol{\widehat\gamma}(\tau-s)-\boldsymbol{\gamma}(\tau-s)\Big) \\
                                                        \vdots \\
                                                        vec\Big(\boldsymbol{\widehat\gamma}(\tau)-\boldsymbol{\gamma}(\tau)\Big) \\
                                                      \end{array}
                                                    \right).
\end{eqnarray*}

In view of this and (\ref{h7}), we conclude
\begin{eqnarray}\label{b10}
\sqrt{T-\tau}\Big(vec\big(\bbS_{\tau}-\boldsymbol{\Sigma}_{\tau}\big)\Big)\stackrel{d}{\rightarrow}N(0,\boldsymbol{\Omega}),
\end{eqnarray}
where
\begin{eqnarray*}
\boldsymbol{\Omega}=\lim_{N,T\rightarrow\infty}var(\sqrt{T-\tau}vec\big(\bbS_{\tau}-\boldsymbol{\Sigma}_{\tau})).
\end{eqnarray*}

By (\ref{b10}) and Slutsky's theorem, we have, as $N,T\rightarrow\infty$,
\begin{eqnarray}\label{b11}
&&\sqrt{T-\tau}\boldsymbol{\mu}_v^{'}(\bbS_{\tau}-\boldsymbol{\Sigma}_{\tau})\bar\bbv_N
=\boldsymbol{\mu}_v^{'}\sqrt{T-\tau}(\bbS_{\tau}-\boldsymbol{\Sigma}_{\tau})(\bar\bbv_N-\boldsymbol{\mu}_v)
+\boldsymbol{\mu}_v^{'}\sqrt{T-\tau}(\bbS_{\tau}-\boldsymbol{\Sigma}_{\tau})\boldsymbol{\mu}_v\non
&=&\big((\bar\bbv_N^{'}-\boldsymbol{\mu}_v^{'})\otimes\boldsymbol{\mu}_v^{'}\big)\sqrt{T-\tau}vec(\bbS_{\tau}-\boldsymbol{\Sigma}_{\tau})
+(\bar\bbv_N^{'}\otimes\boldsymbol{\mu}_v^{'})\sqrt{T-\tau}vec(\bbS_{\tau}-\boldsymbol{\Sigma}_{\tau})\non
&=&(\bar\bbv_N^{'}\otimes\boldsymbol{\mu}_v^{'})\sqrt{T-\tau}vec(\bbS_{\tau}-\boldsymbol{\Sigma}_{\tau})+o_p(1)\stackrel{d}{\rightarrow}\mathcal{N}\Big(0, (\boldsymbol{\mu}_v^{'}\otimes\boldsymbol{\mu}_v^{'})\boldsymbol{\Omega}(\boldsymbol{\mu}_v\otimes\boldsymbol{\mu}_v)\Big),
\end{eqnarray}
where the first equality uses ${\rm vec}(\bbA\bbX\bbB)=(\bbB^{'}\otimes\bbA) {\rm vec}(\bbX)$, with $\bbA: p\times m$, $\bbB: n\times q$ and $\bbX: m\times n$ being three matrices; and $\otimes$ denoting the Kronecker product;
and the last asymptotic distribution uses the fact that
\begin{eqnarray}
\bar\bbv_N\stackrel{i.p.}{\rightarrow}\boldsymbol{\mu}_v,
\end{eqnarray}
which can be verified as in (\ref{b9}).

By (\ref{b21}), (\ref{b11}) and the independence between $\bbS_{\tau}$ and $\bar\bbv_N$, we have
\begin{eqnarray*}
&&\sqrt{\min([N^{\alpha_0}],T-\tau)}(\bar\bbv_N^{'}\bbS_{\tau}\bar\bbv_N-\boldsymbol{\mu}_v^{'}\boldsymbol{\Sigma}_{\tau}\boldsymbol{\mu}_v)\non
&=&\sqrt{\min([N^{\alpha_0}],T-\tau)}(\bar\bbv_N^{'}-\boldsymbol{\mu}_v^{'})(\bbS_{\tau}\bar\bbv_N+\boldsymbol{\Sigma}_{\tau}\boldsymbol{\mu}_v)
+\sqrt{\min([N^{\alpha_0}],T-\tau)}\boldsymbol{\mu}_v^{'}(\bbS_{\tau}-\boldsymbol{\Sigma}_{\tau})\bar\bbv_N\non
&&\stackrel{d}\rightarrow \mathcal{N}\Big(0,\sigma_0^2\Big),
\end{eqnarray*}
where the last step uses the fact that
\begin{eqnarray*}
[N^{\alpha_0}](\bar\bbv_N^{'}-\boldsymbol{\mu}_v^{'})\bbS_{\tau}\bar\bbv_N=[N^{\alpha_0}](\bar\bbv_N^{'}-\boldsymbol{\mu}_v^{'})\boldsymbol{\Sigma}_{\tau}\boldsymbol{\mu}_v+o_P(1).
\end{eqnarray*}

\end{proof}


\begin{proof}[Proof of Lemma \ref{lem2}]
First, we calculate the order of
\begin{eqnarray}
E\Big(\frac{1}{T-\tau}\sum^{T-\tau}_{t=1}\bar u_t\bar u_{t+\tau}\Big)^2.
\end{eqnarray}
From Assumption 1, it follows that
\begin{eqnarray}\label{u1}
&&E\Big(\frac{1}{T-\tau}\sum^{T-\tau}_{t=1}\bar u_t\bar u_{t+\tau}\Big)^2
=\frac{1}{N^4(T-\tau)^2}\sum^{T-\tau}_{t_1,t_2=1}\sum^{N}_{i_1,\ldots,i_4=1}E(u_{i_1t_1}u_{i_2,t_1+\tau}u_{i_3t_2}u_{i_4,t_2+\tau})\non
&=&\frac{1}{N^4(T-\tau)^2}\sum^{T-\tau}_{t_1,t_2=1}\sum^{N}_{i_1,\ldots,i_4=1}E\Big(\sum^{+\infty}_{j_1=0}\phi_{i_1j_1}\sum^{+\infty}_{s_1=-\infty}\xi_{j_1s_1}\nu_{j_1,t_1-s_1}
\sum^{+\infty}_{j_2=0}\phi_{i_2j_2}\sum^{+\infty}_{s_2=-\infty}\xi_{j_2s_2}\nu_{j_2,t_1+\tau-s_2}\non
& \times & \sum^{+\infty}_{j_3=0}\phi_{i_3j_3}\sum^{+\infty}_{s_3=-\infty}\xi_{j_3s_3}\nu_{j_3,t_2-s_3}
\sum^{+\infty}_{j_4=0}\phi_{i_4j_4}\sum^{+\infty}_{s_4=-\infty}\xi_{j_4s_4}\nu_{j_4,t_2+\tau-s_4}\Big).
\end{eqnarray}

Note that there are four random terms appearing in the expectation in (\ref{u1}), i.e. $\nu_{j_1,t_1-s_1}$, $\nu_{j_2,t_1+\tau-s_2}$, $\nu_{j_3,t_2-s_3}$, $\nu_{j_4,t_2+\tau-s_4}$. By Assumption 1, the expectation is not zero only if these four random terms are pairwise equivalent or all of them are equivalent. In view of this, we have
\begin{eqnarray}\label{u3}
E\Big(\frac{1}{T-\tau}\sum^{T-\tau}_{t=1}\bar u_t\bar u_{t+\tau}\Big)^2=\Phi_1+\Phi_2+\Phi_3+\Phi_4,
\end{eqnarray}
where
\begin{eqnarray}\label{u4}
\Phi_1&=&\frac{1}{N^4(T-\tau)^2}\sum^{T-\tau}_{t_1,t_2=1}\sum^{N}_{i_1,\ldots,i_4=1}
E\Big(\sum^{+\infty}_{j_1=0}\phi_{i_1j_1}\phi_{i_2j_1}\sum^{+\infty}_{s_1=-\infty}\xi_{j_1s_1}\xi_{j_1,s_1+\tau}\nu_{j_1,t_1-s_1}^2\Big)\non
& \times & E\Big(\sum^{+\infty}_{j_3\neq j_1}\phi_{i_3j_3}\phi_{i_4j_3}\sum^{+\infty}_{s_3=-\infty}\xi_{j_3s_3}\xi_{j_3,s_3+\tau}\nu_{j_3,t_2-s_3}^2\Big)\non
&=&\frac{1}{N^4(T-\tau)^2}\sum^{T-\tau}_{t_1,t_2=1}\sum^{N}_{i_1,\ldots,i_4=1}E(u_{i_1t_1}u_{i_2,t_1+\tau})E(u_{i_3t_2}u_{i_4,t_2+\tau})\non
&=&\frac{1}{N^4(T-\tau)^2}\sum^{T-\tau}_{t_1,t_2=1}\sum^{N}_{i_1,\ldots,i_4=1}\gamma_1(\tau)\gamma_2(|i_1-i_2|)\gamma_1(\tau)\gamma_2(|i_3-i_4|) = O\Big(\frac{\gamma_1^2(\tau)}{N^2}\Big),
\end{eqnarray}
where the first equality uses $\nu_{j_1,t_1-s_1}=\nu_{j_2,t_1+\tau-s_2}$ and $\nu_{j_3,t_2-s_3}=\nu_{j_4,t_2+\tau-s_4}$. The last equality uses
(2.6) in the main paper.

For $\Phi_2$,
\begin{eqnarray}\label{u5}
\Phi_2&=&\frac{1}{N^4(T-\tau)^2}\sum^{T-\tau}_{t_1,t_2=1}\sum^{N}_{i_1,\ldots,i_4=1}
E\Big(\sum^{+\infty}_{j_1=0}\phi_{i_1j_1}\phi_{i_3j_1}\sum^{+\infty}_{s_1=-\infty}\xi_{j_1s_1}\xi_{j_1,t_2-t_1+s_1}\nu_{j_1,t_1-s_1}^2\Big)\non
& \times & E\Big(\sum^{+\infty}_{j_2\neq j_1}\phi_{i_2j_2}\phi_{i_4j_2}\sum^{+\infty}_{s_2=-\infty}\xi_{j_2s_2}\xi_{j_2,t_2-t_1+s_2}\nu_{j_2,t_1+\tau-s_2}\Big)\non
&\leq&\frac{K}{N^4(T-\tau)^2}\sum^{T-\tau}_{t_2=1}\sum^{N}_{i_1,i_4=1}\sum^{+\infty}_{j_1=0}|\phi_{i_1j_1}|\sum^{N}_{i_3=1}|\phi_{i_3j_1}|
\sum^{+\infty}_{s_1=-\infty}|\xi_{j_1s_1}|\sum^{T-\tau}_{t_1=1}|\xi_{j_1,t_2-t_1+s_1}|\non
& \times & \sum^{N}_{i_2=1}|\phi_{i_2j_2}|\sum^{+\infty}_{j_2\neq j_1}|\phi_{i_4j_2}|\sum^{+\infty}_{s_2=-\infty}|\xi_{j_2s_2}| = O\Big(\frac{1}{N^2(T-\tau)}\Big),
\end{eqnarray}
where the first equality uses $\nu_{j_1,t_1-s_1}=\nu_{j_3,t_2-s_3}$ and $\nu_{j_2,t_1+\tau-s_3}=\nu_{j_4,t_2+\tau-s_4}$. The last equality uses
(2.4) in the main paper.

Similarly, for $\Phi_3$, we have
\begin{eqnarray}\label{u6}
\Phi_3&=&\frac{1}{N^4(T-\tau)^2}\sum^{T-\tau}_{t_1,t_2=1}\sum^{N}_{i_1,\ldots,i_4=1}
E\Big(\sum^{+\infty}_{j_1=0}\phi_{i_1j_1}\phi_{i_4j_1}\sum^{+\infty}_{s_1=-\infty}\xi_{j_1s_1}\xi_{j_1,t_2-t_1+\tau-s_1}\nu^2_{j_1,t_1-s_1}\Big)\non
& \times & \ E\Big(\sum^{+\infty}_{j_2\neq j_1}\phi_{i_2j_2}\phi_{i_3j_2}\sum^{+\infty}_{s_2=-\infty}\xi_{j_2s_2}\xi_{j_2,t_2-t_1-\tau+s_2}\nu_{j_2,t_1+\tau-s_2}^2\Big)\non
&\leq&\frac{K}{N^4(T-\tau)^2}\sum^{T-\tau}_{t_2=1}\sum^{N}_{i_1,i_2=1}\sum^{+\infty}_{j_1=0}|\phi_{i_1j_1}|\sum^{N}_{i_4=1}|\phi_{i_4j_1}|
\sum^{+\infty}_{s_1=-\infty}|\xi_{j_1s_1}|\sum^{T-\tau}_{t_1=1}|\xi_{j_1,t_2-t_1+\tau-s_1}|\non
& \times & \sum^{+\infty}_{j_2\neq j_1}|\phi_{i_2j_2}|\sum^{N}_{i_3=1}|\phi_{i_3j_2}|\sum^{+\infty}_{s_2=-\infty}|\xi_{j_2s_2}| = O\Big(\frac{1}{N^2(T-\tau)}\Big),
\end{eqnarray}
where the first equality uses $\nu_{j_1,t_1-s_1}=\nu_{j_4,t_2+\tau-s_4}$ and $\nu_{j_2,t_1+\tau-s_2}=\nu_{j_3,t_2-s_3}$. The last equality uses (2.4) in the main paper.

For $\Phi_4$,
\begin{eqnarray}\label{u7}
&& \Phi_4 = \frac{1}{N^4(T-\tau)^2}\sum^{T-\tau}_{t_1,t_2=1}\sum^{N}_{i_1,\ldots,i_4=1}
E\Big(\sum^{+\infty}_{j_1=0}\phi_{i_1j_1}\phi_{i_3j_1}\sum^{+\infty}_{s_1=-\infty}\xi_{j_1s_1}\xi_{j_1,t_2-t_1+s_1}\nu_{j_1,t_1-s_1}^4
\nonumber\\
&& \times \ \phi_{i_2j_1}\phi_{i_4j_1}\xi_{j_1,\tau+s_1}\xi_{j_1,t_2-t_1+\tau+s_1}\Big)
\nonumber\\
&&\leq \ \frac{K}{N^4(T-\tau)^2}\sum^{T-\tau}_{t_2=1}\sum^{N}_{i_1=1}
\sum^{+\infty}_{j_1=0}|\phi_{i_1j_1}|\sum^{N}_{i_3=1}|\phi_{i_3j_1}|\sum^{+\infty}_{s_1=-\infty}|\xi_{j_1s_1}|\sum^{T-\tau}_{t_1=1}|\xi_{j_1,t_2-t_1+s_1}|
\sum^{N}_{i_2=1}|\phi_{i_2j_1}|\sum^{N}_{i_4=1}|\phi_{i_4j_1}|
\nonumber\\
&& = \ O\Big(\frac{1}{N^3(T-\tau)}\Big),
\end{eqnarray}
where the first equality uses $\nu_{j_1,t_1-s_1}=\nu_{j_2,t_1+\tau-s_2}=\nu_{j_3,t_2-s_3}=\nu_{j_4,t_2+\tau-s_4}$ and the last equality uses (2.4) in the main paper.

Hence by (\ref{u3}), (\ref{u4}), (\ref{u5}), (\ref{u6}) and (\ref{u7}), we have
\begin{eqnarray}
E\Big(\frac{1}{T-\tau}\sum^{T-\tau}_{t=1}\bar u_t\bar u_{t+\tau}\Big)^2=O\Big(\max\left(\frac{\gamma^2_1(\tau)}{N^2}, \frac{1}{N^2(T-\tau)}\right)\Big).
\end{eqnarray}

Moreover,
\begin{eqnarray}
&&E\Big(\frac{1}{T-\tau}\sum^{T-\tau}_{t=1}\bar u_t\bar u_{t+\tau}\Big)
=E\Big(\frac{1}{(T-\tau)N^2}\sum^{T-\tau}_{t=1}\sum^{N}_{i,j=1}u_{it}u_{j,t+\tau}\Big)\non
&=&\frac{1}{(T-\tau)N^2}\sum^{T-\tau}_{t=1}\sum^{N}_{i,j=1}\gamma_1(\tau)\gamma_2(|i-j|)
=O\Big(\frac{\gamma_1(\tau)}{N}\Big).
\end{eqnarray}

Therefore, we have
\begin{eqnarray}\label{u8}
Var\Big(\frac{1}{T-\tau}\sum^{T-\tau}_{t=1}\bar u_t\bar u_{t+\tau}\Big)=O\Big(\max\left(\frac{\gamma^2_1(\tau)}{N^2}, \frac{1}{N^2(T-\tau)}\right)\Big).
\end{eqnarray}

By (\ref{u8}), we have proved (\ref{u0}).
\end{proof}
{
\setcounter{lem}{2}

\noindent{\bf \large Two lemmas for Lemma \ref{lem1}}
\medskip

This section is to generalize Theorem 8.4.2 of \cite{Ader1994} to the case where the time lag tends to infinity along with the sample size. To this end, we first list a crucial lemma below.

\begin{lem}[Theorem 2.1 of \cite{RW2000}]\label{lem8}

Let $\{X_{n,i}\}$ be a triangular array of mean zero random variables. For each $n=1,2,\ldots$, let $d=d_n$, $m^{'}=m_n$, and suppose $X_{n,1}, \ldots, X_{n,d}$ is an $m^{'}$-dependent sequence of random variables. Define $B^2_{n,\ell,a}\equiv Var\Big(\sum^{a+\ell-1}_{i=a}X_{n,i}\Big)$ and $B^2_n\equiv B^2_{n,d,1}\equiv Var\Big(\sum^{d}_{i=1}X_{n,i}\Big)$.

Let the following conditions hold. For some $\delta>0$ and some $-1\leq\gamma<1$:
\begin{eqnarray}
&&E|X_{n,i}|^{2+\delta}\leq\Delta_n \ \ for \ all \ i; \label{l1}\\
&&B^2_{n,\ell,a}/(\ell^{1+\gamma})\leq K_n \ \ for \ all \ a \ and \ for \ all \ k\geq m^{'}; \label{l2}\\
&&B^2_n/(d(m^{'})^{\gamma})\geq L_n; \label{l3}\\
&&K_n/L_n=O(1); \label{l4}\\
&&\Delta_n/L_n^{(2+\delta)/2}=O(1); \label{l5}\\
&&(m^{'})^{1+(1-\gamma)(1+2/\delta)}/d\rightarrow 0. \label{l6}
\end{eqnarray}
Then
\begin{eqnarray}
B^{-1}_n(X_{n,1}+\cdots+X_{n,d})\Rightarrow\mathcal{N}(0,1).
\end{eqnarray}
\end{lem}

We are now ready to state the following generalization.

\begin{lem}\label{lem9}
Let $\bbf_t=\sum^{+\infty}_{r=0}b_r\boldsymbol{\zeta}_{t-r}$ where $\boldsymbol{\zeta}_t=(\zeta_{1t}, \ldots, \zeta_{mt})$, consisting of i.i.d components with zero mean and unit variance, is an i.i.d sequence of $m$-dimensional random vector.
Assume that for some constant $\delta>0$, $E|\zeta_{it}|^{2+2\delta}<+\infty$; and the coefficients $\{b_r: r=0,1,2,\ldots\}$ satisfy $\sum^{+\infty}_{r=0}|b_r|<+\infty$.
Moreover, we assume that
\begin{eqnarray}\label{w13}
h=o\Big((T-h)^{\delta/(2\delta+2)}\Big)
\end{eqnarray}
and that
the covariance matrix $\Gamma$ of the random vector
\begin{eqnarray}\label{w7}
\Big(C_{ij}(h^{'}): i=1,\ldots,m; j=1,\ldots,m; h^{'}=h-s,\ldots,h+s\Big)
\end{eqnarray}
is positive definite, where $C_{ij}(h^{'})$ is defined in (3.21).

Then, for any fixed positive constants $s$ and $m$,
\begin{eqnarray}
\Big(\sqrt{T-h^{'}}\big(C_{ij}(h^{'})-\sigma_{ij}(h^{'})\big): 1\leq i,j\leq m; h-s\leq h^{'}\leq h+s\Big)
\end{eqnarray}
converges in distribution to a normal distribution with mean $0$ and covariances
\begin{eqnarray}
\Big(\lim_{T\rightarrow\infty}(T-h)Cov\big(C_{i_1j_1}(h_1), C_{i_2j_2}(h_2)\big): 1\leq i_1,i_2,j_1,j_2\leq m; h-s\leq h_1,h_2\leq h+s\Big).
\end{eqnarray}
\end{lem}

\begin{proof}  For $1\leq i,j\leq m$ and $0\leq h\leq T-1$, write $f_{i,t,k}=\sum^{k}_{s^{'}=0}b_{s^{'}}\zeta_{i,t-s^{'}}$, $C_{ij}(h,k)=\frac{1}{T-h}\sum^{T-h}_{t=1}f_{i,t,k}f_{j,t+h,k}\\ =\frac{1}{T-h}\sum^{T-h}_{t=1}\sum^{k}_{s_1,s_2=0}b_{s_1}b_{s_2}\zeta_{i,t-s_1}\zeta_{j,t+h-s_2}$, and
\begin{eqnarray}\label{w0}
&&\sigma_{ij}(h,k)=E(f_{i,t,k}f_{j,t+h,k})\non
&=&\sum^{k}_{s_1,s_2=0}b_{s_1}b_{s_2}E(\zeta_{i,t-s_1}\zeta_{j,t+h-s_2})
=\left\{
   \begin{array}{ll}
     0, & i\neq j; \\
     \sum^{k-h}_{s_1=0}b_{s_1}b_{h+s_1}, & i=j; h=0,1,\ldots,k; \\
     0, & i=j; h=k+1,k+2.
   \end{array}
 \right.
\end{eqnarray}

The proof of this lemma is similar to that of Theorem 8.4.2 of \cite{Ader1994} and it can be divided into two steps:
\medskip

{\bf Step 1}: \ For any fixed k, the first step is to provide the asymptotic theorem for
\begin{eqnarray}
\Big(\sqrt{T-h^{'}}\big(C_{ij}(h^{'},k)-\sigma_{ij}(h^{'},k)\big): 1\leq i,j\leq m; h-s\leq h^{'}\leq h+s\Big);
\end{eqnarray}

{\bf Step 2}: \ The second step is to prove that for any $1\leq i,j\leq m$, in probability,
\begin{eqnarray}
\lim_{T\rightarrow\infty}\sqrt{T-h}\Big(C_{ij}(h^{'})-C_{ij}(h^{'},k)\Big)=0.
\end{eqnarray}

The second step can be verified as in Theorem 8.4.2 of \cite{Ader1994} (i.e. page 479-page 481) and the details are omitted here.


Consider Step 1 now. Let
\begin{eqnarray}
X_{T-h,t}(i,j)=\frac{1}{\sqrt{T-h}}\Big(f_{i,t,k}f_{j,t+h,k}-\sigma_{ij}(h,k)\Big), \ \ 1\leq i, j\leq m,
\end{eqnarray}
so that
\begin{eqnarray}
\sqrt{T-h}\Big(C_{ij}(h,k)-\sigma(h,k)\Big)=\sum^{T-h}_{t=1}X_{T-h,t}(i,j).
\end{eqnarray}

By simple calculations, we see that $f_{i,t,k}f_{j,t+h,k}$ is independent of $f_{i,g,k}f_{j,g+h,k}$ if $t$ and $g$ differ by more than $k+h$ when $i\neq j$ and differ by more than $k$ when $i=j$. So $\{f_{i,t,k}f_{j,t+h,k}: t=1,\ldots,T-h\}$ is a $(k+h)$ or $k$ dependent covariance stationary process with mean $\sigma_{ij}(h,k)$ and covariance
\begin{eqnarray}\label{w1}
&&Cov(f_{i,t,k}f_{j,t+h,k}, f_{i,g,k}f_{j,g+h,k})\non
&=&\sum^{k}_{s_1,\ldots,s_4=0}b_{s_1}b_{s_2}b_{s_3}b_{s_4}E(\zeta_{i,t-s_1}\zeta_{j,t+h-s_2}\zeta_{i,g-s_3}\zeta_{j,g+h-s_4})-\sigma_{ij}^2(h,k)\non
&=&\left\{
     \begin{array}{ll}
       A_1, & i\neq j; \\
       \sum^{4}_{q=1}A_q-\sigma_{ii}^2(h,k), & i=j,
     \end{array}
   \right.
\end{eqnarray}
where
\begin{eqnarray*}
&&A_1=\sum^{k}_{s_1=0}\sum^{k}_{s_2=0}b_{s_1}b_{s_2}b_{g-t+s_1}b_{g-t+s_2}, \ \ A_2=\sum^{k}_{b_1=0}\sum^{k}_{b_3=0}b_{s_1}b_{h+s_1}b_{s_3}b_{h+s_3},\non
&&A_3=\sum^{k}_{s_1=0}\sum^{k}_{s_3=0}b_{s_1}b_{t-g+h+s_3}b_{s_3}b_{g-t+h+s_1}, \ \ A_4=-2\sum^{k}_{s_1=0}b_{s_1}b_{h+s_1}b_{g-t+s_1}b_{g-t+h+s_1},
\end{eqnarray*}
where (\ref{w1}) uses the fact that $E(\zeta_{i,t-s_1}\zeta_{j,t+h-s_2}\zeta_{i,g-s_3}\zeta_{j,g+h-s_4})$ is not equal to zero if and only if the four terms $\zeta_{i,t-s_1}, \zeta_{j,t+h-s_2}, \zeta_{i,g-s_3}, \zeta_{j,g+h-s_4}$ are pairwise equivalent or they are all equivalent. 

Hence for any $1\leq i, j\leq m; h-s\leq h^{'}\leq h+s$, $\{X_{T-h^{'},t}(i,j): t=1,\ldots,T-h^{'}\}$ is a $(k+h^{'})$ or $k$ dependent covariance stationary process. This implies that any linear combination of the process $\{\sum^{m}_{i,j=1}\sum^{h+s}_{h^{'}=h-s}a_{i,j,h^{'}}X_{T-h^{'},t}(i,j): t=1,\ldots,T-h-s\}$ is a $(k+h+s)$ dependent covariance stationary process.
Thus, we need to check Conditions (\ref{l1})--(\ref{l6}) for such a linear combination of the process. 
Moreover, one should note that it is enough to justify those conditions for each stochastic process $\{X_{T-h^{'},t}(i,j): t=1,\ldots,T-h^{'}\}$, where $1\leq i, j\leq m; h-s\leq h^{'}\leq h+s$, since $s$ and $m$ are both fixed.

Observe that
\begin{eqnarray}\label{w6}
&&E\Big|X_{T-h^{'},t}(i,j)\Big|^{2+\delta}
=\Big(\frac{1}{T-h^{'}}\Big)^{(2+\delta)/2}E\Big|f_{i,t,k}f_{j,t+h,k}-\sigma_{ij}(h,k)\Big|^{2+\delta}\non
&\leq&K\Big(\frac{1}{T-h^{'}}\Big)^{(2+\delta)/2}\Big(E\Big|f_{i,t,k}f_{j,t+h,k}\Big|^{2+\delta}+\Big|\sigma_{i,j}(h,k)\Big|^{2+\delta}\Big)\non
&\leq&K\Big(\frac{1}{T-h^{'}}\Big)^{(2+\delta)/2},
\end{eqnarray}
where $K$ is a constant number, and we have also used (\ref{w0}) and the fact that
\begin{eqnarray}
&&E\Big|f_{i,t,k}f_{j,t+h,k}\Big|^{2+\delta}
=E\Big|\sum^{k}_{s_1,s_2=0}b_{s_1}b_{s_2}\zeta_{i,t-s_1}\zeta_{j,t+h-s_2}\Big|^{2+\delta}\non
&\leq&K\sum^{k}_{s_1,s_2=0}|b_{s_1}|^{2+\delta}|b_{s_2}|^{2+\delta}\Big(E|\zeta_{i,t-s_1}|^{4+2\delta}+E|\zeta_{j,t+h-s_2}|^{4+2\delta}\Big)
=O(1).
\end{eqnarray}

In view of (\ref{w6}), taking $\Delta_T=K\Big(\frac{1}{T-h^{'}}\Big)^{(2+\delta)/2}$,  we have
\begin{eqnarray}\label{w12}
E\Big|X_{T-h^{'},t}(i,j)\Big|^{2+\delta}\leq\Delta_T,
\end{eqnarray}
implying (\ref{l1}).

We obtain from (\ref{w1}) that
\begin{eqnarray}\label{w2}
&&B^2_{T,\ell,a}(i,j)\equiv Var(\sum^{a+\ell-1}_{t=a}X_{T-h^{'},t}(i,j))\non
&=&\frac{1}{T-h^{'}}\sum^{a+\ell-1}_{t=a}\sum^{a+\ell-1}_{g=a}Cov(f_{i,t,k}f_{j,t+h^{'},k}, f_{i,g,k}f_{j,g+h^{'},k})\non
&=&\left\{
     \begin{array}{ll}
       \frac{1}{T-h^{'}}\sum^{a+\ell-1}_{t=a}\sum^{a+\ell-1}_{g=a}A_1, & i\neq j; \\
       \frac{1}{T-h^{'}}\sum^{a+\ell-1}_{t=a}\sum^{a+\ell-1}_{g=a}\Big(\sum^{4}_{q=1}A_q-\sigma_{ii}^2(h^{'},k)\Big), & i=j.
     \end{array}
   \right.
\end{eqnarray}

Note that $A_2=\sigma_{ii}^2(h^{'},k)$. Below we only evaluate the remaining terms involving $A_1,A_3,A_4$. By the fact that $\sum^{+\infty}_{r=0}|b_r|<+\infty$, we have
\begin{eqnarray}\label{w3}
&&\Big|\frac{1}{T-h^{'}}\sum^{a+\ell-1}_{t=a}\sum^{a+\ell-1}_{g=a}A_1\Big|
=\Big|\frac{1}{T-h^{'}}\sum^{a+\ell-1}_{t=a}\sum^{a+\ell-1}_{g=a}\sum^{k}_{s_1=0}\sum^{k}_{s_2=0}b_{s_1}b_{s_2}b_{g-t+s_1}b_{g-t+s_2}\Big|\non
&\leq&\frac{K}{T-h^{'}}\sum^{a+\ell-1}_{t=a}\sum^{k}_{s_1=0}|b_{s_1}|\sum^{k}_{s_2=0}|b_{s_2}|\sum^{a+\ell-1}_{g=a}|b_{g-t+s_1}|=O\Big(\frac{\ell}{T-h^{'}}\Big).
\end{eqnarray}

Similarly, one may verify that
\begin{equation}\label{w5}\frac{1}{T-h^{'}}\sum^{a+\ell-1}_{t=a}\sum^{a+\ell-1}_{g=a}A_j=O\Big(\frac{\ell}{T-h^{'}}\Big), \ j=3,4.
\end{equation}

We conclude from (\ref{w2})-(\ref{w5}) that
\begin{eqnarray}\label{w9}
B^2_{T,\ell,a}(i,j)=O\Big(\frac{\ell}{T-h}\Big).
\end{eqnarray}

Taking $\ell=T-h$ in $B^{2}_{T,\ell,a}$, we have $B^2_T=O(1)$. Moreover, for any linear combination of the process $\{\sum\limits_{t}\sum^{m}_{i,j=1}\sum^{h+s}_{h^{'}=h-s}a_{i,j,h^{'}}X_{T-h^{'},t}(i,j): t=1,\ldots,T-h-s\}$, by the  assumption of $\Gamma>0$ (see (\ref{w7}), its variance $B^2_T$ is
\begin{eqnarray}\label{w8}
B^2_T=\bba^{'}\boldsymbol{\Gamma}\bba>0.
\end{eqnarray}

 In view of (\ref{w9}) and (\ref{w8}), we can take $\gamma=0$, $K_T=\widetilde{K}_1\frac{1}{T-h-s}$ and $L_T=\widetilde{K}_2\frac{1}{T-h-s}$ for the purpose of verifying (\ref{l2})--(\ref{l6}), where $\widetilde{K}_1$ and $\widetilde{K}_2$ are two constants. Then
\begin{eqnarray}\label{w11}
&&\frac{B_{T,\ell,a}^2}{\ell^{1+\gamma}}\leq K_T, \ \ for \ all \ a \ and \ all \ \ell\geq k+h;\non
&&\frac{B_T^2}{(T-h^{'})(k+h)^{\gamma}}\geq L_T.
\end{eqnarray}

Moreover, $K_T$, $L_T$ and $\Delta_T$ satisfy
\begin{eqnarray}\label{w10}
\frac{K_T}{L_T}=O(1) \ \ \mbox{and} \ \ \frac{\Delta_T}{L_T^{(2+\delta)/2}}=O(1).
\end{eqnarray}

By (\ref{w13}), we have that, for any fixed $k$,
\begin{eqnarray}\label{w14}
\frac{(k+h)^{2+2/\delta}}{T-h}\rightarrow 0, \ \ as \ T\rightarrow\infty.
\end{eqnarray}

From (\ref{w12}), (\ref{w11}), (\ref{w10}), (\ref{w14}) and Lemma \ref{lem8}, we conclude that
\begin{eqnarray*}
\Big(\sqrt{T-h^{'}}\big(C_{ij}(h^{'},k)-\sigma_{ij}(h^{'},k)\big): 1\leq i,j\leq m; h-s\leq h^{'}\leq h+s\Big)
\end{eqnarray*}
converges in distribution to a standard normal distribution with mean zero and covariances
\begin{eqnarray*}
\Big(\lim_{T\rightarrow\infty}(T-h)Cov\big(C_{i_1j_1}(h_1,k), C_{i_2j_2}(h_2,k)\big): 1\leq i_1,i_2,j_1,j_2\leq m; h-s\leq h_1,h_2\leq h+s\Big).
\end{eqnarray*}
Hence the proof of step 1 is completed.

\end{proof}
}

\end{document}